\documentclass[11pt]{article}         
\usepackage[margin=1in]{geometry} % Custom margins for single page, etc.
\usepackage{fullpage}     
\usepackage[utf8]{inputenc} % allow utf-8 input
\usepackage[T1]{fontenc}    % use 8-bit T1 fonts
\usepackage{url}            % simple URL typesetting
\usepackage{booktabs}       % professional-quality tables
\usepackage{amsfonts}       % blackboard math symbols
\usepackage{nicefrac}       % compact symbols for 1/2, etc.
\usepackage{microtype}      % microtypography
\usepackage{bm}             % bold symbols
\usepackage[colorlinks=true, citecolor=blue]{hyperref}
\usepackage{natbib}
\usepackage{algorithm}
\usepackage{algorithmic}
\usepackage{enumitem}
\usepackage{graphicx}
\usepackage{amssymb}
\usepackage{xcolor}
\usepackage{dsfont} % for identity operator
\usepackage{array}
\newcolumntype{P}[1]{>{\centering\arraybackslash}p{#1}}
\usepackage{amsmath}
\usepackage{setspace}
\usepackage{makecell}
\usepackage{caption}
\usepackage[tableposition=above]{caption} % for the caption to be placed nicely below.
\usepackage{subcaption}
\usepackage{placeins}
\usepackage[flushleft]{threeparttablex}
\usepackage[english]{babel}
\usepackage{amsthm}
\usepackage{multirow}
\usepackage{mathtools}

\theoremstyle{definition}

\newtheorem{prop}{Proposition}
\newtheorem{corollary}{Corollary}
\newtheorem{theorem}{Theorem}

\theoremstyle{definition}
\newtheorem{remark}{Remark}
\theoremstyle{definition}
\newtheorem{assumption}{Assumption}
\newtheorem{example}{Example}

\makeatletter
\renewenvironment{proof}[1][\proofname]{%
  \par\pushQED{\qed}\normalfont%
  \topsep6\p@\@plus6\p@\relax
  \trivlist\item[\hskip\labelsep\bfseries#1\@addpunct{.}]%
  \ignorespaces
}{%
  \popQED\endtrivlist\@endpefalse
}
    \DeclareMathOperator*{\argmax}{argmax} % thin space, limits underneath in displays

  \DeclareMathOperator{\tr}{\mathsf{Tr}}

\usepackage{tabularx,booktabs}
\newcolumntype{Y}{>{\centering\arraybackslash}X}

\usepackage{comment}

\usepackage[textsize=tiny]{todonotes}

\renewcommand{\texttt}[2][black]{\textcolor{#1}{\ttfamily #2}}% \texttt[<color>]{<stuff>}

\usepackage{enumitem} % Nice listing options in itemize and enumerate.
\setlist{noitemsep} % or \setlist{noitemsep} to leave space around whole list
\usepackage{datetime}                       % Custom date format for date field
\newdateformat{mydate}{\monthname[\THEMONTH] \THEYEAR}   % Defining month year date format

\title{Canonical Portfolios: Optimal Asset and Signal Combination}

\definecolor{darkblue}{rgb}{0.0,0.0,0.66}   
\definecolor{electricultramarine}{rgb}{0.25, 0.0, 1.0}

\title{Canonical Portfolios: Optimal Asset and Signal
Combination\footnote{
The authors would like to thank Geert Bekaert (editor), Mihai Cucuringu, Brian Healy, Anthony Ledford, Slavi Marinov, Attilio Meucci, Yin Cheng Ng, Ana Maria Pires, and two anonymous referees for the helpful comments. The paper has also benefited from participants in seminars at the University of Oxford and Man AHL.  This work was supported by the Oxford-Man Institute of Quantitative Finance. 
}}

\author{ Nikan Firoozye\thanks{Department of Computer Science, University College London. E-mail: \href{mailto:n.firoozye@ucl.ac.uk}{n.firoozye@ucl.ac.uk}} \and
Vincent Tan\thanks{Oxford-Man Institute of Quantitative Finance, University of Oxford. E-mail: \href{mailto:vincent.tan@eng.ox.ac.uk}{vincent.tan@eng.ox.ac.uk}} 
 \and Stefan Zohren\thanks{Oxford-Man Institute of Quantitative Finance, University of Oxford. E-mail: \href{mailto:stefan.zohren@eng.ox.ac.uk}{stefan.zohren@eng.ox.ac.uk}}}

\newcommand*{\thisdraft}{This version: July 2023} % define command
\newcommand*{\firstdraft}{First version: February 2022}  % define command

\begin{document}
\begin{titlepage}

% \date{\mydate \today}
\date{\firstdraft \\ \thisdraft}
\maketitle
\thispagestyle{empty}

%%%%%%%%%%%%%%%%%%%%%%%%%%%%%%%%%%%%%%%%%%%    ABSTRACT
\vspace{-.5cm}

\setstretch{1.15}
\begin{abstract}
% \onehalfspacingh

This paper presents a novel framework for analyzing the optimal asset and signal combination problem. Our approach builds upon the dynamic portfolio selection problem introduced by \cite{brandt2006dynamic} and consists of two stages. First, we reformulate their original investment problem into a tractable one that allows us to derive a closed-form expression for the optimal portfolio policy that is scalable to large cross-sectional financial applications.  Second, we recast the problem of selecting a portfolio of correlated assets and signals into selecting a set of uncorrelated managed portfolios through the lens of Canonical Correlation Analysis of \cite{hotelling1936relations}. The new investment environment of uncorrelated managed portfolios offers unique economic insights into the joint correlation structure of our optimal portfolio policy. We also operationalize our theoretical framework to bridge the gap between theory and practice, showcasing the improved performance of our proposed method over natural competing benchmarks. 

\end{abstract} \medskip

\begin{flushleft} JEL Classification: G11, D81, C1. \\ \medskip Keywords: Canonical correlation analysis, dynamic portfolio selection, mean-variance analysis \end{flushleft}

\end{titlepage}

\clearpage
% \onehalfspacing

%%%%%%%%%%%%%%%%%%%%%%%%%%%%%%%%%%%%%%%%%%%    INTRODUCTION

\setstretch{1.15}
\section{Introduction}

The investment decisions of portfolio managers are often guided by return-predictive signals that reflect their expectations about future returns.\footnote{The terminology `signals' is also often referred to as predictors, alphas, attributes, characteristics, state variables, instrumental variables, cross-sectional anomalies, etc. in the finance literature.} Traditionally, the process of constructing a portfolio conditional on the investor's information set is divided into two stages according to modern portfolio theory pioneered by \cite{markowitz}. In the first stage, investors make inferences about the return-generating process by either plug-in estimation or subjective belief formation. In the second stage, optimal portfolio weights are formed using estimates from the first stage.

In this paper, we draw inspiration from \cite{brandt2006dynamic} and solve a dynamic portfolio selection problem with multiple return-predictive signals in a single stage. Our first contribution is to reformulate the authors' original investment problem into a tractable one that yields a dynamic portfolio policy that directly invests in each basis asset. Additionally, our approach provides a theoretical framework for optimal investment in the presence of cross-predictability among assets and signals, as well as correlations within them. Our proposed solution is straightforward to implement in practice as the standard static Markowitz solution and is well-suited to handle large-scale cross-sectional applications, such as the equity universe.

The second contribution of our work is to provide insights into our proposed solution by reorganizing the set of assets and accompanying signals into a set of uncorrelated managed portfolios. Each managed portfolio invests in each asset in a proportion that scales with the size of the signals \citep{cochrane2001asset}. As discussed in \citet[Section 4]{firoozye2020optimal}, this can be accomplished through the use of a powerful tool from multivariate analysis known as canonical correlation analysis (CCA), which was developed by \cite{hotelling1936relations} and has a deep economic interpretation, which we provide. CCA generalizes the principal component analysis (PCA) of \cite{hotelling1933analysis} to two sets of random variables and their joint association. This makes it particularly relevant to our case, where we have multiple correlated assets and signals that are mutually linked by a joint correlation objective. CCA allows us to express any portfolio returns derived from our framework in terms of their exposures to uncorrelated sources of returns, by reweighting the original set of managed portfolio returns while maximally retaining their correlation. We refer to these special weights as \textit{canonical portfolios} since they can be viewed as long or short positions across different assets and signals.\footnote{The idea of recasting the asset universe into their orthogonal components through PCA can be traced back to \cite{partovi2004principal} with further extensions and applications from \cite{meucci2009managing} and \cite{avellaneda2010statistical}, as well as from \cite{KozakNagelSantosh2018JF, kozak2020shrinking} in the context of risk pricing in a no near-arbitrage framework. Our work aims to broaden that horizon to \textit{two} sets of variables by also considering the signal as an important input in the analysis of their joint association with returns.} 

There are two common approaches for solving the dynamic portfolio selection problem when conditioning information is available. The first approach involves specifying the joint conditional distribution of asset returns and then forming the optimal portfolio. Alternatively, \cite{hansen1987role} proposed a second approach, which involves augmenting the basis assets to include conditional portfolios, and solving a simpler unconditional mean-variance problem. Research efforts that have built upon the latter approach to developing dynamic trading strategies that exploit time-varying investment opportunities include those by \cite{ferson2001efficient}, \cite{brandt2006dynamic}, and others. Of particular relevance to our study is \cite{brandt2006dynamic}, in which the authors analyze parametric portfolio policies that are linear in the state variables and/or firm characteristics. Through this parameterization, they are able to recast a dynamic portfolio selection problem involving basis assets into a static portfolio selection problem in an asset space augmented with managed portfolios. The optimization of parameters in the static problem leads to a dynamic strategy represented as a fixed combination of managed portfolios, which consists of  ``conditional'' and ``timing'' portfolios. By directly focusing on the portfolio weights, their optimal portfolio policy accounts for time-variation in the entire return distribution, rather than time-variation in expected returns only. 

In this study, we directly model the portfolio weights with a similar parameterization and adopt a joint Gaussian distribution to model the relationship between returns and signals. Our modeling approach allows us to make analytical progress in two main areas. First, it facilitates the derivation of a closed-form dynamic portfolio policy that is applicable even in large-dimensional settings. Second, it allows us to decompose the joint correlation structure of our optimal portfolio policy into linearly independent orthogonal portfolios using CCA, thereby providing valuable insights into the underlying factors that drive our portfolio's performance. 

However, these innovations come with additional input requirements when compared to the conventional mean-variance framework, which solely relies on the covariance of the returns matrix. Specifically, we need to consider the covariance of the signals and the cross-covariances between the returns and signals. To render our theoretical framework practicable, we utilize regularization techniques for these large-dimensional objects to encourage stability in the out-of-sample results. In particular, we apply shrinkage to the covariances, a technique commonly employed in portfolio selection problems; see, for example, \cite{frost1986empirical} and \cite{ledoit2004honey, ledoit2004well, ledoit2017nonlinear}. Additionally, we identify the leading few canonical portfolios that exhibit the highest predictability, discarding noisy and unstable ones. Through a series of backtesting simulations, we demonstrate that our proposed method consistently outperforms competitive benchmarks.

Our study extends the work of \cite{koshiyama2019avoiding} on the use of total least squares (TLS) \citep{golub1980analysis} for optimally combining signals for a univariate return. Their original study focused on a novel objective function for return prediction. While most research analysts will seek to find a good forecast for future returns via methods such as ordinary least squares (OLS), or in the nonlinear context, via a large suite of machine learning-based methods, \cite{koshiyama2019avoiding, firoozye2020optimal} argue that if the goal is to maximize the Sharpe ratio of trading strategy, then this can be achieved through a linear combination of signals that maximizes the correlation between their combination and the returns. In the linear context, the solution to this problem comes out via TLS, which is an errors-in-variables formulation of regression. TLS has been well-studied primarily among numerical analysts and is used less formally on trading desks of investment banks and hedge funds, typically under the moniker of PCA regression. Given the close relation between TLS and CCA, our work extends this research to a multivariate context, which is more relevant for problems in portfolio selection.

The remainder of the paper is organized as follows. Section \ref{sec:2} gives the description of our dynamic portfolio selection problem. Section \ref{sec:3} provides the financial interpretation of our optimal portfolio policy with CCA and details our proposed estimation approach. Section \ref{sec:empirical} describes the empirical methodology and presents the results of the out-of-sample backtest experiments with Fama-French equity sorted portfolios. Section \ref{sec:conclusion} concludes. Appendix \ref{appendix:figures}--\ref{appendix:proofofprop} contain all the figures, tables, and additional mathematical derivations.

\section{Setting the Stage} \label{sec:2}

\subsection{Notation}

In this section, we introduce the notational convention for our analysis. Let the subscript $i$ index the variables such that $i \in \{1, \ldots, N\}$, where $N$ denotes the dimension of the asset universe. The subscript $t$ indexes the trading dates such that $t \in \{1,\ldots ,T\}$, where $T$ denotes the number of observations. The notation $\mathbb{E}_t[\cdot] \coloneqq \mathbb{E}[\cdot | \mathcal{F}_t]$ represents the mathematical expectation operator of a random vector, conditioned on the information set $\mathcal{F}_t$ available at time $t$. Furthermore,  $\mathsf{Cov}(\cdot , \cdot)$ represents the covariance matrix between two random vectors,  $\mathsf{Tr}(\cdot)$ represents the trace of a square matrix, and $\mathsf{Diag}(\cdot)$ represents a diagonal matrix of a given vector. Denote $\mathbb{I}_N$ to be an identity matrix of dimensions $N \times N$. For any $N$-dimensional vector $z$, its $\ell _2$-norm is given by $\|z\|_2 \coloneqq \sqrt{z_1^2 + \ldots z_N ^2}$, and for any real matrix $Z$, its Frobenius norm is $\|Z\|_{\mathrm{F}} \coloneqq \sqrt{\mathsf{Tr}(Z'Z)}$. Here, the symbol $(')$ denotes the transpose operator of a vector or matrix, the symbol $\coloneqq$ is a definition sign, and the symbol $=$ is the equal sign.

Let $r_{t,i}$ be the return for a risky asset $i$ sampled at trading date $t$, stacked into a vector $r_{t}\coloneqq (r_{t,1}, \ldots, r_{t,N})'$.  Also, let $x_{t,i}$ be an $M$-dimensional vector return-predictive signal at date $t$ for asset $i$, stacked into a vector $x_t\coloneqq (x_{t,1}, \ldots, x_{t,N})'$.  This means every asset is accompanied by at least one signal and so $M$ is necessarily a positive integer. For example, if there are two signals for one asset, then $M=2$. The signals are understood to proxy the future expected returns conditional on the investor's information set at time $t$. 

We impose the following assumptions to both the returns and signals whenever applicable.

\begin{assumption}[Stationarity] \label{as:stationarity}
The return sequence $\{ r_{t+1} \}$ and the signal sequence $\{ x_t \}$ exhibit stationarity and ergodicity, whose moments can be calculated by taking time-series averages. 
\end{assumption}

\begin{assumption}[Gaussianity] \label{as:gaussian}
The returns and signals have multivariate distributions that are assumed to be Gaussian with zero expectations and covariances $\mathsf{Var}(r_{t+1}) \coloneqq \mathsf{Cov}(r_{t+1}, r_{t+1})  = \Sigma_r$ and $\mathsf{Var}(x_t) \coloneqq \mathsf{Cov}(x_t, x_t) = \Sigma_x$, and a cross-covariance $\mathsf{Cov}(r_{t+1}, x_t)  = \Sigma_{rx} = \Sigma_{xr}'$. The column vector of stacked subsequent returns and signals $(r_{t+1}, x_t)'$ has a joint covariance matrix expressed in the following partitioned form:
\begin{align}
    \Sigma =
    \begin{pmatrix}
    \Sigma_{r} & \Sigma_{rx} \\
    \Sigma_{rx}' & \Sigma_{x}
    \end{pmatrix} .
\end{align}
\end{assumption}

In our notation, we denote estimated quantities of their population analogs with a hat accent ($\string^$) on them. We use $S_r$, $S_x$, and $S_{rx}$ to represent the sample covariances between the asset returns and signals, the sample covariances of the signals themselves, and the sample cross-covariance between the returns and signals, respectively. Specifically, $S_r$ is computed as $T^{-1} \sum ^{T-1} _{t=0} (r_{t+1}-\bar{r}) (r_{t+1} -\bar{r})'$, where $\bar{r}$ is the sample mean of the returns. Similarly, $S_x$ is computed as $T^{-1} \sum ^{T-1} _{t=0} (x_{t} - \bar{x}) (x_{t} - \bar{x})'$, where $\bar{x}$ is the sample mean of the signals. Finally, $S_{rx}$ is computed as $T^{-1} \sum ^{T-1} _{t=0} (r_{t+1} - \bar{r}) (x_{t} - \bar{x})'$. 

Due to the symmetric positive semidefinite property of the covariance matrices, they admit a spectral decomposition. For example, $\Sigma_{r} = P_r \mathsf{Diag}(d_{r,1} , \ldots , d_{r,N}) P'_r$, where $P_r\coloneqq[p_{r,1}, \ldots, p_{r,N}]$ is an orthogonal matrix ($P_r'P_r=P_rP_r' = \mathbb{I}_N$) whose columns contain the eigenvectors and $\mathsf{Diag}(d_{r,1} , \ldots , d_{r,N})$ is a diagonal matrix containing the eigenvalues. The eigenvalues are assumed to be sorted in ascending order. The matrix factor $\Sigma_{r}^{1/2}$ can be obtained by retaining the eigenvectors of $\Sigma _r$, but reassembling them with the square roots of the eigenvalues given by $(\sqrt{d_{r,1}}, \ldots , \sqrt{d_{r,N}})'$. An analogous notation applies for the decomposition of the covariance of signals $\Sigma _x$ and its matrix of factors $\Sigma_{x}^{1/2}$.

\subsection{Dynamic Trading Strategies}

We begin our analysis by considering the portfolio policies to be linear in the vector of signals at time $t$ based on the following specification. 
\begin{assumption}[Linear Portfolio Policies]\label{as:lpp}
Let $w_t \coloneqq A' x_t$, where $A$ is a $NM \times N$ constant matrix of coefficients.\footnote{Note that the assumption that the optimal portfolio weights are linear functions of the signals is not entirely restrictive. Indeed, each of the signals in $x_t$ can be composed of non-linear functions of a more primitive set of underlying signals.}    
\end{assumption}

The portfolio weights are conditional on the strength of the signals but the matrix $A$ is a static object where column $i$ maps the signal vector to the corresponding portfolio weight invested in basis asset $i$. From this parameterization, the dynamic portfolio selection problem that an investor faces at time $t$ can be formulated in terms of a conditional mean-variance objective as follows
\begin{equation} \label{eq:cond}
    \begin{aligned} 
        &\max _{A} \mathbb{E}_t[x_t ' A r_{t+1}]  -  \frac{\gamma}{2} \mathsf{Var}_t[x_t ' A r_{t+1}] ,
\end{aligned} 
\end{equation}
where $\gamma$ is an investor's risk aversion parameter.\footnote{Other reformulations of this problem include the maximization of mean return or minimization of risk. Regardless of the choice, all three problems result in the same mean-variance trade-off.} This objective function underscores that the investor chooses to simultaneously allocate between the assets and signals in order to yield a trading strategy that optimizes the returns of the portfolio at time $t+1$, which we denote as $r_{t+1}^{w}$.

\cite{brandt2006dynamic} showed that this problem formulation is equivalent to an unconditional portfolio selection problem since the matrix of coefficients is constant through time and Assumption \ref{as:stationarity} holds. Thus, the matrix $A$ that optimizes the investor's conditional problem at any given date is the same for all dates and hence, it also optimizes the investor's unconditional problem. Hence, we can drop the subscript $t$ in \eqref{eq:cond} and consider the following revised problem instead\footnote{Note that Equations \eqref{eq:cond} and \eqref{eq:uncond} are not generally the same when time-varying covariances are present; see \cite{ang2002international}. However, we will defer the investigation of this setting to future research endeavors.}
\begin{equation} \label{eq:uncond}
    \begin{aligned} 
        &\max _{A} \mathbb{E}[x_t ' A r_{t+1}] -   \frac{\gamma}{2} \mathsf{Var}[x_t ' A r_{t+1}]  .
\end{aligned} 
\end{equation}

If we further invoke Assumption \ref{as:gaussian}, we can express this objective in terms of the unconditional second-moments with the following proposition. 
\begin{prop}\label{prop:fullobj}
The objective function to the investor's problem \eqref{eq:uncond} is given by
\begin{equation}\label{eq:obj2}
    \begin{aligned} 
        &\max _{A} \mathsf{Tr}(A \Sigma_{rx}) - \frac{\gamma}{2} \left ( \mathsf{Tr}(\Sigma_{x} A \Sigma_{r} A') + \mathsf{Tr}(\Sigma_{rx} A \Sigma_{rx} A) \right ).
\end{aligned} 
\end{equation}
A closed-form expression for matrix $A$ is provided in Theorem \ref{thm:main} from Appendix \ref{appendix:proofofprop}.
\end{prop}

It is analytically convenient to assume that $\mathsf{Tr}(\Sigma_{rx} A \Sigma_{rx} A) \approx 0$ in Proposition \ref{prop:fullobj}. This approximation offers the advantage of formulating the optimal portfolio policy solely based on unconditional moments, without any recourse to canonical correlation analysis, which is necessary to solve the investment problem \eqref{eq:obj2}. It is a reasonable approximation as long as the weighted sum contributions of squared expected returns from the dollar-neutral managed portfolios, represented as $x_t$ multiplied by the return of each asset, are small and can be approximated as zero.\footnote{For instance, if a portfolio managed on the basis of momentum signals has an average return of 5\%, then the square of its return is 0.25\%.} Notwithstanding this approximation, we have further validated both the exact and approximate solutions empirically in Section \ref{sec:empirical}.

\begin{prop}\label{prop:partialobj}
Consider the following objective function 
\begin{equation}\label{eq:aux}
    \begin{aligned} 
        &\max _{A}  \mathsf{Tr}(A \Sigma_{rx}) -\frac{\gamma}{2}  \mathsf{Tr}(\Sigma_{x} A \Sigma_{r} A') .
\end{aligned} 
\end{equation}
The solution to the investor's problem \eqref{eq:aux} is given by
\begin{align}\label{eq:optA}
    A = \frac{1}{\gamma} \Sigma^{-1}_{x}  \Sigma_{rx} ' \Sigma^{-1}_{r} .
\end{align}
%11
The weight $w_t$ allocated to each asset conditional on the signal at time $t$ is then 
\begin{align}\label{eq:optw}
    w_t = \frac{1}{\gamma} \Sigma^{-1}_{r}  \Sigma_{rx}  \Sigma^{-1}_{x} x_t .
\end{align}
\end{prop}
In the absence of any further constraints on $w_t$, this is the frictionless portfolio policy that invests directly in each asset, and the size of the portfolio is determined by the level of the investor's risk aversion. The dynamics of the portfolio are guided by the trajectory of the signal $x_t$ at each period. We see that the portfolio takes into account the cross-sectional information from the asset returns, the signals as well as their cross-relationships. The cross-covariance matrix $\Sigma _{rx}$ can be interpreted as a matrix consisting of managed portfolios that encapsulate the potential earnings opportunities that are available to an investor.\footnote{In \cite{kelly2020principal}, the authors term the cross-covariance matrix as the `prediction matrix', which is also another valid interpretation since the correlation between an asset return and some signal is one measure of signal-return predictability. Given that a cross-covariance matrix is not symmetric in general, the predictive strength of a signal $i$ on asset $j$ may be different from that of signal $j$ on asset $i$.} For this reason itself, we term this matrix a `conditional portfolios' matrix, which consists of returns scaled by their predictive signals.

We conclude this section by presenting two examples along with some remarks.

\begin{example}[Univariate Portfolio]
    Let us consider one return-accompanying signal and postulate that all covariances in Equation \eqref{eq:optw} are the identity matrix. In this case, the weight is proportional to the conditioning variable, that is, $w_t = \gamma ^{-1} x_t$. This portfolio policy is interesting in its own right because it ignores the impact that the returns and signals can have on the covariances. Hence, it is essentially a univariate system in which the signals only forecast their associated asset. We shall henceforth refer to this conditional portfolio as the univariate factor (UNI).
\end{example}

\begin{example}[Two-Assets Portfolio]
    Suppose that $N=2$, $M=1$, and $\gamma = 1$. If the asset returns and signals have unit variances, then the covariances are 
\begin{equation}
    \Sigma _r = \left( {\begin{array}{cc}
   1 & \rho _r \\
   \rho _r & 1 \\
  \end{array} } \right), 
  \quad \Sigma _{x} = \left( {\begin{array}{cc}
   1  & \rho _x \\
   \rho _x & 1 \\
  \end{array} } \right), \quad \text{and} \quad \Sigma _{rx} = \left( {\begin{array}{cc}
   \xi _{11}  & \xi _{12} \\
   \xi _{21} & \xi _{22} \\
  \end{array} } \right) .
\end{equation}
For simplicity, we also suppose that the signals are uncorrelated, that is, $\rho _x = 0$. Inserting these objects into the optimal matrix $A$, we can write the weights invested in the two basis assets as 
\begin{equation}
    w_t = \left (\frac{\xi_{11} - \rho_r \xi_{21}}{1 - \rho ^2_r} x_1 + \frac{\xi_{12}-\rho _r\xi_{22} }{1 - \rho _r^2}x_2  , \frac{\xi_{21}-\rho _r\xi_{11}}{1 - \rho^2_r} x_1 + \frac{\xi_{22}-\rho _r\xi_{12}}{1 - \rho _r^2} x_2 \right ).
 \end{equation}
Moreover, the expected portfolio return is
\begin{align}\label{eq:er}
    \mathsf{Tr}(A \Sigma _{rx}) = \mathsf{Tr}(\Sigma _r ^{-1} \Sigma _{rx} \Sigma _{x}^{-1} \Sigma _{rx}' ) &= \frac{\xi _{11} ^2 + \xi _{12}^2 + \xi _{21}^2 + \xi _{22}^2 -2 \rho _r (\xi _{11} \xi_{21} + \xi _{12} \xi _{22}) }{1 - \rho _r ^2} .
\end{align}
This expression is non-negative given $\rho _r \in [-1,1]$ since $\xi _{11} ^2 + \xi _{21} ^2 \geq 2\rho _r \xi _{11} \xi _{21}$ and $\xi _{12} ^2 + \xi _{22} ^2 \geq 2\rho _r \xi _{12} \xi _{22}$ by the inequality of arithmetic and geometric means. Let us consider three cases. First, if $\rho _r = 0$, then the expected return is simply the squared expected returns of all conditional portfolios that exploit own- and cross-predictabilities. Second, if $\rho _r \in (-1,0)$, then there is an additional positive contribution to the conditional portfolios that arises from diversifying into imperfectly negatively correlated assets. The position is leveraged according to a factor of $1/(1-\rho _r ^2)$, which increases as $|\rho _r|$ approaches one from below. Third, if $\rho _r \in (0,1)$, then the expected portfolio return is now reduced since the gains to diversification are limited by the extent to which the assets are positively correlated. Notwithstanding the diminished returns, the optimal strategy is also leveraged with a factor $1/(1-\rho _r ^2)$ in order to increase the portfolio returns to meet its profit objective. A similar conclusion holds for the case $\rho _r =0$ and $\rho _x \neq 0$.
\end{example}

\begin{remark}[Conditional Portfolio Policies]
    We can gain further intuition into the role of conditioning that our problem posits. Consider the following conditional portfolio selection problem:
\begin{equation} \label{eq:mvo}
    \begin{aligned} 
        &\max _{w_t} \mathbb{E}_t[w_t ' r_{t+1}]  -  \frac{\gamma}{2} \mathsf{Var}_t[w_t ' r_{t+1}] .
\end{aligned} 
\end{equation}
Since the weights are known at time $t$, the problem can be solved to give
\begin{align}\label{eq:condMarko}
    w_t =  \frac{1}{\gamma}  \mathsf{Var}_t[r_{t+1}] ^{-1}\mathbb{E}_t[r_{t+1}] .
\end{align}
Observe that this optimal portfolio is determined by the mean and variance of returns that incorporate all available information up to time $t$. Hence, restricting the weights to be linear functions of a set of signals also restricts the solution set, which may lead to a suboptimal solution. However, specifying the conditional means and variances of the returns is a notoriously difficult task. Our weights restriction simplifies these dual tasks by converting a dynamic problem into an equivalent static problem for which a conditional portfolio policy is available solely in terms of the unconditional moments of the assets and signals. 
\end{remark}

\begin{remark}[Regression-Based Policies]
Notice that the portfolio weights \eqref{eq:optw} can be seen as one that had been obtained from a two-stage process of forecasting and asset allocation. To see this, let us consider the following return-generating process:
\begin{align}\label{eq:regress}
    r_{t+1} = B' x_t + \varepsilon _{t+1},
\end{align}
where $B$ is a matrix of coefficients and $\varepsilon _t$ is an idiosyncratic error that is independent and identically distributed (i.i.d.). A simple forecasting approach that captures this relationship between both variables would be to perform an ordinary least squares (OLS) regression. In particular, one would first perform $N$ independent univariate OLS regressions of the  returns to assets on the signals in order to obtain a vector of cross-sectional predictive signals and then insert them into \eqref{eq:condMarko} to arrive at the vector of allocations: 
\begin{align}\label{eq:wOLSMarko}
    w_t = \frac{1}{\gamma}\Sigma _\varepsilon ^{-1} \Sigma _{rx} \Sigma _x ^{-1} x_t ,
\end{align}
where $\Sigma _{\varepsilon} \coloneqq \mathsf{Cov}(\varepsilon _t) = \Sigma _r - \Sigma _{rx}\Sigma_{x}^{-1} \Sigma_{rx} '$ is the covariance matrix of the residuals. In fact, this portfolio is also optimal under our assumption of joint Gaussianity between the asset returns and signals. By invoking the standard properties of the multivariate Gaussian distribution, the conditional distribution of $r_{t+1}$ given $x_t$ is also multivariate Gaussian distributed with means and variances given by
\begin{align}
    \mathbb{E}_t[r_{t+1}] &= \Sigma _{rx}\Sigma^{-1}_{x} x_t, \\
    \mathsf{Var}_t [r_{t+1}] &= \Sigma _{r} - \Sigma _{rx} \Sigma ^{-1}_{x} \Sigma _{rx}.
\end{align}
Inserting these expressions into \eqref{eq:condMarko} gives portfolio \eqref{eq:wOLSMarko}. In contrast to our portfolio \eqref{eq:optw}, this regression-based portfolio depends on the covariance structure of the unexplained parts of the returns. Another key difference here is that our portfolio policy respects symmetry; exchanging the roles of the assets and signals yields identical returns to the portfolio. This symmetry property may be warranted in situations where the signals have an influence on the asset returns but at the same time, the reverse association is also economically plausible. It also implies that our approach is more closely related to total least squares (TLS) than it is to OLS since TLS is an errors-in-variables regression model that optimizes for the correlation of both variables, which is a symmetric metric.
\end{remark}

\subsection{Fully Invested Constraint}\label{sec:ptfconstraints}

In quantitative equity investing, it is commonplace to impose economic constraints on the portfolio weights. A prominent example would be the fully-invested portfolio, which enforces a budget constraint so that the weights add to one. This equality constraint can be included in our optimization problem and given that it is linear in the matrix of coefficients $A$, the portfolio weights can be fortunately solved in closed form. 

\begin{prop}\label{prop:fullyinvested}
We can formulate the optimal asset and signal combination problem with a fully-invested constraint as follows
\begin{equation}\label{eq:fi}
    \begin{aligned} 
        &\max _{A}  \mathsf{Tr}(A \Sigma_{rx}) -\frac{\gamma}{2}  \mathsf{Tr}(\Sigma_{x} A \Sigma_{r} A') \\
        \text{subject to} \enskip &  \mathbf{1}'A'x_t =1 ,
\end{aligned} 
\end{equation}
where $\mathbf{1}$ denotes the conformable vector of ones of dimension $N$. The problem has the following analytical solution 
\begin{align}\label{eq:finvested}
    w_t ^{\mathrm{FI}} &= (1-\kappa) \frac{\Sigma_r ^{-1} \mathbf{1}}{\mathbf{1}' \Sigma_r ^{-1} \mathbf{1}} + \kappa \frac{\Sigma_r ^{-1} \Sigma_{rx} \Sigma_x ^{-1} x_t}{\mathbf{1}'\Sigma_r ^{-1} \Sigma_{rx} \Sigma_x ^{-1} x_t}  ,
\end{align}
where $\kappa \coloneqq \gamma ^{-1}\mathbf{1}' \Sigma_r ^{-1} \Sigma_{rx} \Sigma_x ^{-1} x_t$ is a scalar value.
\end{prop}
Proposition \ref{prop:fullyinvested} shows that the portfolio weight is expressed as a convex linear combination of the global minimum variance portfolio and the optimal unconstrained portfolio policy that is renormalized to one. The latter portfolio resembles the so-called `tangency portfolio', which is the highest Sharpe ratio of a portfolio of risky assets from the standard mean-variance framework. The fully-invested portfolio $w_t ^{\mathrm{FI}}$ has a similar form to the one obtained from the standard Markowitz paradigm found in \cite{merton1972analytic}.

\subsection{Relation to Existing Literature}

At this juncture, it is important to highlight the differences between our work from those of the existing literature. We first draw our attention to the fact that the objective function \eqref{eq:aux} under our consideration is similar to that of \cite{brandt2006dynamic}.\footnote{A subtle difference worth pointing out is that the authors model the portfolio policy to be \textit{affine} in the state variables, and hence the allocations are also driven by the assets themselves. Here, we work in a slightly restricted setting where the portfolio policy is only driven by the conditional managed portfolios themselves.} Suppose that the joint distribution of the returns and signals is Gaussian and the squared returns of the conditional managed portfolios are relatively small and can be approximately zero. We can express the optimal matrix of coefficients $A$ from Equation \eqref{eq:optA} in vectorized form as
\begin{align}
    \mathsf{vec}(A) &= \frac{1}{\gamma} \mathsf{vec}( \Sigma^{-1}_{x}  \Sigma_{rx}' \Sigma^{-1}_{r}) \\
    &= \frac{1}{\gamma}  (\Sigma_x^{-1} \otimes \Sigma_r^{-1} ) \mathsf{vec} ( \Sigma_{rx} ') \\
    &=  \frac{1}{\gamma}  (\Sigma_x \otimes \Sigma_r )^{-1} \mathsf{vec} ( \Sigma_{rx}' ) \label{eq:approxbsc} \\
    &\approx \frac{1}{\gamma} \mathbb{E}[  (x_t x_t ') \otimes (r_{t+1} r_{t+1}')  ]^{-1} \mathbb{E} [ x_t \otimes r_{t+1} ], \label{eq:bsc}
\end{align}
where $\mathsf{vec}(\cdot)$ is an operator that stacks the columns of matrix $A$ into a vector, and $\otimes$ is the Kronecker product of two matrices. The fourth (approximate) equality invokes our sparsity assumptions so that the second-moment matrix of the conditional product returns $x_t \otimes r_{t+1}$ is a Kronecker product of two covariances. This corresponds to the solution that one obtains by rewriting the portfolio returns as $x_t ' A r_{t+1} = \mathsf{vec}(A) (x_t \otimes r_{t+1})$, and optimizing over $\mathsf{vec}(A)$ with respect to $x_t \otimes r_{t+1}$.

% without strictly assuming independence between $r_{t+1}$ and $x_t$

The authors considered this new optimization scheme to be an augmented asset space form of the mean-variance optimization. They streamlined the forecasting and allocation process by combining them into a single step, effectively minimizing the potential for model misspecification and estimation error \citep{brandt1999estimating}.  However, there are two issues with this approach. First, augmenting the asset space results in a severe expansion of the problem dimension. The unrestricted covariance matrix in \eqref{eq:bsc} has $N^2M(N^2M+1)/2$ parameters, while our Kronecker product factorization in \eqref{eq:approxbsc} has only $N(N+1)/2 + NM(NM+1)/2$ parameters. Clearly, the former becomes less computationally tractable when the dimensions of both the asset and signal space are large. Second, there is less clarity on the structure of their optimal solution due to the complicated interactions in the covariance between the subsequent returns and signals. Notwithstanding these shortcomings, our efforts in recasting the optimal strategy  \eqref{eq:bsc} into \eqref{eq:optA} (or  \eqref{eq:optAfull} in Theorem \ref{thm:main}) help to amend their approach towards large-scale financial applications. Although this requires us to model the joint distribution of returns and signals, our solution lends itself nicely to the study of the joint correlation structure of the optimal matrix of coefficients $A$, as we will demonstrate in Section  \ref{sec:ccaptf}.

Second, our work is related to the classic portfolio selection problem of \cite{markowitz} in its static form. If we assume independent and identically distributed (i.i.d.) returns with constant moments and time-invariant weights $w_t \equiv w$, the conditional portfolio problem \eqref{eq:mvo} simplifies to an unconditional one, which has the following well-known solution
\begin{align}
    w = \frac{1}{\gamma} \mathsf{Var}[r_{t+1} ]^{-1}\mathbb{E}[r_{t+1}] .
\end{align}
This portfolio policy is an unconditional strategy because it does not use any information today and is typically implemented through plug-in estimation. In this paper, we focus on the more realistic setting of non-i.i.d. returns, which motivates our model parameterization and use of conditioning variables that affect the joint distribution of returns.

Finally, our work is also related to \cite{kelly2020principal} on cross-predictability. In their work, the authors pursued a similar parameterization for the portfolio weights but optimize an objective function that is subject to a robust risk constraint that controls for the leverage of $A$. As a result, their optimal portfolio policy depends only on the cross-covariance matrix. When viewed in the context of our framework, the analysis is simplified and avoids the challenges of inverting two additional large-dimensional matrices. However, this departs from our paper on two fronts. First, we assume that the investors have mean-variance preferences in a similar spirit to the standard mean-variance problem.\footnote{\cite{levy1979approximating} and \cite{markowitz1991foundations} argue that the mean-variance preferences serve as a reasonable approximation to other utility preferences in portfolio selection problems.} Second, our application of canonical correlation analysis formally suggests that the cross-covariance between \textit{decorrelated} asset returns and signals is the more appropriate object for investment analysis as opposed to the cross-covariance of the original asset returns and signals. This is because the latter contains non-trivial variations that are embedded in both financial variables that may obscure the inference.

\section{Strategy Diversification}\label{sec:3}

\subsection{Canonical Correlation Analysis}\label{sec:cca}

The goal of canonical correlation analysis (CCA) is to perform dimension reduction on two different data sets that comprise a large number of interrelated variables, while simultaneously retaining as much of the correlation present in the two data sets. This can be achieved through a transformation of the original set of variables into a new set of mutually orthogonal paired variables, which are ranked so that the largest few retain most of the correlation present in all of the original variables.

There are several standard approaches to solving the CCA problem as outlined in \cite{uurtio2017tutorial}. Indeed, this may involve solving a standard eigenvalue problem \citep{hotelling1936relations}, or a singular value decomposition (SVD) \citep{healy1957rotation}, or through a generalized eigenvalue problem \citep{bach2002kernel}. We adopt the SVD approach in this brief exposition of the classical subject by emphasizing a sequence of two change of basis operations to arrive at a reduced problem whose financial interpretation we provide in Section \ref{sec:ccaptf}.

Given that scale invariance is a property of correlation, the first step is to orthogonalize the random vectors $r$ and $x$ such that their covariances are the identity matrices. This can be achieved through the following linear transformations $\tilde{r} \coloneqq \Sigma^{-1/2}_r r$ and $\tilde{x} \coloneqq \Sigma^{-1/2}_x x$. The transformed objects are now uncorrelated and have a unit variance and their joint covariance matrix is given by
\begin{align}
    \begin{pmatrix}
    \mathbb{I}_{N} &\Sigma_{r}^{-1/2} \Sigma_{rx} \Sigma_{x}^{-1/2} \\
    \Sigma_{x}^{-1/2} \Sigma_{rx}' \Sigma_{r}^{-1/2} & \mathbb{I}_{NM}
    \end{pmatrix} .
\end{align}
It is helpful to introduce the object
\begin{align}\label{eq:adjcov}
    \Sigma_{\tilde{r}\tilde{x}} \coloneqq \mathbb{E}[\tilde{r} \tilde{x}'] = \Sigma_{r}^{-1/2} \Sigma_{rx} \Sigma_{x}^{-1/2},
\end{align}
as the adjusted cross-correlation matrix between the transformed variables $\tilde{r}$ and $\tilde{x}$.\footnote{Note that the matrix \eqref{eq:adjcov} is different from the \textit{unadjusted} cross-covariance matrix $\Sigma _{rx}$. The adjusted cross-covariance matrix is the central object in CCA as it internalizes the variabilities from the asset returns and signals that might potentially obscure the relationships between both of the variables. Moreover, since the optimal matrix of coefficients for our problem \eqref{eq:aux} is indeed the adjusted cross-covariance matrix under an appropriate change of basis (see Section \ref{sec:ccaptf}), it makes sense for us to focus on this object for our analysis.} 

The second step is to perform an SVD operation on the adjusted cross-covariance matrix $\Sigma_{\tilde{r}\tilde{x}}$. Let $((s_1, \ldots , s_N); (u_1, \ldots , u_N); (v_1, \ldots , v_N))$ denote a system of singular values and singular vectors of the adjusted cross-covariance matrix $\Sigma_{\tilde{r}\tilde{x}}$. We assume that the singular values $s_i$ are sorted in increasing order. The canonical correlations correspond to the singular values. The canonical variates of $r$ and $x$ are defined as $u'_1 \tilde{r}, \ldots , u' _N  \tilde{r}$ and $v'_1  \tilde{x}, \ldots , v' _N  \tilde{x}$, respectively. As the singular values, $s_i$ are sorted in increasing order, the canonical variate with the largest correlation is given by the pair $(u'_N \tilde{r}, v' _N  \tilde{x})$ while the canonical variate with the smallest correlation is $(u'_1 \tilde{r}, v'_1  \tilde{x})$. 

Finally, since the solutions are expressed on a different basis than our original problem, we have to translate them back to our original basis. By inverting the change of variables we made, the $i$th canonical variates can be written as $(u'_i \Sigma_{r}^{-1/2} r, v'_i \Sigma_{x}^{-1/2} x)$. This pair can be seen as a linear combination of the original assets and signals with coefficients given by $q_{r,i} \coloneqq \Sigma^{-1/2}_r u_i$ and $q_{x,i} \coloneqq \Sigma^{-1/2}_x v_i $, which are the so-called canonical directions. The application of change of basis operations simplifies the covariance structure considerably. 

In order to implement the CCA in practice, a common approach is to replace the population second-moments $\Sigma_{r}, \Sigma_{x}$, and $\Sigma_{rx}$, with their sample counterparts $S_{r}, S_{x}$, and $S_{rx}$ computed from random samples $r_1, \ldots, r_{T}$ and $x_0, \ldots, x_{T-1}$. Let $((\hat{s}_1, \ldots , \hat{s}_N); (\hat{u}_1, \ldots , \hat{u}_N); (\hat{v}_1, \ldots , \hat{v}_N))$ be a system of singular values (sorted in increasing order) and singular vectors of the sample analog of \eqref{eq:adjcov} denoted by $S_{\tilde{r}\tilde{x}}$. Then the sample canonical variates of $r$ and $x$ are given by $\hat{u}'_1 \tilde{r}, \ldots , \hat{u}'_N \tilde{r}$ and $\hat{v}'_1  \tilde{x}, \ldots , \hat{v}'_N  \tilde{x}$, respectively. However, one should be skeptical of sample canonical correlations as they are not reflective of the population canonical correlations. Indeed, the true canonical correlation of the $i$th canonical variate pair $(\hat{u}'_i \tilde{r}, \hat{v}'_i  \tilde{x})$ is $\hat{v}_i ' \Sigma_{\tilde{r}\tilde{x}}' \hat{u}_i$, as opposed to $\hat{s}_i = \hat{v}_i ' S_{\tilde{r}\tilde{x}}' \hat{u}_i$. This is because the sample canonical correlations are known to be inconsistent estimates of their population counterparts when the dimensionality of the problem for both variables is large relative to the number of observations. 

To bring this point home, let us consider a specific case when the two variable sets $r$ and $x$ are independent and each set consists of i.i.d. Gaussian random variables. Figure \ref{fig:lsd} shows the distribution of sample canonical correlations for various values of the limiting ratios $N/T \in \{ 0.003, 0.03, 0.3\}$ and $M/T \in \{ 0.006, 0.06, 0.6 \}$. We obtained this figure by invoking the asymptotic formula derived by \citet[Theorem 3.1]{wachter1980limiting}.\footnote{Wachter's result also holds under a less restrictive setting of i.i.d. data with finite second-moments; see \citet[Theorem 1]{yang2012convergence}.} We observe that the sample canonical correlations are biased upwards away from zero and the severity of the bias worsens as $N/T$ or $M/T$ or both increases. This distributional shape is expected to hold when $N$, $M$, and $T$ are approximately large and it only depends on the ratio $N/T$ and $M/T$; this is true regardless of any particular realization of the adjusted sample cross-covariance matrix. We attempt to address these problems with regularization in Section \ref{sec:estimation}. 

\subsection{Relation to PCA, PLS, and RRR}

Having laid the groundwork for CCA, it is interesting to explore the relationship between CCA and a variety of linear dimension-reduction techniques that are commonly employed to reveal underlying economic structures. We shall restrict our attention to the principal component analysis (PCA), partial least squares (PLS), and reduced rank regression (RRR) due to their tight connection with one another and extensive applications in the finance and econometrics literature. 

PCA, for instance, has been used to extract common factors from a covariance matrix of individual stocks among others by \cite{roll1980empirical}, \cite{connor1986performance, connor1988risk}, and \cite{lehmann1988empirical}. In the context of describing risk compensation through estimation of the stochastic discount factor, \cite{KozakNagelSantosh2018JF, kozak2020shrinking} invokes no near-arbitrage arguments for the use of PCA to identify factor risk exposures that price the cross-section of expected returns of characteristic-managed
portfolios. For PLS, one of the earliest applications to finance can be ascribed to \cite{kelly2013market}, where the authors used a three-pass regression filter (which is a special case of PLS) to forecast market return and cash flow growth. The RRR, alongside CCA, features prominently in the estimation of the celebrated vector-error-correction model (VECM) of \cite{engle1987co} as a means to detect cointegration, which are economic relationships that cannot be obtained through standard regression approaches; see \cite{johansen1988statistical, johansen1991estimation, johansen1995likelihood}.

% https://memento.epfl.ch/public/upload/files/ZviadadzePaper.pdf

Starting with PCA developed by \cite{hotelling1933analysis}, the method seeks to break down the variation of the data into mutually orthogonal components whose variances can be ranked from the smallest to the largest. In the context of a paired dataset consisting of asset returns with their associated signals, PCA focuses on searching for directional vectors that have the maximum variation in each variable separately. In contrast to CCA, however, it does not take into account any relationship between the two financial variables. In this paper, we would like to focus on how financial variables are related, not how much they vary. We summarize the key properties distinguishing PCA and CCA in Table \ref{tab:lineardim}. 

On the other hand, PLS was introduced by \cite{wold1975soft} as an econometric technique that aims to explain the relationship of a paired dataset by studying the (cross-)covariance. This can be understood by relaxing the variance constraints in the optimization problem of CCA, which results in a maximum (cross-)covariance problem. The solution(s) can be obtained via singular value decomposition of the cross-covariance matrix, $\Sigma _{rx}$, and the corresponding managed portfolios closely mimic that of \cite{kelly2020principal}. This differs from CCA where the cross-covariance matrix is normalized with respect to the covariances from both the $r_{t+1}$ and the $x_t$ variables.  

% https://www.diva-portal.org/smash/get/diva2:288565/FULLTEXT01.pdf
% Hence, these three cases can be seen as the same problem, covariance maximization, where the variables have been subjected to different, data-dependent, scaling. 

% https://stats.stackexchange.com/questions/206587/what-is-the-connection-between-partial-least-squares-reduced-rank-regression-a

% https://edepot.wur.nl/248794
% https://www.jstor.org/stable/pdf/2336276.pdf?refreqid=excelsior%3A6d64082f63a4fb0e7c31a23adf088d39&ab_segments=&origin=&acceptTC=1

%

Finally, the RRR proposed by \cite{anderson1949estimation, anderson1950asymptotic} and \cite{anderson1951estimating} can be understood through the regression framework \eqref{eq:regress} with the following error minimization problem: 
\begin{align}\label{eq:objRR}
    \min _ B \mathbb{E} \big [ \big \| \Gamma ^{1/2} (r_{t+1} -B'x_t) \big \|^2 _2 \big ],
\end{align} 
where $\Gamma$ is a $N \times N$ weighting matrix, and the matrix $B$ is subject to a rank constraint; see also, \cite{velu2013multivariate}. Expressing the objective function in terms of the unconditional moments, we have
%
 % tr [E{r2 (Yt -ABXt) ( Yt -ABXt)'Fl}
%http://weixingsong.weebly.com/uploads/7/4/3/5/7435707/modern_multivariate_analysis.pdf
\begin{align}
    \mathbb{E} \big [ \big \| \Gamma ^{1/2}  (r_{t+1} -B'x_t) \big \|^2 _2 \big] 
    &= \mathbb{E} \big [ \big \|\Gamma ^{1/2}  r_{t+1} \big \| ^2_2 \big] - 2\mathbb{E}\big [x_t 'B \Gamma    r_{t+1} \big] + \mathbb{E}\big [x_t 'B  \Gamma  B' x_t \big] \\
    &= N - 2\mathsf{Tr}(B \Gamma    \Sigma _{rx}) + \mathsf{Tr}(B \Gamma   B' \Sigma _x) \\
        & =\big  \| \Gamma ^{1/2} B\Sigma^{1/2}_x - \Gamma ^{1/2} \Sigma _{rx}\Sigma _x ^{-1/2} \big \| ^2 _{\mathrm{F}} + \text{constant}. 
\end{align}
The solution to this rank-restricted minimization problem follows from the theorem of \cite{eckart1936approximation} and it relies on the singular value decomposition of $\Gamma ^{1/2} \Sigma _{rx}\Sigma _x ^{-1/2}$.  The typical objective function of RRR occurs when $\Gamma =\mathbb{I}_N$. If additionally, $\Sigma_x = \mathbb{I}_M$, then RRR coincides with PLS. Finally, observe that if $\Gamma = \mathbb{I}_N$ and $r_{t+1}=x_t$, then the objective function coincides with that of PCA.

Interestingly, if $\Gamma = \Sigma _r ^{-1}$, then the solution of RRR generalizes to CCA. This speaks to the difference between optimizing a regression-based objective and a correlation-based objective. Indeed, regression attempts to explain as much of the variation in the returns $r_{t+1}$ as possible. However, if $r_{t+1}$ is decorrelated and isotropic along all variance directions, then minimizing a regression loss function becomes equivalent to maximizing a correlation objective. Unlike conventional regression-based approaches,  there is no notion of independent and dependent variables in CCA because the correlation metric is symmetric; that is, if we exchange the roles of asset returns and signals, CCA yields identical outcomes for the canonical directions and canonical correlations.

In summary, RRR, PLS, and CCA can all be viewed under the same framework, which is solving a cross-covariance maximization problem. Each of these linear dimensional reduction techniques is subject to different normalization schemes with CCA serving as the most general one that nests these methods.

\subsection{Reformulation as Canonical Portfolios}\label{sec:ccaptf}

The portfolio returns obtained with \eqref{eq:optA} is not particularly intuitive as the expression involves large-dimensional matrix inversion and multiplication operations. However, we can make progress by performing a CCA in order to decompose the portfolio selection problem into one that we can provide financial interpretation. 

We start by expressing the portfolio returns in terms of their transformed objects as 
\begin{align}\label{eq:optwret}
    x_t ' Ar_{t+1} = \frac{1}{\gamma} x_t ' \Sigma^{-1}_{x}  \Sigma_{rx}' \Sigma^{-1}_{r} r_{t+1} = \frac{1}{\gamma} \tilde{x}_t '  \Sigma_{\tilde{r}\tilde{x}} \tilde{r}_{t+1} ,
\end{align}
Since the decorrelated objects $\tilde{r}_{t+1}$ and $\tilde{x}_t$ span the same universe as the original assets and signals, we shall term them `synthetic assets' and `synthetic signals'. These newly defined objects have identity covariances, and so this is the cross-sectional version of risk parity.

From expression \eqref{eq:optwret}, we can see that the portfolio returns depend intricately on the synthetic asset returns and synthetic signals, which are coupled through their adjusted cross-covariances. The portfolio returns do not depend on whether we exchange the role of the asset returns of signals; what matters is how they are related to each other.

Our next step is to choose basis vectors such that $\Sigma_{\tilde{r}\tilde{x}}$ is diagonal. In particular, we perform an SVD operation such that all cross-relationships are eliminated in the new basis in order to arrive at the following simplified expression for the portfolio returns:
\begin{align}\label{eq:optwsvd}
    x_t ' Ar_{t+1} = \frac{1}{\gamma} \sum ^N _{i=1} s_i  (v_i ' \tilde{x}_t ) (u_i' \tilde{r}_{t+1} ) .
\end{align}
Thus, we see that a generic strategy return can be viewed either as a combination of the original assets and signals that have been optimally blended with the matrix of coefficients $A$, or as a combination of $N$ uncorrelated managed long-short portfolios weighted by their canonical correlations. The $N$ long-short managed portfolios span the same space of investment opportunities as the original assets and signals.

Economically, we interpret the change of basis as an operation that reorganizes the set of $N$ synthetic assets and $N$ synthetic signals into a set of $N$ uncorrelated managed portfolios. Each of these managed portfolios is expressed as a certain weighted combination of the assets and signals themselves. These particular combinations are determined by the singular vectors or canonical directions; in either case, we shall refer to them as \textit{canonical portfolios} henceforth to reinforce the idea that these vectors are essentially managed portfolios that are loading onto their respective asset and signal variables. The higher the $i$th canonical correlation, the higher the return of the $i$th canonical portfolio. Said differently, the canonical portfolios are ordered in such a way that the highest canonical correlation corresponds to the \textit{most linearly predictable portfolio}, and the second highest canonical correlation, the \textit{second most linearly predictable portfolio}, and so on.  This notion of predictability distinguishes canonical portfolios from other concepts of orthogonal managed portfolios such as those constructed by PCA.

% \footnote{PCA decomposes the variance of a portfolio of correlated returns, while CCA decomposes the expected returns of a portfolio of synthetic returns and synthetic signals. The variance of each principal portfolio equals its corresponding eigenvalue, whereas the expected return of each canonical portfolio equals its squared singular value. }

Circling back to Equation \eqref{eq:optwsvd}, the source of the portfolio returns can be seen to be distributed across the canonical portfolios of $\Sigma_{\tilde{r}\tilde{x}} $. This reinforces the economic notion of diversification: Having exposure to different uncorrelated canonical portfolios is akin to putting each of your eggs into different baskets. We see that the capital assigned to each canonical portfolio is proportional to its original returns $r_{t+1}$ and signal $x_t$, and proportional to the correlation $s_i$. Put differently, we want to assign capital to managed portfolios that have high correlations but are also orthogonal to each other. 

\subsection{Canonical Portfolio Analysis}

With our CCA decomposition, we can also gain some insight into how the canonical portfolios impact the returns of a portfolio. We start by working out the expected returns and variance of these canonical portfolios with the next proposition. 

\begin{prop}\label{prop:return2singvals}
    Let $\pi _{t,i} \coloneqq s_i (v_i ' \tilde{x}_t ) (u_i' \tilde{r}_{t+1} )$ be the return that the $i$th canonical portfolio generates. The expected value and variance return of each canonical portfolio are  
    \begin{align}\label{eq:s2}
        \mathbb{E}[\pi _{t,i}] &=   s_i ^2, \quad \text{and} \quad \mathsf{Var}[\pi _{t,i} ] = s_i ^2 (1+s_i ^2) .
    \end{align}
    Furthermore, the expected value of the optimal portfolio returns and its squared returns are given by 
    \begin{align}\label{eq:sums2}
        \mathbb{E}[x_t 'A r_{t+1} ] &=  \frac{1}{\gamma} \sum _{i=1} ^N s_i ^2, \quad \text{and} \quad \mathbb{E}[x_t ' A \Sigma _r A' x_{t}] =  \frac{1}{\gamma ^2} \sum _{i=1} ^N s_i ^2 .
    \end{align}
\end{prop}
We see that the expected portfolio returns are positive, and hence, the optimal portfolio is a profit-generating strategy that has economic value. Moreover, the expected portfolio returns can be expressed as the sum of the contributions from the squared canonical correlations; the larger the expected return of the $i$th canonical portfolio, the larger its contribution towards the overall returns.\footnote{The formulas \eqref{eq:s2} and \eqref{eq:sums2} have some close similarities to a result in \citet[Proposition 4]{kelly2020principal} for the expected returns. One notable difference, however, is that our singular values are raised to the power of two. This is implicit in the fact that the expected return of each canonical portfolio is scaled by a leverage factor of $s_i$. On the other hand, this multiplicative factor is absent in their expression since leverage is explicitly controlled for in their problem formulation.} This suggests that the squared canonical correlations can serve as a natural measure to rank the performance of the canonical portfolios.

Interestingly, the canonical correlations are also intimately connected to the Sharpe ratio of the returns associated with our portfolio policy. 

\begin{corollary}\label{corollary:sharpe}
    The Sharpe ratio of the $i$th canonical portfolio and the optimal portfolio is given by
    \begin{align}\label{eq:IR}
        \mathsf{SR}_i \coloneqq \frac{\mathbb{E}[\pi _{t,i}]}{\sqrt{\mathsf{Var}[\pi _{t,i}  ]}} = \frac{s_i}{\sqrt{1 + s_i ^2}}, \quad \text{and} \quad  \mathsf{SR} \coloneqq \frac{\mathbb{E}[x_t 'A r_{t+1} ]}{\sqrt{\mathbb{E}[x_t ' A \Sigma _r A' x_{t} ]}} = \sqrt{\sum ^N _{i=1} s_i ^2} .
    \end{align}
\end{corollary}

From Corollary \ref{corollary:sharpe}, we see the squared canonical correlations also influence the portfolio's Sharpe ratio. We end this section with two simple examples to give further intuition into these results.

\begin{example}[Equal canonical correlations]
If we suppose $s_i = s$ for all $i$, then Equation \eqref{eq:IR} simplifies to $\mathsf{SR} =  s \sqrt{N}$. In this special case, we arrive at the Fundamental Law of Active Management of \cite{grinold1989fundamental},  where the so-called `information coefficient' is $s$ and the effective number of statistically independent investments (also referred to as `breadth') is $N$. We see that the Sharpe ratio improves with the average canonical correlation and the number of canonical portfolios, with the latter exhibiting diminishing gains to returns due to the square-root scaling in its exponent. In the general case where the canonical correlations are not identical across all canonical portfolios, the canonical correlations are thought to be the information coefficient over clusters of signals.\footnote{We thank Attilio Meucci for this insight.} 
\end{example}

\begin{example}[Bias of In-Sample Returns] \label{prop:bias}
Let  $\hat{A}\coloneqq \gamma^{-1}S_{x}^{-1} S_{rx}' S_r^{-1}$ is the optimal matrix of coefficients that replaces the population moments in $A$ with sample moment-based estimates. For simplicity, suppose that $S_r = \Sigma_r$ and $S_x = \Sigma_{x}$. Since the Frobenius norm is a convex function, by Jensen's inequality, the expected portfolio return satisfies
\begin{align}
    \mathbb{E}[\mathsf{Tr}(\hat{A}S_{rx})] = \frac{1}{\gamma}\mathbb{E}[\tr{(S_{\tilde{r}\tilde{x}}'S_{\tilde{r}\tilde{x}}})] \geq \frac{1}{\gamma}\mathsf{Tr}(\mathbb{E}[S_{\tilde{r}\tilde{x}}]'\mathbb{E}[S_{\tilde{r}\tilde{x}}] ) = \frac{1}{\gamma}\mathbb{E}[\tr{(S_{\tilde{r}\tilde{x}}'\Sigma_{\tilde{r}\tilde{x}}})],
\end{align}
where we used the fact that $\mathbb{E}[S_{rx}] = \Sigma_{rx}$ in the last equality due to the unbiasedness of the sample cross-covariance. Hence, 
\begin{align}
    \sum ^N_{i=1} \mathbb{E}[\hat{s}_i ^2 ] \geq \sum ^N_{i=1} s_i ^2 = \sum ^N_{i=1} \mathbb{E}[\hat{s}_i s_i ^\circ] ,
\end{align}
where $s_i ^ \circ \coloneqq \hat{q}_{r,i}' \Sigma_{rx} \hat{q}_{x,i}$ is the out-of-sample `canonical correlation' associated with the $i$th estimated canonical portfolio $\hat{q}_{r,i}$ and $\hat{q}_{x,i}$.\footnote{In fact, a more precise relationship between the in-sample and out-of-sample portfolio returns have been established by \citet[Proposition 2.9]{benaych2023optimal}.} Hence, on average, the in-sample returns are always optimistic but the out-of-sample evaluation disappoints. This is because both the in-sample singular values and singular vectors are estimated with a bias. Hence, in order to ensure that the in-sample and out-of-sample returns are more in sync, we have to shrink the in-sample singular values and align singular vectors closer to the truth. 
\end{example}

\subsection{Static and Dynamic Returns Decomposition}

In our standing assumptions, we supposed that both the asset returns and signals are centered and have zero means. However, when evaluating the out-of-sample performance of the portfolio, the original non-demeaned returns and signals tend to be used instead. Moreover, the means of both financial variables may be non-negligible and contain useful information for describing the cross-predictability of future returns. In general, we can write the second-moment matrix of conditional portfolios using the following decomposition:
\begin{equation}\label{eq:sdbets}
    \Sigma _{rx} = \mathbb{E}[r_{t+1} x_t ']= \mathbb{E}[r_{t+1}]\mathbb{E}[x_t]' + 
    \mathsf{Cov}(r_{t+1}, x_t) .
\end{equation}

The first component is a rank-one outer product matrix of the unconditional means in the asset returns and signals while the second component is the cross-covariance matrix of the demeaned returns and signals. In particular, the first component places emphasis on the cross-sectional differences driven by the unconditional levels since all time-series variabilities have been averaged out, while the second component captures time-series variations in the returns and signals that are expressed as the deviations from their unconditional means. \cite{kelly2020principal} describe the first component as static bets and the second component as dynamic bets. 

We can derive the expected returns of the optimal portfolio in this setting in the next proposition.

\begin{prop}\label{prop:staticdyn}
    If the expected value of the returns and signals are $\mathbb{E}[r_{t+1}]=\mu _r$ and $\mathbb{E}[x_t]=\mu _x$, then the expected return of a portfolio utilizing the second-moment matrix \eqref{eq:sdbets} is given by
    \begin{align}\label{eq:staticdynret}
        \mathbb{E}[x_t 'A r_{t+1} ]  =  
 \frac{1}{\gamma} (\mu _r ' \Sigma _r^{-1} \mu _r)(\mu _x ' \Sigma _x^{-1} \mu _x) + \frac{1}{\gamma} \sum ^N_{i=1} s_i ^2 .
    \end{align}
\end{prop}
Proposition \ref{prop:staticdyn} shows that there is an additional non-negative contribution to the portfolio returns, which is due to the (squared) maximum Sharpe ratio achievable from the assets. This observation aligns nicely with our interpretation that the sum of the squared canonical correlations is related to the squared Sharpe ratio of optimal portfolio returns from Corollary \ref{corollary:sharpe}. That is, the expected portfolio returns can be attributed to the Sharpe ratios from both static and dynamic bets.

We shall let our portfolio exploit both investment opportunities that arise from the second-moment matrix of managed portfolio returns. This is possible within our framework since the assumption of zero-mean returns and signals strictly applies to the variance of the portfolio returns, and so only the covariances of the returns and signals have to be centered. However, we shall apply CCA to the cross-covariance matrix to fulfill its modeling assumptions since otherwise, it can result in a generic top canonical portfolio that is mainly driven by the unconditional levels.\footnote{CCA is sometimes applied to non-demeaned variables. In such a case, this is more closely related to using a cosine similarity objective instead of a correlation objective.}

\subsection{Estimation} \label{sec:estimation}

Our optimal portfolio policy requires the knowledge of the population covariances of the asset returns, the covariances of the signals, and cross-covariance between both variables. These objects are generally unknown to us, and so in order to render our framework to practice, we have to estimate them with real data. Given the large-dimensional nature of our problem, it is necessary to regularize the covariances of our portfolio to reduce the estimation errors. Unfortunately, simply maintaining the top canonical portfolio is not sufficient for empirical analysis. Indeed, if either the assets or the signals have rank $T$ (which tends to be the case when $N \gg T$), then the canonical portfolios can take on any arbitrary values. Moreover, since the smallest amount of variability in each dataset gets rescaled to one, CCA can produce spurious outcomes. 

The challenge of estimating the covariance matrix of financial covariances is well known amongst practitioners \citep{jobson1980estimation}. A standard approach is to use the sample covariance matrix. However, when the dimensionality of the problem is large relative to the number of observations, estimation error of the sample covariance matrix can create issues for portfolio optimizers; they tend to place extreme bets on low-risk sample eigenvectors. In fact, this observation led \cite{michaud} to refer to mean-variance optimizers as `error maximization' schemes. There have been several approaches from practice to address this problem using methods from bootstrapping \citep{michaud2008efficient} to Bayesian estimators \citep{black1992global, lai2011mean}.

We obtain regularized covariance matrices by applying the linear shrinkage technology from \cite{ledoit2004well}. The covariances of returns will be estimated as $\hat{\Sigma}_r = (1-\delta _r) S_r + \delta _r N^{-1} \mathsf{Tr}(S_r)  \mathbb{I}_N$, where the shrinkage intensity $\delta _r$ is determined based on an asymptotic formula. We also apply the same linear shrinkage technology for the covariances of signals $\hat{\Sigma}_x$ with a shrinkage intensity parameter $\delta_x$. The choice of linear shrinkage of the covariances has a nice interpretation in our context in that varying the shrinkage intensities allows us to interpolate between the maximum covariance (PLS) and maximum correlation (CCA) problems. Shrinking towards maximum covariance helps to break the singularities by considering managed canonical portfolios that have better out-of-sample properties. In the extreme regularization setting with both $\hat{\Sigma}_r$ and $\hat{\Sigma}_x$ being identity matrices, our estimated canonical portfolios mimics the principal portfolios approach of \cite{kelly2020principal}.

With both of these estimated covariances at hand, we proceed to build a regularized adjusted cross-covariance matrix by pre-conditioning the sample covariances $S_{rx}$ on both sides with the estimated matrix factors $\hat{\Sigma} _r ^{1/2}$ and $\hat{\Sigma} _x ^{1/2}$. One can then choose to maintain the top few canonical correlations of the regularized adjusted sample cross-covariance matrix defined as
\begin{align}\label{eq:regularizedxscov}
    \hat{\Sigma} _r ^{-1/2} S_{rx} \hat{\Sigma} _x ^{-1/2} ,
\end{align}
and set the bottom ones to zero.\footnote{The regularized sample adjusted cross-covariance matrix \eqref{eq:regularizedxscov} has a close similarity to that of \cite{vinod1976canonical}, where ridge regression was proposed as a means of regularizing the sample adjusted cross-covariance matrix. We resort to the class of linear shrinkage estimators due to their ability to also reduce the influence of large variance directions.} This thresholding operation has the effect of regularizing the problem as it reduces the effective number of parameters that we have to estimate by maintaining the top few canonical portfolios that are the most predictable.

\section{Empirical Analysis}\label{sec:empirical}

\subsection{Data and Portfolio Construction Rules}
For our empirical analysis, we download six datasets from Kenneth French's data library, which are characteristic-sorted long-short stock portfolios.\footnote{The description of all portfolio construction can be found on Kenneth French's website: \url{https://mba.tuck.dartmouth.edu/pages/faculty/ken.french/data_library.html}. At the time of writing, these datasets were based on the 10-2022 CRSP database.} They are daily returns on portfolios of stocks sorted on the basis of size and book-to-market (FF), size and operating profitability (ME/OP), and size and investment (ME/INV), each of which is of universe size 25 and 100. In comparison to individual stocks, each return of a  portfolio is an average return of a group of stocks sharing similar characteristics and so they are less subject to large variabilities due to idiosyncratic risks. Hence, they serve as useful test assets that may allow us to easily harvest the predictability in stock returns.

Although some of these portfolio returns have been available since 1926, we conduct most of our analysis on the period from July 1963 to October 2022, for which most of the returns are available. For simplicity, we suppose that 21 consecutive days constitute one trading `month' and 252 consecutive days as one trading `year'. We adopt a sequential updating scheme and rebalance the portfolio every `month' on a rolling walk-forward basis. To obtain a well-defined investment universe for which we can estimate the portfolios, we use the following rule. For each rebalancing date, we select test assets that have a complete return history over the most recent $T=120$ months as well as a complete return `future' over the subsequent trading month. The backward and forward restrictions ensure that we have data to estimate our models and to evaluate out-of-sample. This provides us with 578 monthly out-of-sample returns, which covers the out-of-sample investment period from 09-09-1974 through 10-21-2022. 

\subsection{Signal Construction}

As a demonstration of our method for a dynamic portfolio selection problem between equity-sorted portfolios, we start with a single conditioning variable as our signal. We consider the momentum signal as there is substantial empirical evidence that documents this anomaly in the returns of individual stocks \citep{jegadeesh1990evidence, jegadeesh1993returns}, industries \citep{moskowitz1999industries}, and of size and value portfolios \citep{lewellen2010skeptical}. This application is similar to \cite{kelly2020principal} for reproducibility.

To construct a momentum signal, we compute for each asset the lagged one-month return defined as the simple average return over the previous 21 trading days.  Similar to \cite{asness2019quality}, \cite{freyberger2020dissecting}, \cite{kozak2020shrinking} and \cite{kelly2020principal}, we rank the momentum signals across the assets from 1 to $N$, dividing the ranks by the number of assets, and then centering the normalized ranks to map the signals into the range $[-0.5, 0.5]$. This provides us with a set of dollar-neutral signals that are insensitive to outliers for which we further divide by the sum of their absolute values. This keeps the gross exposure (that is, the sum of the absolute amount of long and short positions) fixed since otherwise doubling the number of assets at any time $t$ will result in signals that are two times more aggressive even though the investment opportunities remain the same. We assume \$1 of capital is invested to \$1 of long and short positions.

We collect the individual momentum values of the $N$ assets to yield a predictive signal $x_t$ for the subsequent monthly returns. Additionally, we impose a one-day buffer between the constructed signals and the subsequent returns to limit the effects of illiquidity from driving our results and to bring our backtest simulations closer to being tradeable practice.\footnote{The challenge working with the equity portfolios from Kenneth French’s website is that they contain small illiquid stocks. Moreover, given that we are using daily returns, which are close-to-close returns, asynchronous trading at the end of the day may arise; see, for example, \cite{lo1990econometric}. Consequently, this may induce some lead-lag relationships among the stocks, which can lead to autocorrelation in the portfolio returns and spurious correlation estimates due to the \cite{epps1979comovements} effect. The latter problem can be particularly acute given our use of multivariate techniques in the estimation of portfolios. Therefore, we employ a one-day buffer along with monthly returns to mitigate the effects that asynchronous or infrequent trading can have on the portfolio returns. }

\subsection{Candidate Portfolios}

Given the time series panels of asset returns and signals, we consider the following portfolios in our study:
\begin{itemize}[label={\tiny\raisebox{1ex}{\textbullet}}]
    \item \textbf{CP2}: Our proposed portfolio contruction methodology based on Equation \eqref{eq:optw}.
    \item \textbf{MVO}: The mean-variance optimization portfolio of \cite{markowitz}.
    \item \textbf{PP2}: The Principal Portfolios of \cite{kelly2020principal}.
    \item \textbf{UNI}: The univariate factor where the weights are the signals.
\end{itemize}
The competing portfolios have been chosen because they can be subsumed in our proposed method and hence, serve as natural benchmarks for us to determine where the contribution to any improved performance comes from. The suffix number attached to the portfolio labels CP and PP indicates the number of managed portfolios that we retain; for example, CP2 means that we choose to keep the leading two most predictable managed portfolios.\footnote{We cross-sectionally demean the returns in the construction of the cross-covariance matrix for PP as suggested in \cite{kelly2020principal} to focus on the cross-sectional differences. However, we choose to ignore otherwise for the other portfolios.} 

The covariance of returns in CP2 and MVO is estimated with the linear shrinkage of \cite{ledoit2004well}. Additionally, we also apply linear shrinkage to the covariance of signals for CP2 but choose a high shrinkage intensity with value $\delta _x = 0.9$ without relying on the asymptotic formula from \cite{ledoit2004well} for this purpose since it was developed for financial returns that are assumed to be independent and identically distributed.

Finally, we renormalize the estimated portfolios so that the sum of the absolute value of their weights equals one. This allows all the portfolios to be comparable in scale. It also implies that the gross exposure for all portfolios is one dollar by construction, that is, we apply one dollar of capital for one dollar of long and short positions. This is sensible for long-short equity hedge fund managers who face institutional constraints such as limits on gross exposure by their prime brokerage.\footnote{ The insights are similar if we renormalize the portfolios to achieve a target level of volatility or return.}

\subsection{Evaluation Methodology}

To evaluate the performance of the different portfolios, we report three main out-of-sample performance measures: the average cumulative out-of-sample returns, the standard deviation of the out-of-sample returns, and the Sharpe ratio defined as the ratio of the average returns to the standard deviation of returns. For ease of interpretability, all performance measures are annualized with 12 trading `months'. The Sharpe ratio is computed with respect to the actual returns (as opposed to returns in excess of the risk-free rate) since we believe it is more relevant in our context where the portfolios are formed solely on the basis of risky assets. 

We also report three additional performance measures based on the out-of-sample returns in excess of a 6-factor benchmark; that is, the 5 factors from \cite{fama2015five} augmented with out-of-sample returns from UNI. We compute the Jensen's alpha, beta to the out-of-sample UNI returns, and information ratio from a 6-factor regression model. This is done for all portfolios except for UNI. The alpha and information ratio are annualized with 12 trading `months'. We also provide the t-statistics of the Sharpe ratio and information ratio, which are computed with approximate standard errors from \cite{lo2002statistics}.

Additionally, we report the following portfolio weight statistics averaged over the 578 trading months: turnover, proportional leverage (computed as the fraction of negative weights), the sum of negative weights, and the minimum and maximum weight. Note that these statistics are not our primary focus since our proposed method is not optimized to account for these measures. Nevertheless, they are provided to give a better overview of the different methods.

\subsection{Application of Canonical Portfolio Analysis}

We can use the results of Proposition \ref{prop:return2singvals} to identify the sources of portfolio returns by estimating the canonical portfolios that contribute most to its profitability. To this end, we will consider the FF25, ME/OP25, and ME/INV25 test assets for this purpose. The FF25 is widely studied amongst academics and it allows us to check if our method produces the expected results. Moreover, we know from \cite{lewellen2010skeptical} that the size and value portfolio returns have a strong factor structure explained mostly by the three-factor model \cite{fama1993common}. \cite{KozakNagelSantosh2018JF} finds that retaining the first three principal components extracted from FF25 closely reproduces the Fama and French three-factor model. Thus, this observation could potentially be exploited in our method.

The left panels of Figure \ref{fig:svals} show the squared sample canonical correlations of the demeaned conditional portfolios matrix for the three test assets. It is also overlaid with the squared canonical correlations generated from randomly permuting the signals observations for each asset and repeating a similar exercise.\footnote{Randomly shuffling the time series helps generate a null distribution. It is a useful heuristic to determine the importance of each canonical correlation relative to a random benchmark compared to a formal statistical test of significance from \cite{yang2015independence}.} The squared canonical correlations are averaged over the rebalancing dates and ordered from the smallest to the largest. From this in-sample analysis, we can see that the leading two squared sample canonical correlations have values larger than their pseudo-random generated counterparts and that the top one `sticks out' and extends beyond the value of one.\footnote{The values of (squared) regularized canonical correlations can exceed one unlike its sample-based analog, which is constrained to the interval (0,1).} Moreover, the spacing between the sample canonical correlations is wider at the top end of the spectrum and more uniform at the bottom end of the spectrum. 

We contrast these findings with the right panels of Figure \ref{fig:svals}, which shows the out-of-sample returns of each canonical portfolio computed as the product of the in-sample and out-of-sample canonical correlation, accompanied by their standard error. Not surprisingly, we find a performance deterioration in the out-of-sample performance due to the bias in the in-sample predictions. Nevertheless, there is some coherence between the in-sample predictions and the out-of-sample returns in that the top canonical portfolio possesses the most realized returns followed by the second one, while the bottom ones are close to zero. 

Given the prominence of the leading canonical correlation, we plot its corresponding weights. The left panels of Figure \ref{fig:weights} show the top canonical portfolios (that is, the canonical directions) averaged over the rebalancing dates. To ensure the signs of the canonical portfolios are consistent across time, we flip the sign of the canonical portfolios at any given rebalancing date if its cosine similarity with the canonical portfolios obtained from the previous rebalancing date is negative. While there is not any clear pattern of trades that we can immediately discern, it nonetheless differs from other research findings vis-\'{a}-vis PCA; for example, \cite{kelly2020principal} finds that the top principal component of a symmetrized cross-covariance matrix goes long (short) on big (small) equity portfolios, and long (short) on value (growth) portfolios. The right panels of Figure \ref{fig:weights} provide the final weights invested in each asset.

\subsection{Empirical Results}

Table \ref{tab:performance} summarizes the performance of the various portfolio methods for the different test assets. Restricting our attention to the FF25 column, we see that CP2 has an average return, which is lower than that of PP2 (2.20\% versus 3.76\%) but has much lower volatility (2.18\% versus 5.85\%).\footnote{Note that the magnitudes of the returns for all portfolios appear to be low in comparison to those published in hedge fund return indices. This is expected since the gross exposures of our portfolios are all constrained to one. In practice, one would typically apply a leverage factor greater than one in order to magnify the returns.} Altogether, this translates into Sharpe ratio of $1.01$ which is a 50\% improvement over the closest competitor, PP2, at $0.66$. The alpha of the CP2 is 1.84\% and has a low beta to the univariate factor return of 0.15. Adjusting the alpha by the idiosyncratic volatility of 2.04\% gives an information ratio of 0.9 for CP2. A similar conclusion holds for the ME/OP25 and ME/INV25 test assets, although we observe a deterioration in the Sharpe ratio for all portfolios in larger-sized test assets possibly due to estimation errors. CP2 also underperforms PP2 in the ME/INV100 dataset.   

% The benefits of exploiting the joint time-variation of the entire return distribution using our method is apparent. 

Table \ref{tab:weightstat} describes the distribution of the portfolio weights of the estimated portfolios. In the FF25 column, we see that CP2 has the lowest turnover. This is interesting given that we make no effort to control the trajectory of the weights. The average sum of negative weights in the CP2 is $-0.496$ and the average proportional leverage in CP2 is less than 0.5 indicating a slight tilt towards long positions. The weights are CP2 appear to be the most dispersed but do but are relatively not extreme. We draw a similar conclusion for the other test assets.

It would be interesting to investigate where the performance improvement of CP2 comes from: Is better at selecting assets that have historically performed well on average or at varying the positions in the assets dynamically? One way to discern between the two possible explanations is to decompose the portfolio returns into the following components:
\begin{align}\label{eq:statdyn}
    \mathbb{E}[w_t ' r_{t+1}] =  \underbrace{\mathbb{E}[w_t]'\mathbb{E}[r_{t+1}]}_{\text{static}} + \underbrace{\mathsf{Cov}(r_{t+1}, w_t)}_{\text{dynamic}} .
\end{align}
Intuitively, static bets refer to the investor's ability to get the long-run allocations right. On the other hand, the covariance between the signal and subsequent returns refers to an investor's ability to tactically `time' the market, this is, to accurately identify movements in the assets and gain exposure to those assets accordingly.

Table \ref{tab:decom} reports the results of this decomposition, where we estimate the static and dynamic components of returns according to Equation \eqref{eq:statdyn} with their corresponding sample analogs. Panels A and B show that both the static and dynamic components contribute positively to the overall performance of the portfolio. Panel C of Table \ref{tab:decom} breaks down the share of total returns due to taking dynamic bets for CP2 and P2. It shows that on average the returns of CP2 and PP2 come from taking dynamic bets. Therefore, the performance of our proposed method does not come from its ability to long (short) more highly performing (underperforming) assets; it must emanate from its ability to tactically time the market.

Finally, we also isolate the portfolio returns due to the long and short legs of the trade. We write $r^w_t = l((r^w_t)_+ - (r^w_t)_-)$ where $(r^w_t)_+$ is the return on the long leg and $(r^w_t)_-$ is the return on the short leg, with weights of both trade legs normalized to sum to one. Here, $l$ denotes the leverage of the long-short portfolio, but since our estimated portfolios have a unit gross exposure by construction, $l$ is close to one. Panel D and E of Table \ref{tab:decom} present the statistics of the long and short legs, respectively. For CP2, the average return on the long leg is 8.01\% and the short leg is 5.71\%. This indicates that the profits of our portfolio come from the long side of the trades.

\subsection{Robustness Checks}

In this section, we inspect whether the outperformance of our proposed portfolio construction methodology is robust to different revisions in the current empirical set-up. In particular, we will be interested in results based on (1) subsample period, (2) forecast horizon (3) shrinkage in the signal covariance, (4) momentum lookback window, and (5) `approximate' versus `actual' portfolio policy.

\subsubsection{Sub-Period Analysis}

In this section, we check if there are any peculiar subsample effects that may drive the performances of our proposed scheme. We divide the out-of-sample period into four roughly equally-sized subsamples of 144 months (that is, 12 trading years) each: (1) 1986-09-25 to 1986, (2) 1986 to 1998, (3) 1998 to 2010, and (4) 2010 to 2022. Then we perform the same procedure in each subsample. The results are provided in Table \ref{tab:subsample}. 

Generally, the performance of all portfolios appears stronger in the earlier periods of the sample but poorer in the recent decade. We see that the outperformance of CP2 over the competing portfolios is consistent over time for FF25 and ME/INV25. However, CP2 underperforms PP2 in (1) ME/INV100 for most of the subsamples, and (2) FF100, ME/OP25, and ME/OP100 in the earlier subsample.

\subsubsection{Forecast Horizon}

To make a forecast of subsequent returns, we have used a horizon length of 21 trading days, which roughly corresponds to one month. Given our use of non-overlapping observations between the signals and subsequent returns, this implies that the data are sampled on a monthly basis. which also corresponds to the frequency with which we rebalance our portfolios. We now change the forecast horizon from 21 days to 1, 5, and 10 days. Each of these forecast horizons covers an investment period from (1) 01-06-1965 to 10-21-2022, (2) 12-06-1966 to 10-21-2022, and (3) 06-12-1969 to 10-21-2022; this provides us with $14{,}550$ daily, $2{,}814$ weekly, and $1{,}347$ fortnightly out-of-sample returns, respectively. 

Table \ref{tab:rebalfreq} shows that the annualized Sharpe ratio of all portfolios tends to be better at shorter holding periods. Barring the potential side effects of illiquidity in daily returns, this suggests that the conditioning information becomes more relevant as the holding period decreases. This makes sense since the conditional portfolios become more reactive to changes in the states of the market. Overall, the ranking of the methods remains similar relative to Table \ref{tab:performance} with the exception of the 1-day horizon, where CP2 is the best performer for all test assets.

\subsubsection{Shrinkage Intensity}

We examine the effect of shrinkage on the covariance of signals in our proposed method. Our default choice in the analysis was $\delta _x = 0.9$. The annualized Sharpe ratios for CP2 corresponding to shrinkage values of $\delta_x \in \{0,0.1,0.5,0.8,1\}$  are presented in Table \ref{tab:shrinkage} over different subsamples. 

We see that there is no specific level of shrinkage that provides consistent outperformance for all test assets. This indicates that the shrinkage intensity $\delta _x$ is time-varying in nature. Shrinkage values greater than or equal to 0.5 appears to work well across different subsamples for larger-sized equity portfolios up until the recent decade, where there is some benefit of using more sample information from the signal correlations in datasets ME/OP100 and ME/INV100. 

\subsubsection{Momentum Lookback Window}

The 21-day momentum signal that we used as a base case is perceived to be a relatively `fast' signal, reacting quickly to changes in market conditions. We now consider using momentum signals computed with a longer lookback window of sizes 42 days, 63 days, 84 days, 126 days, and 252 days. The remaining details remain similar.

Table \ref{tab:momspeed} demonstrates that CP2 generally outperforms the other portfolios in terms of annualized Sharpe ratio for different momentum signals and different test assets. Although we observe some performance degradation for all portfolios as we increase the lookback size from 21 days to 126 days, the ranking of the methods remains similar to Table \ref{tab:performance}. There is, however, a significant performance gain in CP2 not seen in the other portfolios as we extend the momentum lookback size from 126 days to 252 days. This is comforting since momentum signals with longer lookback horizons tend to have lower turnover.

\subsubsection{Approximate Versus Actual Solution}

We have worked with the `approximated' mean-variance problem from Proposition \ref{prop:partialobj} throughout this paper since it helped us to simplify the analysis. Given that we also have a closed-form expression to the `full' mean-variance problem from  Theorem \ref{thm:main} in Proposition \ref{prop:fullobj}, we now check if it offers any practical benefits over our approximate solution. The only difference between both solutions is in terms of how the canonical correlations enter into the reconstruction of the optimal portfolio policy. The full solution essentially applies a nonlinear adjustment to the canonical correlations, while no adjustment takes place in the approximated solution. Figure \ref{fig:approx} shows that both solutions are similar for small canonical correlation values but for large canonical correlation values, the conservative behavior of the optimal strategy is reinforced by downweighting its influence. This makes sense because we are taking into account more terms that affect the risk profile of the strategy.

From Table \ref{tab:approx}, we see that our approximated formula generally performs better than the full formula across different test assets. One reason for this observation is that the nonlinear adjustment of the canonical correlations may be too conservative relative to the unadjusted one. For example, a canonical correlation value of one gets reduced by half through the nonlinear adjustment. This holds mechanically irrespective of the data. Consequently, this behavior may inadvertently under-leverage the leading two canonical portfolios, which are the most profitable streams. 

\subsection{Extension: Two-Signal Case}
Our canonical portfolios modeling framework is flexible enough to accommodate multiple signals. In this empirical exercise, we expand our signal vector to include two momentum signals of different lookback windows; one with a 21-day lookback window, and another with a 252-days lookback window. This expanded $2N$-dimensional signal vector will serve as input for both CP and PP. On the other hand, we assume that both MVO and UNI take in an equal-weighted average of the two momentum signals as inputs since these methods do not have an `optimal' way of blending different signals together in a single stage. The rest of the empirical setup remains unchanged.

From Table \ref{tab:twosig}, we see that the Sharpe ratio of all portfolios, with the exception of MVO, generally improves in comparison to the one signal case. More importantly, the outperformance of CP2 over competing methods continues to hold up for different test assets. 

Turning our attention to Table \ref{tab:twosigswgtstat}, we see that the turnover is reduced for all portfolios compared to the base case. This observation can be attributed to the inclusion of the signal with a long lookback window, which tends to have a lower turnover than one with a short lookback window. This is appealing since we can expect the performance after factoring in transaction costs to be better than using a single conditioning variable. We also observe lower proportional leverage throughout the portfolios indicating a tilt toward long positions. The weights also appear to be less dispersed than in the base case.

\section{Conclusion}\label{sec:conclusion}

In this paper, we provide a novel framework for portfolio managers and academics to conceptualize the optimal asset and signal combination problem with canonical correlation analysis (CCA). Our contribution can be summarized as follows. First, we recast the original investment problem of \cite{brandt2006dynamic} into a tractable one that allows us to derive an optimal portfolio policy that is applicable to large cross-sectional financial applications. Our portfolio policy is able to ingest multiple return-predictive signals, and account for cross-predictability and correlations in both the returns and signals. All of these properties are achieved by solving the portfolio selection problem in a single stage. Second, we attempted to lift the veil of complexity from our large-dimensional investment problem through a novel application of CCA. In particular, CCA breaks down the correlations of all of the asset returns and signals into independent long–short managed portfolios, which we term as \textit{canonical portfolios}. Each of the canonical portfolios can be ranked from the one with the smallest correlation to the one with the highest; the canonical portfolios with the highest predictable returns get scaled up the most.

Having established the theoretical contents of our method, we bring it to the empirical test. We ran backtest simulations on Fama-French equity sorted portfolios with a momentum signal. Our findings indicate that our proposed method consistently outperforms natural benchmarks. The performance of our method further improves when the analysis is extended to two momentum signals of different lookback windows. These results are made possible by introducing regularization techniques to overcome estimation errors and exploiting the most predictable dimensions of the data.

In terms of future work, our proposed modeling framework is not set in stone and is flexible enough to accommodate further improvements. There are several potential avenues for research, such as incorporating nonlinearities or regime-switching into the modeling process. Another interesting avenue to explore would be to apply these ideas to develop a test for asset pricing models or to form portfolios on different asset classes such as individual stocks portfolios, fixed-income portfolios, currency portfolios, and so forth. Furthermore, extending the framework to incorporate transaction costs will be pursued in subsequent work. Last but not least, recasting the portfolio selection problem into a CCA framework enables researchers to leverage the insights and techniques from the rich literature of CCA that has been expanded by developments in machine learning.

\newpage
% \printbibliography
% https://www.economics.utoronto.ca/osborne/latex/BIBTEX.HTM
\bibliographystyle{apalike}
\bibliography{papers}  %%% Remove comment to use the external .bib file (using bibtex).

\begin{thebibliography}{}

\bibitem[Anderson, 1951]{anderson1951estimating}
Anderson, T.~W. (1951).
\newblock Estimating linear restrictions on regression coefficients for
  multivariate normal distributions.
\newblock {\em The Annals of Mathematical Statistics}, pages 327--351.

\bibitem[Anderson and Rubin, 1949]{anderson1949estimation}
Anderson, T.~W. and Rubin, H. (1949).
\newblock Estimation of the parameters of a single equation in a complete
  system of stochastic equations.
\newblock {\em The Annals of Mathematical Statistics}, 20(1):46--63.

\bibitem[Anderson and Rubin, 1950]{anderson1950asymptotic}
Anderson, T.~W. and Rubin, H. (1950).
\newblock The asymptotic properties of estimates of the parameters of a single
  equation in a complete system of stochastic equations.
\newblock {\em The Annals of Mathematical Statistics}, pages 570--582.

\bibitem[Ang and Bekaert, 2002]{ang2002international}
Ang, A. and Bekaert, G. (2002).
\newblock International asset allocation with regime shifts.
\newblock {\em The Review of Financial Studies}, 15(4):1137--1187.

\bibitem[Asness et~al., 2019]{asness2019quality}
Asness, C.~S., Frazzini, A., and Pedersen, L.~H. (2019).
\newblock Quality minus junk.
\newblock {\em Review of Accounting Studies}, 24(1):34--112.

\bibitem[Avellaneda and Lee, 2010]{avellaneda2010statistical}
Avellaneda, M. and Lee, J.-H. (2010).
\newblock Statistical arbitrage in the {U}{S} equities market.
\newblock {\em Quantitative Finance}, 10(7):761--782.

\bibitem[Bach and Jordan, 2002]{bach2002kernel}
Bach, F.~R. and Jordan, M.~I. (2002).
\newblock Kernel independent component analysis.
\newblock {\em Journal of Machine Learning Research}, 3:1--48.

\bibitem[Benaych-Georges et~al., 2023]{benaych2023optimal}
Benaych-Georges, F., Bouchaud, J.-P., and Potters, M. (2023).
\newblock Optimal cleaning for singular values of cross-covariance matrices.
\newblock {\em The Annals of Applied Probability}, 33(2):1295--1326.

\bibitem[Black and Litterman, 1992]{black1992global}
Black, F. and Litterman, R. (1992).
\newblock Global portfolio optimization.
\newblock {\em Financial Analysts Journal}, 48(5):28--43.

\bibitem[Brandt, 1999]{brandt1999estimating}
Brandt, M.~W. (1999).
\newblock Estimating portfolio and consumption choice: A conditional {E}uler
  equations approach.
\newblock {\em The Journal of Finance}, 54(5):1609--1645.

\bibitem[Brandt and Santa-Clara, 2006]{brandt2006dynamic}
Brandt, M.~W. and Santa-Clara, P. (2006).
\newblock Dynamic portfolio selection by augmenting the asset space.
\newblock {\em The Journal of Finance}, 61(5):2187--2217.

\bibitem[Cochrane, 2001]{cochrane2001asset}
Cochrane, J.~H. (2001).
\newblock {\em Asset Pricing}.
\newblock Princeton University Press, Princeton.

\bibitem[Connor and Korajczyk, 1986]{connor1986performance}
Connor, G. and Korajczyk, R.~A. (1986).
\newblock Performance measurement with the arbitrage pricing theory: A new
  framework for analysis.
\newblock {\em Journal of Financial Economics}, 15(3):373--394.

\bibitem[Connor and Korajczyk, 1988]{connor1988risk}
Connor, G. and Korajczyk, R.~A. (1988).
\newblock Risk and return in an equilibrium {A}{P}{T}: {A}pplication of a new
  test methodology.
\newblock {\em Journal of Financial Economics}, 21(2):255--289.

\bibitem[Eckart and Young, 1936]{eckart1936approximation}
Eckart, C. and Young, G. (1936).
\newblock The approximation of one matrix by another of lower rank.
\newblock {\em Psychometrika}, 1(3):211--218.

\bibitem[Engle and Granger, 1987]{engle1987co}
Engle, R.~F. and Granger, C.~W. (1987).
\newblock Co-integration and error correction: Representation, estimation, and
  testing.
\newblock {\em Econometrica}, 55(2):251--276.

\bibitem[Epps, 1979]{epps1979comovements}
Epps, T.~W. (1979).
\newblock Comovements in stock prices in the very short run.
\newblock {\em Journal of the American Statistical Association},
  74(366a):291--298.

\bibitem[Fama and French, 1993]{fama1993common}
Fama, E.~F. and French, K.~R. (1993).
\newblock Common risk factors in the returns on stocks and bonds.
\newblock {\em Journal of Financial Economics}, 33(1):3--56.

\bibitem[Fama and French, 2015]{fama2015five}
Fama, E.~F. and French, K.~R. (2015).
\newblock A five-factor asset pricing model.
\newblock {\em Journal of Financial Economics}, 116(1):1--22.

\bibitem[Ferson and Siegel, 2001]{ferson2001efficient}
Ferson, W.~E. and Siegel, A.~F. (2001).
\newblock The efficient use of conditioning information in portfolios.
\newblock {\em The Journal of Finance}, 56(3):967--982.

\bibitem[Firoozye and Koshiyama, 2019]{koshiyama2019avoiding}
Firoozye, N. and Koshiyama, A. (2019).
\newblock Avoiding backtesting overfitting by covariance-penalties: {A}n
  empirical investigation of the ordinary and total least squares cases.
\newblock {\em The Journal of Financial Data Science}, 1(4):63--83.

\bibitem[Firoozye and Koshiyama, 2020]{firoozye2020optimal}
Firoozye, N. and Koshiyama, A. (2020).
\newblock Optimal dynamic strategies on {G}aussian returns.
\newblock {\em Journal of Investment Strategies}, 9(1):23--53.

\bibitem[Freyberger et~al., 2020]{freyberger2020dissecting}
Freyberger, J., Neuhierl, A., and Weber, M. (2020).
\newblock Dissecting characteristics nonparametrically.
\newblock {\em The Review of Financial Studies}, 33(5):2326--2377.

\bibitem[Frost and Savarino, 1986]{frost1986empirical}
Frost, P.~A. and Savarino, J.~E. (1986).
\newblock An empirical {B}ayes approach to efficient portfolio selection.
\newblock {\em Journal of Financial and Quantitative Analysis}, 21(3):293--305.

\bibitem[Golub and Van~Loan, 1980]{golub1980analysis}
Golub, G.~H. and Van~Loan, C.~F. (1980).
\newblock An analysis of the total least squares problem.
\newblock {\em SIAM Journal on Numerical Analysis}, 17(6):883--893.

\bibitem[Grinold, 1989]{grinold1989fundamental}
Grinold, R.~C. (1989).
\newblock The fundamental law of active management.
\newblock {\em The Journal of Portfolio Management}, 15(3):30--37.

\bibitem[Haldane, 1942]{haldane1942moments}
Haldane, J. (1942).
\newblock Moments of the distributions of powers and products of normal
  variates.
\newblock {\em Biometrika}, 32(3/4):226--242.

\bibitem[Hansen and Richard, 1987]{hansen1987role}
Hansen, L. and Richard, S.~F. (1987).
\newblock The role of conditioning information in deducing testable
  restrictions implied by dynamic asset pricing models.
\newblock {\em Econometrica}, 55(3):587--613.

\bibitem[Healy, 1957]{healy1957rotation}
Healy, M. (1957).
\newblock A rotation method for computing canonical correlations.
\newblock {\em Mathematics of Computation}, 11(58):83--86.

\bibitem[Hotelling, 1933]{hotelling1933analysis}
Hotelling, H. (1933).
\newblock Analysis of a complex of statistical variables into principal
  components.
\newblock {\em Journal of Educational Psychology}, 24(6):417.

\bibitem[Hotelling, 1936]{hotelling1936relations}
Hotelling, H. (1936).
\newblock Relations between two sets of variates.
\newblock {\em Biometrika}, 28(3/4):321--377.

\bibitem[Isserlis, 1918]{isserlis1918formula}
Isserlis, L. (1918).
\newblock On a formula for the product-moment coefficient of any order of a
  normal frequency distribution in any number of variables.
\newblock {\em Biometrika}, 12(1/2):134--139.

\bibitem[Jegadeesh, 1990]{jegadeesh1990evidence}
Jegadeesh, N. (1990).
\newblock Evidence of predictable behavior of security returns.
\newblock {\em The Journal of Finance}, 45(3):881--898.

\bibitem[Jegadeesh and Titman, 1993]{jegadeesh1993returns}
Jegadeesh, N. and Titman, S. (1993).
\newblock Returns to buying winners and selling losers: Implications for stock
  market efficiency.
\newblock {\em The Journal of Finance}, 48(1):65--91.

\bibitem[Jobson and Korkie, 1980]{jobson1980estimation}
Jobson, J.~D. and Korkie, B. (1980).
\newblock Estimation for {M}arkowitz efficient portfolios.
\newblock {\em Journal of the American Statistical Association},
  75(371):544--554.

\bibitem[Johansen, 1988]{johansen1988statistical}
Johansen, S. (1988).
\newblock Statistical analysis of cointegration vectors.
\newblock {\em Journal of Economic Dynamics and Control}, 12(2-3):231--254.

\bibitem[Johansen, 1991]{johansen1991estimation}
Johansen, S. (1991).
\newblock Estimation and hypothesis testing of cointegration vectors in
  gaussian vector autoregressive models.
\newblock {\em Econometrica}, 59(6):1551--1580.

\bibitem[Johansen, 1995]{johansen1995likelihood}
Johansen, S. (1995).
\newblock {\em Likelihood-based Inference in Cointegrated Vector Autoregressive
  Models}.
\newblock Oxford: Oxford University Press.

\bibitem[Kelly et~al., 2022]{kelly2020principal}
Kelly, B., Malamud, S., and Pedersen, L.~H. (2022).
\newblock Principal portfolios.
\newblock {\em The Journal of Finance}.

\bibitem[Kelly and Pruitt, 2013]{kelly2013market}
Kelly, B. and Pruitt, S. (2013).
\newblock Market expectations in the cross-section of present values.
\newblock {\em The Journal of Finance}, 68(5):1721--1756.

\bibitem[Kozak et~al., 2018]{KozakNagelSantosh2018JF}
Kozak, S., Nagel, S., and Santosh, S. (2018).
\newblock Interpreting factor models.
\newblock {\em The Journal of Finance}, 73(3):1183--1223.

\bibitem[Kozak et~al., 2020]{kozak2020shrinking}
Kozak, S., Nagel, S., and Santosh, S. (2020).
\newblock Shrinking the cross-section.
\newblock {\em Journal of Financial Economics}, 135(2):271--292.

\bibitem[Lai et~al., 2011]{lai2011mean}
Lai, T.~L., Xing, H., and Chen, Z. (2011).
\newblock Mean--variance portfolio optimization when means and covariances are
  unknown.
\newblock {\em The Annals of Applied Statistics}, 5(2A):798--823.

\bibitem[Ledoit and Wolf, 2004a]{ledoit2004honey}
Ledoit, O. and Wolf, M. (2004a).
\newblock Honey, {I} shrunk the sample covariance matrix.
\newblock {\em The Journal of Portfolio Management}, 30(4):110--119.

\bibitem[Ledoit and Wolf, 2004b]{ledoit2004well}
Ledoit, O. and Wolf, M. (2004b).
\newblock A well-conditioned estimator for large-dimensional covariance
  matrices.
\newblock {\em Journal of Multivariate Analysis}, 88(2):365--411.

\bibitem[Ledoit and Wolf, 2017]{ledoit2017nonlinear}
Ledoit, O. and Wolf, M. (2017).
\newblock Nonlinear shrinkage of the covariance matrix for portfolio selection:
  Markowitz meets {G}oldilocks.
\newblock {\em The Review of Financial Studies}, 30(12):4349--4388.

\bibitem[Lehmann and Modest, 1988]{lehmann1988empirical}
Lehmann, B.~N. and Modest, D.~M. (1988).
\newblock The empirical foundations of the arbitrage pricing theory.
\newblock {\em Journal of Financial Economics}, 21(2):213--254.

\bibitem[Levy and Markowitz, 1979]{levy1979approximating}
Levy, H. and Markowitz, H.~M. (1979).
\newblock Approximating expected utility by a function of mean and variance.
\newblock {\em The American Economic Review}, pages 308--317.

\bibitem[Lewellen et~al., 2010]{lewellen2010skeptical}
Lewellen, J., Nagel, S., and Shanken, J. (2010).
\newblock A skeptical appraisal of asset pricing tests.
\newblock {\em Journal of Financial Economics}, 96(2):175--194.

\bibitem[Lo, 2002]{lo2002statistics}
Lo, A.~W. (2002).
\newblock The statistics of {S}harpe ratios.
\newblock {\em Financial Analysts Journal}, 58(4):36--52.

\bibitem[Lo and MacKinlay, 1990]{lo1990econometric}
Lo, A.~W. and MacKinlay, A.~C. (1990).
\newblock An econometric analysis of nonsynchronous trading.
\newblock {\em Journal of Econometrics}, 45(1-2):181--211.

\bibitem[L{\"u}tkepohl, 1997]{lutkepohl1997handbook}
L{\"u}tkepohl, H. (1997).
\newblock Handbook of matrices.
\newblock {\em Computational Statistics and Data Analysis}, 2(25):243.

\bibitem[Markowitz, 1952]{markowitz}
Markowitz, H. (1952).
\newblock Portfolio selection.
\newblock {\em The Journal of Finance}, 7:77–91.

\bibitem[Markowitz, 1991]{markowitz1991foundations}
Markowitz, H.~M. (1991).
\newblock Foundations of portfolio theory.
\newblock {\em The Journal of Finance}, 46(2):469--477.

\bibitem[Merton, 1972]{merton1972analytic}
Merton, R.~C. (1972).
\newblock An analytic derivation of the efficient portfolio frontier.
\newblock {\em Journal of Financial and Quantitative Analysis},
  7(4):1851--1872.

\bibitem[Meucci, 2009]{meucci2009managing}
Meucci, A. (2009).
\newblock Managing diversification.
\newblock {\em Risk}, pages 74--79.

\bibitem[Michaud, 1989]{michaud}
Michaud, R. (1989).
\newblock The {M}arkowitz optimization enigma: {I}s optimized optimal?
\newblock {\em Financial Analysts Journal}, 45:31–42.

\bibitem[Michaud and Michaud, 2008]{michaud2008efficient}
Michaud, R.~O. and Michaud, R.~O. (2008).
\newblock {\em Efficient Asset Management: A Practical Guide to Stock Portfolio
  Optimization and Asset Allocation}.
\newblock Oxford University Press.

\bibitem[Moskowitz and Grinblatt, 1999]{moskowitz1999industries}
Moskowitz, T.~J. and Grinblatt, M. (1999).
\newblock Do industries explain momentum?
\newblock {\em The Journal of Finance}, 54(4):1249--1290.

\bibitem[Partovi and Caputo, 2004]{partovi2004principal}
Partovi, M.~H. and Caputo, M. (2004).
\newblock Principal portfolios: Recasting the efficient frontier.
\newblock {\em Economics Bulletin}, 7(3):1--10.

\bibitem[Roll and Ross, 1980]{roll1980empirical}
Roll, R. and Ross, S.~A. (1980).
\newblock An empirical investigation of the arbitrage pricing theory.
\newblock {\em The Journal of Finance}, 35(5):1073--1103.

\bibitem[Uurtio et~al., 2017]{uurtio2017tutorial}
Uurtio, V., Monteiro, J.~M., Kandola, J., Shawe-Taylor, J., Fernandez-Reyes,
  D., and Rousu, J. (2017).
\newblock A tutorial on canonical correlation methods.
\newblock {\em ACM Computing Surveys (CSUR)}, 50(6):1--33.

\bibitem[Velu and Reinsel, 2013]{velu2013multivariate}
Velu, R. and Reinsel, G.~C. (2013).
\newblock {\em Multivariate reduced-rank regression: {T}heory and
  applications}, volume 136.
\newblock Springer Science \& Business Media.

\bibitem[Vinod, 1976]{vinod1976canonical}
Vinod, H.~D. (1976).
\newblock Canonical ridge and econometrics of joint production.
\newblock {\em Journal of Econometrics}, 4(2):147--166.

\bibitem[Wachter, 1980]{wachter1980limiting}
Wachter, K.~W. (1980).
\newblock The limiting empirical measure of multiple discriminant ratios.
\newblock {\em The Annals of Statistics}, pages 937--957.

\bibitem[Wick, 1950]{wick1950evaluation}
Wick, G. (1950).
\newblock The evaluation of the collision matrix.
\newblock {\em Physical Review}, 80(2):268--272.

\bibitem[Wold, 1975]{wold1975soft}
Wold, H. (1975).
\newblock Soft modelling by latent variables: {T}he non-linear iterative
  partial least squares ({N}{I}{P}{A}{L}{S}) approach.
\newblock {\em Journal of Applied Probability}, 12(S1):117--142.

\bibitem[Yang and Pan, 2012]{yang2012convergence}
Yang, Y. and Pan, G. (2012).
\newblock The convergence of the empirical distribution of canonical
  correlation coefficients.
\newblock {\em Electronic Journal of Probability}, 17:1--13.

\bibitem[Yang and Pan, 2015]{yang2015independence}
Yang, Y. and Pan, G. (2015).
\newblock Independence test for high dimensional data based on regularized
  canonical correlation coefficients.
\newblock {\em The Annals of Statistics}, 43(2):467--500.

\end{thebibliography}


% \end{thebibliography}

\newpage

\clearpage
\pagestyle{empty}

\appendix
\numberwithin{equation}{section}

\section{Figures and Tables}\label{appendix:figures}

\vspace*{\fill}
\begin{figure}[H] 
\caption{Wachter Law}
\captionsetup{font=footnotesize}
\caption*{This figure displays the limiting spectral density of the sample canonical correlations of \cite{wachter1980limiting} for different ratios $N/T$ and $M/T$. The population cross-covariance matrix is the zero matrix, and so all the population canonical correlation coefficients are equal to zero. }
    \centering
    \includegraphics[width=0.49\textwidth]{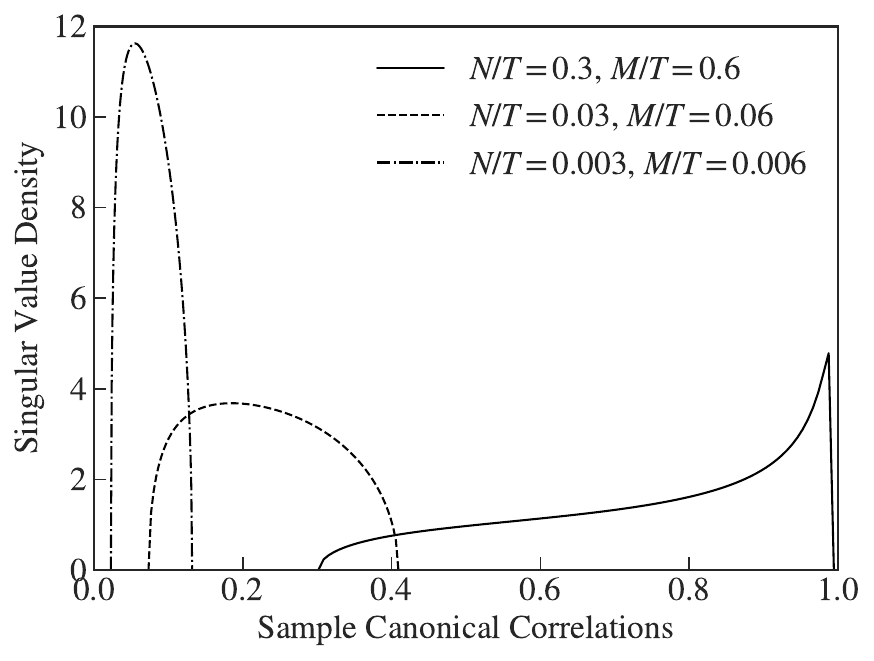}
    \label{fig:lsd}
\end{figure}

\begin{figure}[H] 
\caption{Optimal Singular Value Adjustment Versus Linear Approximation}
\captionsetup{font=footnotesize}
\caption*{This figure displays the optimal adjustment of the $i$th singular value versus an unadjusted $i$th singular value.}
    \centering
    \includegraphics[width=0.49\textwidth]{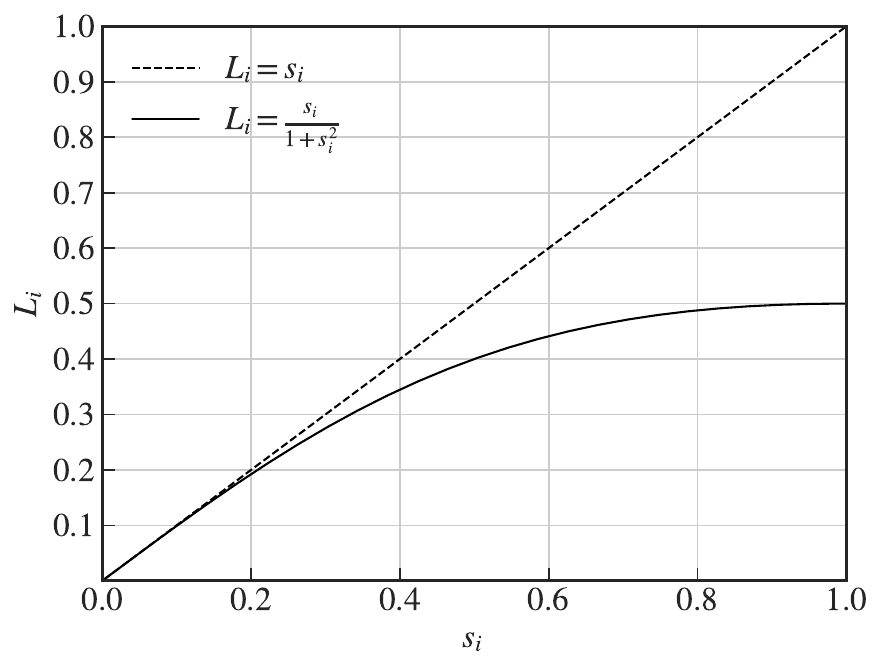}
    \label{fig:approx}
\end{figure}

\vspace{\fill}

\newpage

% \vspace*{\fill}
\begin{figure}[H] 
\caption{In-Sample and Out-of-Sample Canonical Portfolio Returns}
\captionsetup{font=footnotesize}
\caption*{This figure displays the in-sample and out-of-sample returns for each canonical portfolio. The left panel shows the squared in-sample regularized canonical correlation coefficients (black dots) overlaid with their pseudo-random generated counterparts (maroon dots). The right panel shows the out-of-sample returns of each canonical portfolio overlaid with a $\pm 2$ standard error band. Each dot corresponds to values that are averaged over the rebalancing dates. The estimates are obtained by using a 21-day momentum signal to predict subsequent monthly returns and 120 non-overlapping monthly returns. The test assets are the FF25 (top panel), ME/OP25 (middle panel), and MEINV25 (bottom panel). The out-of-sample period is covered from 09-09-1974 until 10-21-2022.}
    \centering
    \begin{subfigure}[b]{\textwidth}
        \caption{25 Size and Book-to-Market}
         \centering
         \includegraphics[width=0.49\textwidth]{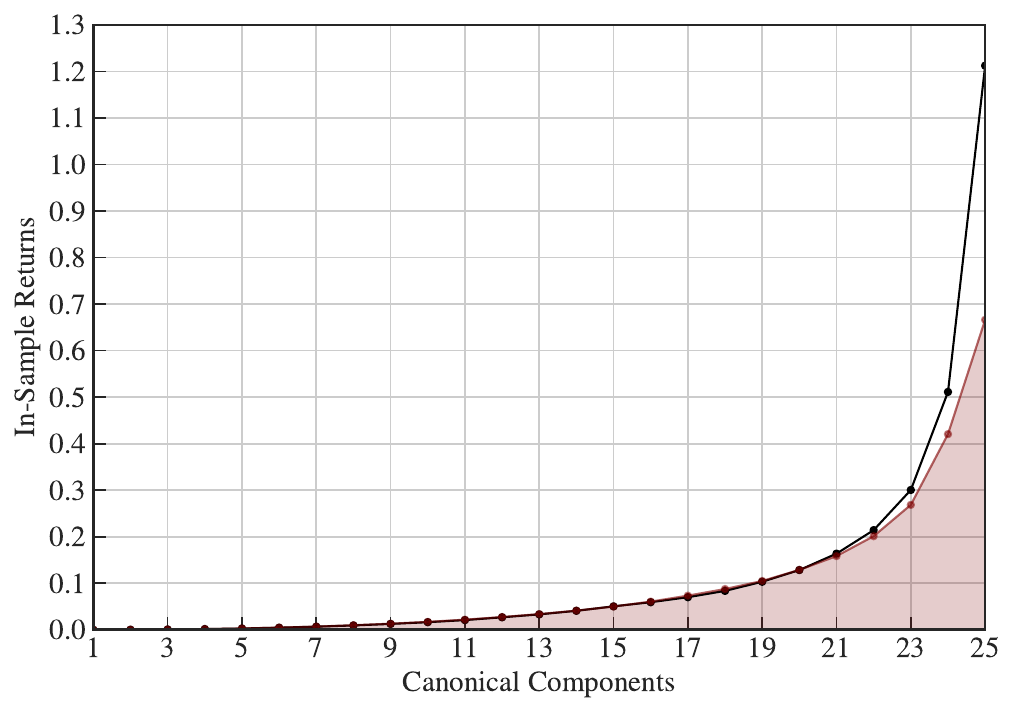}
         \centering
         \includegraphics[width=0.49\textwidth]{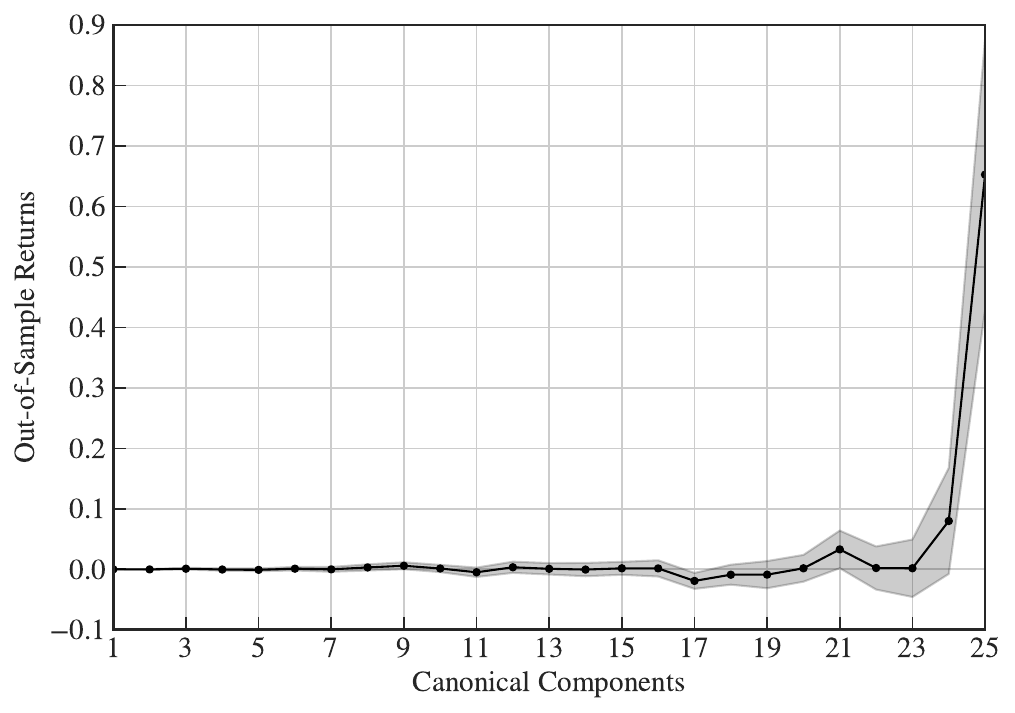}
     \end{subfigure}
     \begin{subfigure}[b]{\textwidth}
        \caption{25 Size and Operating Profitability}
         \centering
         \includegraphics[width=0.49\textwidth]{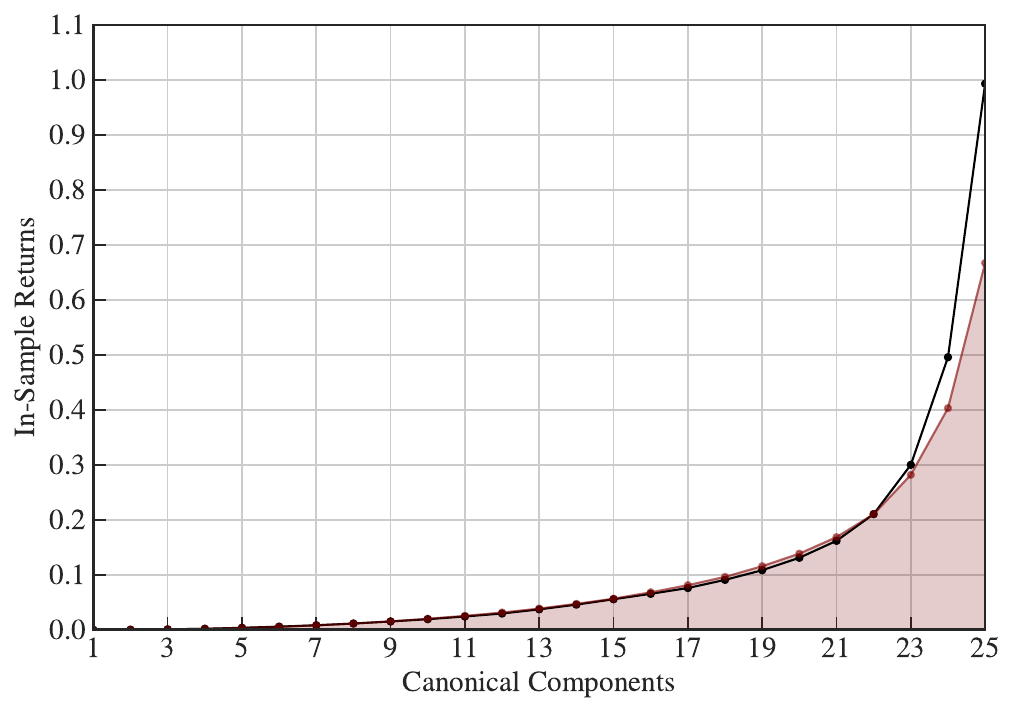}
         \centering
         \includegraphics[width=0.49\textwidth]{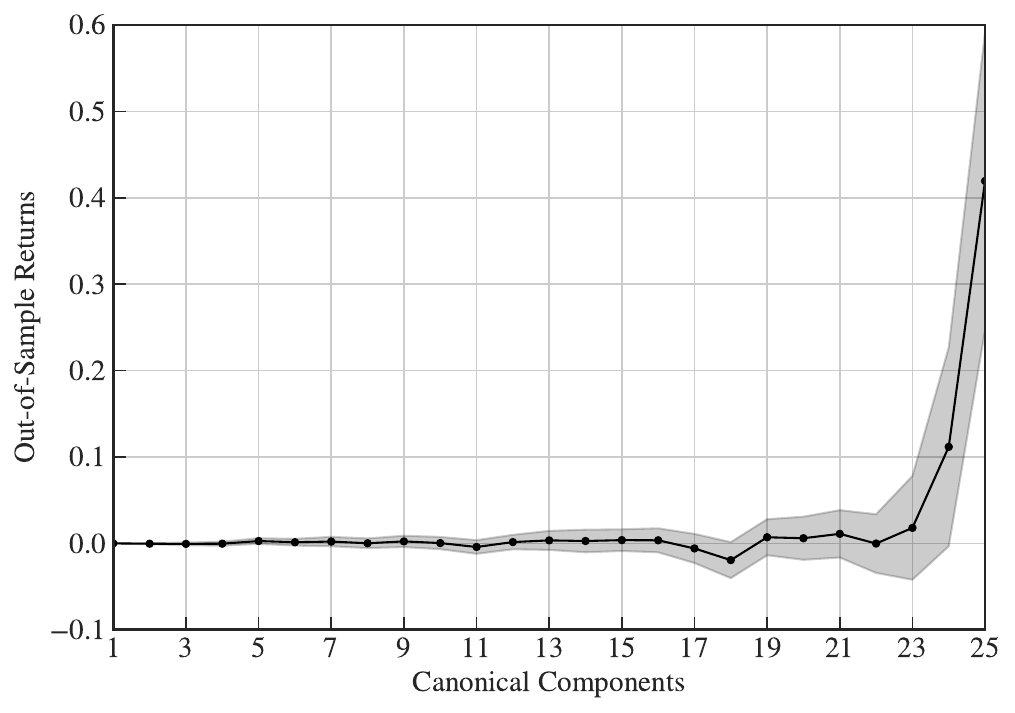}
     \end{subfigure}
     \begin{subfigure}[b]{\textwidth}
         \caption{25 Size and Investment}
         \centering
         \includegraphics[width=0.49\textwidth]{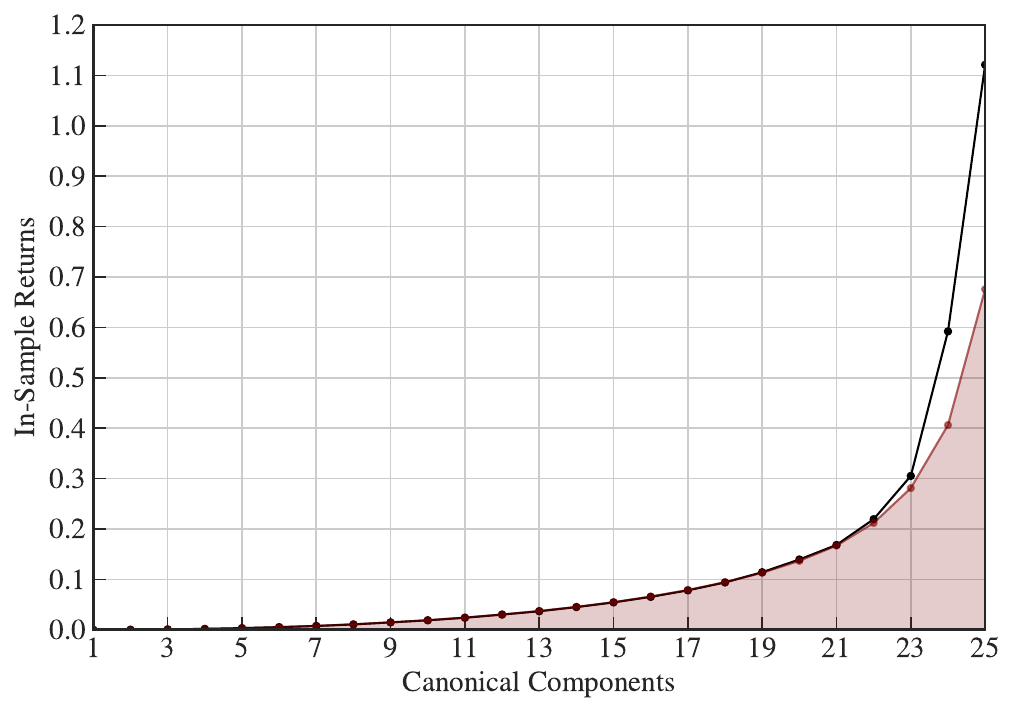}
         \centering
         \includegraphics[width=0.49\textwidth]{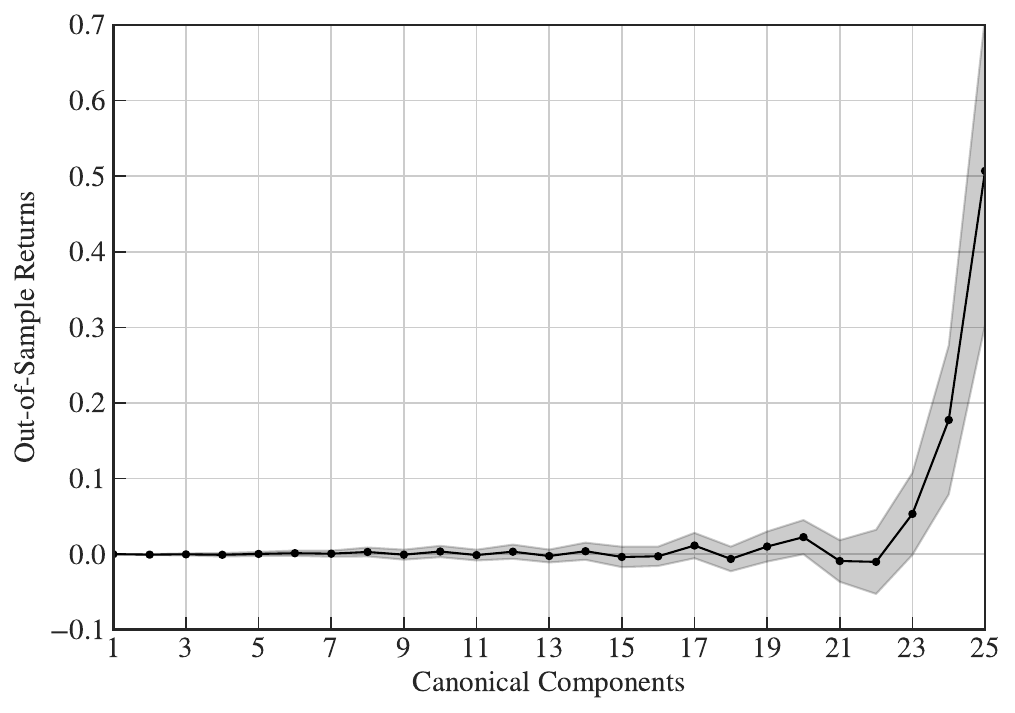}
     \end{subfigure}
    \label{fig:svals}
\end{figure}

\begin{figure}[H] 
\caption{Weights for the Top Canonical Portfolio}
\captionsetup{font=footnotesize}
\caption*{This figure displays the weights of the leading canonical portfolio. The left panel shows the canonical `weights' on the returns (blue) and signals (red). The right panel shows the final portfolio weights invested in each obtained from combining the returns and signals together through their canonical weights and scaling them up with the corresponding canonical correlation. The height of each bar corresponds to values that are averaged over the rebalancing dates, and they are overlaid with a $\pm 2$ standard error band. The estimates are obtained by using a 21-day momentum signal to predict subsequent monthly returns and 120 non-overlapping monthly returns. The test assets are the FF25 (top panel), ME/OP25 (middle panel), and MEINV25 (bottom panel). The out-of-sample period is covered from 09-09-1974 until 10-21-2022.}
    \centering
    \begin{subfigure}[b]{\textwidth}
        \caption{25 Size and Book-to-Market}
         \centering
         \includegraphics[width=0.495\textwidth]{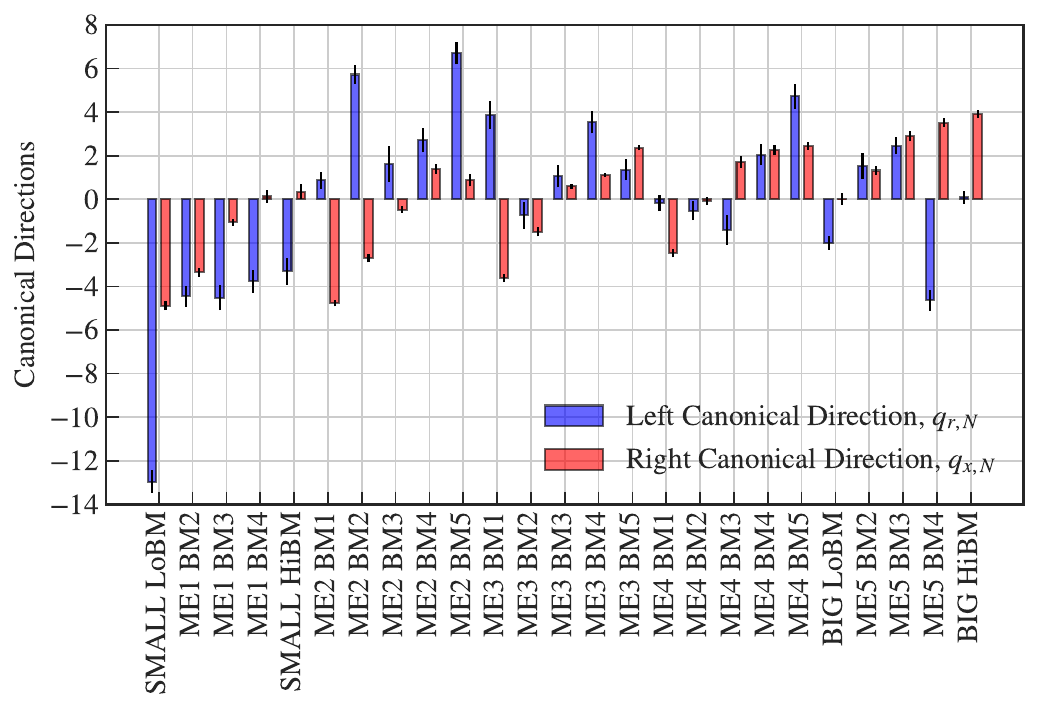}
         \centering
        \includegraphics[width=0.495\textwidth]{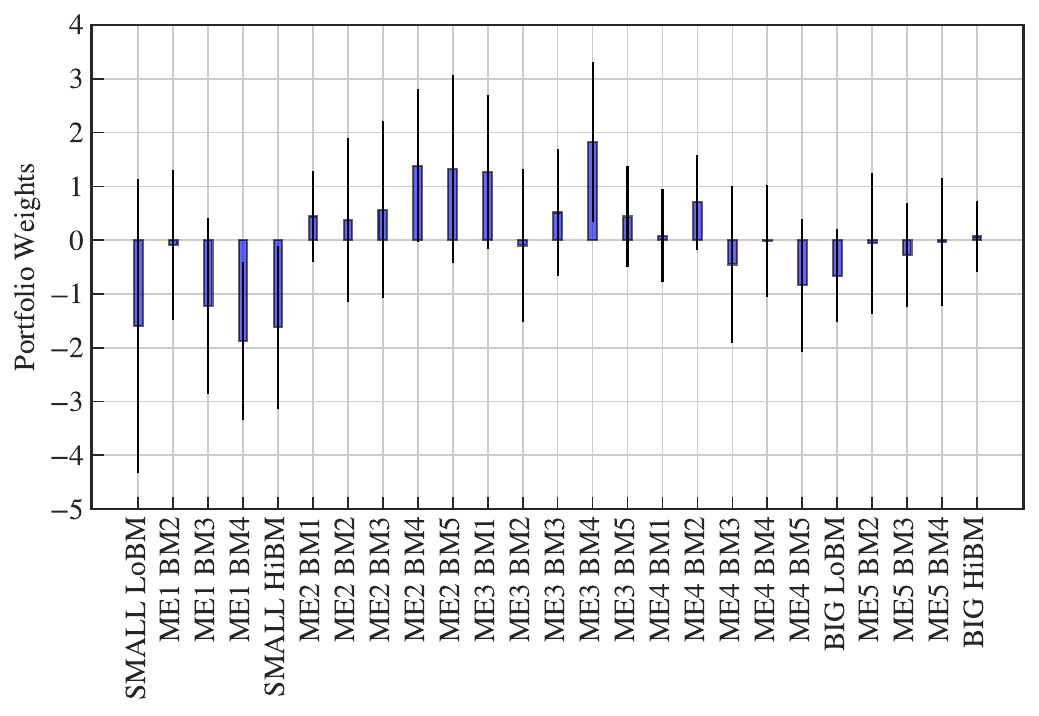}
     \end{subfigure}
     \begin{subfigure}[b]{\textwidth}
        \caption{25 Size and Operating Profitability}
         \centering
         \includegraphics[width=0.495\textwidth]{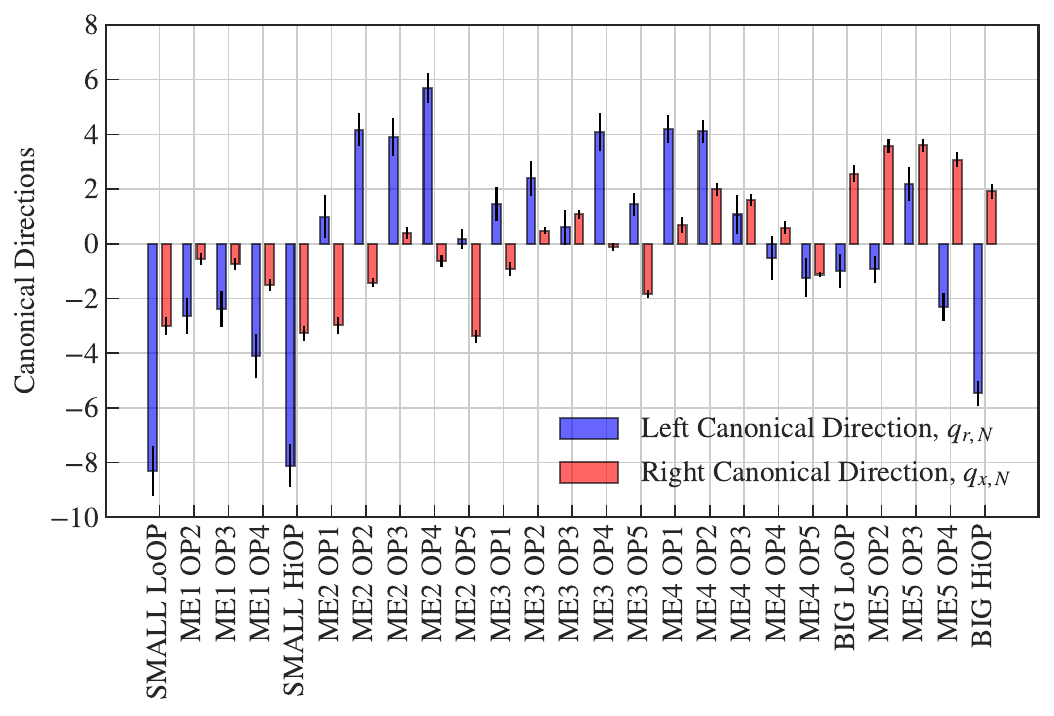}
         \centering
         \includegraphics[width=0.495\textwidth]{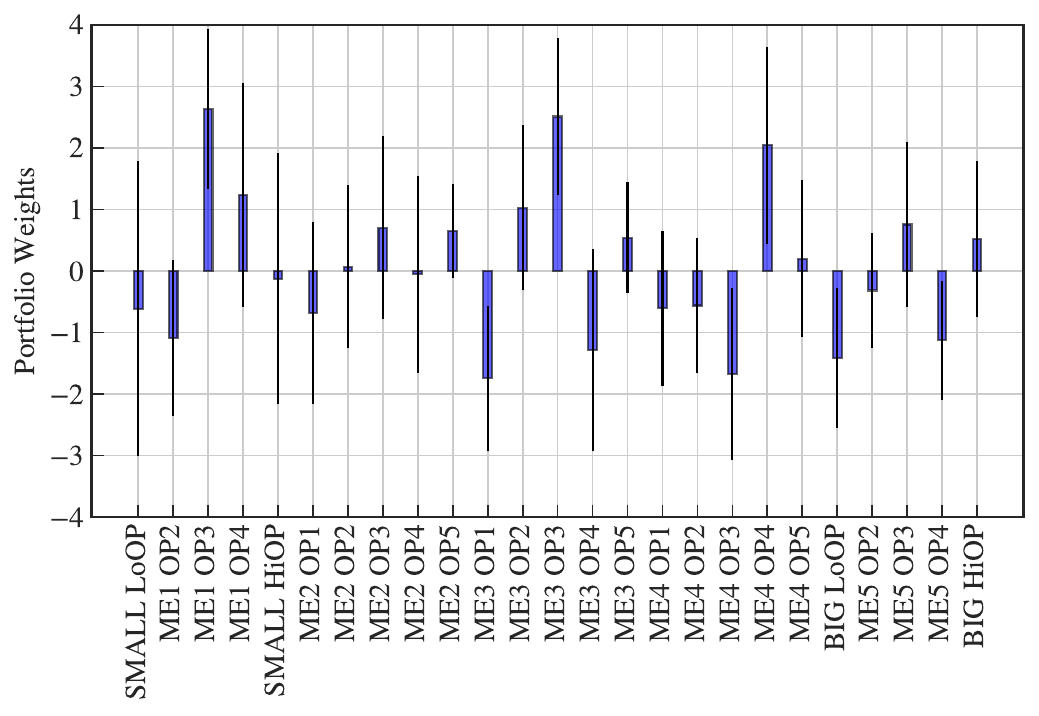}
     \end{subfigure}
     \begin{subfigure}[b]{\textwidth}
         \caption{25 Size and Investment}
         \centering
         \includegraphics[width=0.495\textwidth]{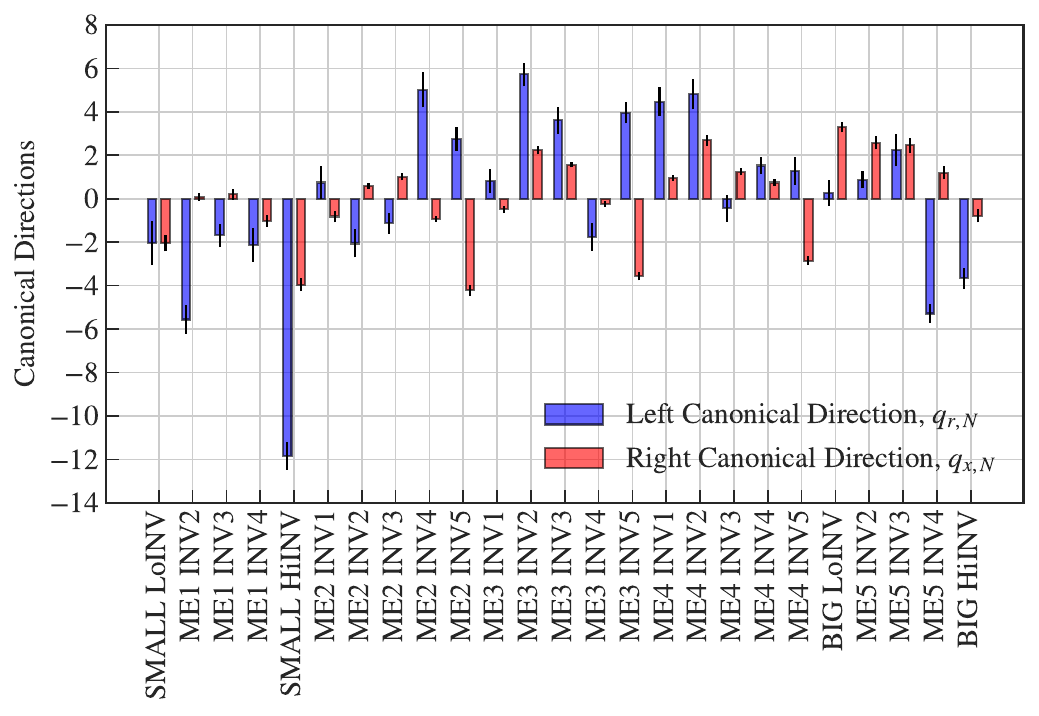}
         \centering
         \includegraphics[width=0.495\textwidth]{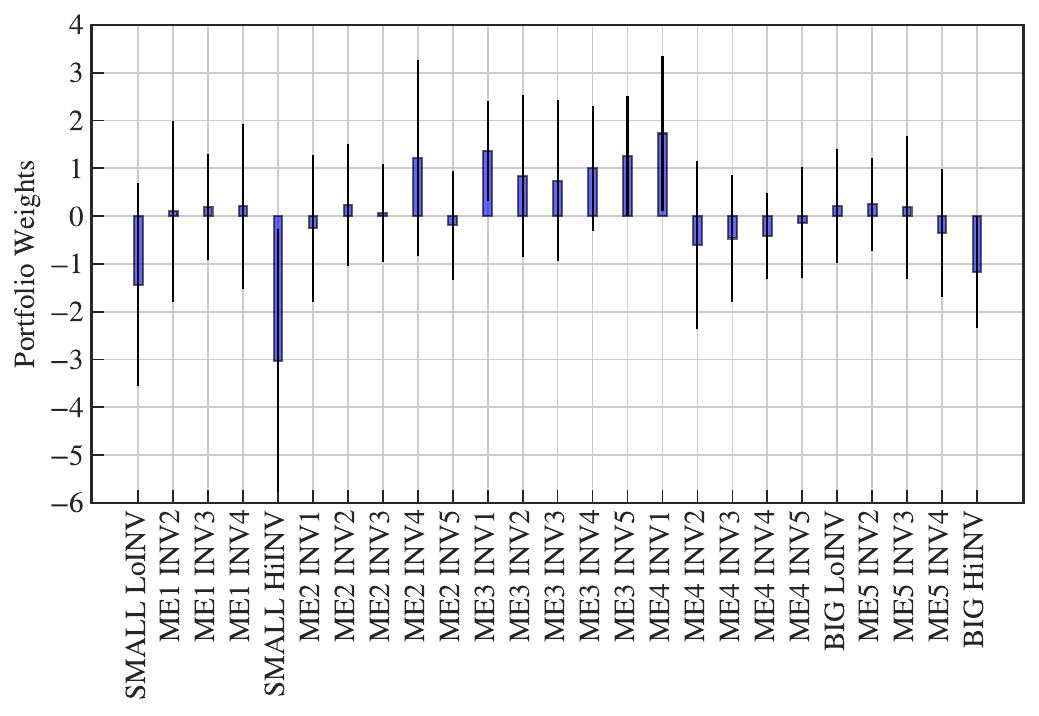}
     \end{subfigure}
    \label{fig:weights}
\end{figure}

\newpage

\vspace*{\fill}
\begin{table}[htbp]
\centering
\renewcommand{\arraystretch}{2.5}
\caption{Summary of PCA and CCA}
    \small
    \captionsetup{font=footnotesize}
    \caption*{This table summarizes the key distinguishing features of PCA and CCA applied to a paired dataset of asset returns $r_{t+1}$ and signals $x_t$. For CCA (PCA), it displays (1) the correlation (variance) maximization objective, (2) the equivalent error minimization objective, (3) the $i$th canonical (principal) direction, and (4) the $i$th correlation (variance) of the corresponding canonical (principal) direction. }
    \centering
    \begin{tabular}{|c|l|c|}
    \hline
    \multirow{ 5}{*}{PCA} & \makecell[l]{(1) Maximize Variance} & \makecell[c]{ $\displaystyle \max _{\substack{q_r, q_x  \\ \| q_r\|_2 = \| q_x\|_2 = 1} } \mathsf{Var}(q_r ' r_{t+1}) + \mathsf{Var}(q_x ' x_{t}) $} \\
    \cline{2-3} 
    & (2) Minimize Error & \makecell[c]{\(\displaystyle \min _{\| q_r\|_2 = \| q_x\|_2 = 1 } \mathbb{E} \big [ \big \| r_{t+1} - q_rq'_r r_{t+1} \big \|^2 _2 \big ] +  \mathbb{E} \big [ \big \| x_{t} - q_xq'_x x_{t} \big \|^2_2  \big ] \)} \\ 
    \cline{2-3} & (3) Principal Direction & \makecell[c]{\(\displaystyle q_{r,i}, q_{x,i} = \argmax _{\substack{\| q_r\|_2 = \| q_x\|_2 = 1 \\q_r'q_{r,j}=0, \text{ } j=1,\ldots, i-1 \\ q_x'q_{x,j}=0, \text{ } j=1,\ldots, i-1 } } \mathsf{Var}(q_r ' r_{t+1}) + \mathsf{Var}(q_x ' x_{t})  \) }\\ 
     \cline{2-3} & (4) Principal Variance & \makecell[c]{\(\displaystyle d_{r,i} = \mathsf{Var}(q_{r,i} ' r_{t+1}), \quad d_{x,i} = \mathsf{Var}(q_{x,i} ' x_{t}) \)} \\ 
    \hline
    \multirow{5}{*}{CCA} & (1) Maximize Correlation & \makecell[c]{\(\displaystyle \max _{\mathsf{Var}(q_r' r_{t+1}) = \mathsf{Var}(q_x' x_{t}) =1} \mathsf{Cov}(q_r ' r_{t+1}, q_x ' x_t) \)} \\
    \cline{2-3}  & (2) Minimize Error & \makecell[c]{\(\displaystyle \min _{\mathsf{Var}(q_r' r_{t+1}) = \mathsf{Var}(q_x' x_{t}) =1} \mathbb{E}\big [\big\| q_r'r_{t+1} - q'_x x_t \big \|^2_2 \big ]\)}  \\
    \cline{2-3} & (3) Canonical Direction &\makecell[c]{ \(\displaystyle q_{r,i}, q_{x,i} = \argmax _{\substack{q_r' \Sigma _r q_r = q_x ' \Sigma _x q_x =1 \\ q_r'q_{r,j}=0 , \text{ } j=1,\ldots, i-1 \\ q_x'q_{x,j} =0, \text{ } j=1,\ldots, i-1} } \mathsf{Cov}(q_r ' r_{t+1}, q_x ' x_t)  \) }\\ 
    \cline{2-3} & (4) Canonical Correlation & \makecell[c]{\(\displaystyle s_i = \frac{\mathsf{Cov}(q_{r,i} ' r_{t+1}, q_{x,i} ' x_t)}{\sqrt{\mathsf{Var}(q_{r,i}' r_{t+1})} \sqrt{\mathsf{Var}(q_{x,i}' x_{t})}}  \) }\\ 
    \hline
    \end{tabular}
\label{tab:lineardim}
\end{table}
\vspace{\fill}

\newpage

\vspace*{\fill}
\begin{table}[H]
    \footnotesize
    \caption{Performance Summary}
    \captionsetup{font=footnotesize}
    \caption*{This table shows the out-of-sample performances (in percentages) of various portfolio selection methods using a 21-day momentum signal. The test assets are the 25 and 100 equity portfolios independently sorted on size and three other factors (book-to-market, operating profitability, and investment). Panels A--G display the average returns, the Sharpe ratio of returns, the alpha with respect to the Fama-French 5-factor regression augmented with the univariate factor (UNI), the beta of the simple factor, the idiosyncratic volatility, and the information ratio. The numbers in parenthesis are the t-statistics for the portfolio's Sharpe ratio and information ratio using the standard error of \cite{lo2002statistics}. The covariances are estimated on a rolling walk-forward basis with $T=120$ non-overlapping monthly returns. All measures are computed with 579 monthly portfolio weights from 09-09-1974 until 10-21-2022.}
    \centering
\begin{tabular}{lcccccc} \toprule
 & FF25 & FF100 & \makecell{ME/OP \\ 25} & \makecell{ME/OP \\ 100} & \makecell{ME/INV \\ 25} & \makecell{ME/INV \\ 100}\\ \cmidrule{1-7}
\multicolumn{7}{l}{Panel A. Average Returns (\%)} \\\cmidrule{1-7}
CP2 & $2.196$ & $1.334$ & $1.741$ & $0.953$ & $1.809$ & $0.782$\\ 
PP2 & $3.762$ & $3.462$ & $3.806$ & $3.490$ & $3.377$ & $3.432$\\ 
MVO & $0.381$ & $0.188$ & $0.358$ & $-0.097$ & $0.365$ & $0.097$\\ 
UNI & $2.268$ & $1.787$ & $2.261$ & $1.477$ & $2.141$ & $1.473$\\ 
\cmidrule{1-7}
\multicolumn{7}{l}{Panel B. Standard Deviation of Returns (\%)} \\\cmidrule{1-7}
CP2 & $2.175$ & $1.366$ & $2.048$ & $1.260$ & $2.057$ & $1.248$\\ 
PP2 & $5.853$ & $5.690$ & $4.912$ & $4.931$ & $4.642$ & $4.433$\\ 
MVO & $1.857$ & $1.366$ & $1.685$ & $1.301$ & $1.595$ & $1.261$\\ 
UNI & $5.280$ & $4.814$ & $4.173$ & $3.663$ & $4.384$ & $3.970$\\ 
\cmidrule{1-7}
\multicolumn{7}{l}{Panel C. Sharpe Ratio} \\\cmidrule{1-7}
CP2 & $\underset{(6.872)}{ 1.010}$ & $\underset{(6.655)}{ 0.977}$ & $\underset{(5.837)}{ 0.854}$ & $\underset{(5.208)}{ 0.759}$ & $\underset{(6.028)}{ 0.883}$ & $\underset{(4.341)}{ 0.631}$\\ 
PP2 & $\underset{(4.554)}{ 0.662}$ & $\underset{(4.325)}{ 0.628}$ & $\underset{(5.389)}{ 0.786}$ & $\underset{(4.954)}{ 0.722}$ & $\underset{(5.077)}{ 0.740}$ & $\underset{(5.377)}{ 0.785}$\\ 
MVO & $\underset{(1.485)}{ 0.214}$ & $\underset{(1.004)}{ 0.145}$ & $\underset{(1.528)}{ 0.220}$ & $\underset{(-0.473)}{-0.068}$ & $\underset{(1.639)}{ 0.236}$ & $\underset{(0.577)}{ 0.083}$\\ 
UNI & $\underset{(3.125)}{ 0.452}$ & $\underset{(2.719)}{ 0.393}$ & $\underset{(3.843)}{ 0.557}$ & $\underset{(2.899)}{ 0.419}$ & $\underset{(3.493)}{ 0.506}$ & $\underset{(2.690)}{ 0.389}$\\ 
\cmidrule{1-7}
\multicolumn{7}{l}{Panel D. Alpha (\%)} \\\cmidrule{1-7}
CP2 & $1.844$ & $1.189$ & $1.116$ & $0.823$ & $1.340$ & $0.683$\\ 
PP2 & $2.252$ & $2.240$ & $2.183$ & $2.347$ & $2.069$ & $2.565$\\ 
MVO & $-0.210$ & $-0.097$ & $-0.342$ & $-0.391$ & $-0.208$ & $-0.214$\\ 
\cmidrule{1-7}
\multicolumn{7}{l}{Panel E. Beta to Simple Factor} \\\cmidrule{1-7}
CP2 & $0.146$ & $0.072$ & $0.259$ & $0.086$ & $0.172$ & $0.054$\\ 
PP2 & $0.684$ & $0.695$ & $0.703$ & $0.734$ & $0.583$ & $0.583$\\ 
MVO & $0.256$ & $0.168$ & $0.300$ & $0.201$ & $0.260$ & $0.188$\\ 
\cmidrule{1-7}
\multicolumn{7}{l}{Panel F. Idiosyncratic Volatility (\%)} \\\cmidrule{1-7}
CP2 & $2.041$ & $1.321$ & $1.749$ & $1.219$ & $1.910$ & $1.228$\\ 
PP2 & $4.583$ & $4.583$ & $3.951$ & $4.157$ & $3.889$ & $3.778$\\ 
MVO & $1.281$ & $1.102$ & $1.136$ & $1.077$ & $1.119$ & $1.018$\\ 
\cmidrule{1-7}
\multicolumn{7}{l}{Panel G. Information Ratio} \\\cmidrule{1-7}
CP2 & $\underset{(6.170)}{ 0.903}$ & $\underset{(6.149)}{ 0.900}$ & $\underset{(4.393)}{ 0.638}$ & $\underset{(4.642)}{ 0.675}$ & $\underset{(4.819)}{ 0.701}$ & $\underset{(3.835)}{ 0.556}$\\ 
PP2 & $\underset{(3.396)}{ 0.491}$ & $\underset{(3.379)}{ 0.489}$ & $\underset{(3.810)}{ 0.553}$ & $\underset{(3.892)}{ 0.565}$ & $\underset{(3.671)}{ 0.532}$ & $\underset{(4.668)}{ 0.679}$\\ 
MVO & $\underset{(-1.137)}{-0.164}$ & $\underset{(-0.609)}{-0.088}$ & $\underset{(-2.087)}{-0.301}$ & $\underset{(-2.512)}{-0.363}$ & $\underset{(-1.292)}{-0.186}$ & $\underset{(-1.456)}{-0.210}$\\ 
\bottomrule \end{tabular}
\label{tab:performance}
\end{table}
\vspace{\fill}

\newpage

\vspace*{\fill}
\begin{table}[H]
    \caption{Weight Statistics}
    \captionsetup{font=footnotesize}
    \caption*{This table shows the average monthly weight statistics of various portfolio selection methods estimated using a 21-day momentum signal. The test assets are the 25 and 100 equity portfolios independently sorted on size and three other factors (book-to-market, operating profitability, and investment). Panels A--F display the turnover, the proportional leverage, the sum of negative weights in the portfolio, and the minimum and maximum portfolio weights. The covariances are estimated on a rolling walk-forward basis with $T=120$ non-overlapping monthly returns. All measures are computed with 579 monthly portfolio weights from 09-09-1974 until 10-21-2022. }
    % \vspace*{5mm}
    \centering
    \footnotesize
\begin{tabular}{lcccccc} \toprule
 & FF25 & FF100 & \makecell{ME/OP \\ 25} & \makecell{ME/OP \\ 100} & \makecell{ME/INV \\ 25} & \makecell{ME/INV \\ 100}\\ \cmidrule{1-7}
\multicolumn{7}{l}{Panel A. Turnover} \\\cmidrule{1-7}
CP2 & $1.158$ & $1.202$ & $1.179$ & $1.225$ & $1.165$ & $1.217$\\ 
PP2 & $1.149$ & $1.148$ & $1.167$ & $1.154$ & $1.173$ & $1.162$\\ 
MVO & $1.377$ & $1.424$ & $1.378$ & $1.416$ & $1.381$ & $1.428$\\ 
UNI & $1.252$ & $1.282$ & $1.261$ & $1.295$ & $1.263$ & $1.298$\\ 
\cmidrule{1-7}
\multicolumn{7}{l}{Panel B. Proportional Leverage} \\\cmidrule{1-7}
CP2 & $0.493$ & $0.497$ & $0.499$ & $0.494$ & $0.498$ & $0.495$\\ 
PP2 & $0.496$ & $0.497$ & $0.493$ & $0.497$ & $0.500$ & $0.502$\\ 
MVO & $0.499$ & $0.499$ & $0.504$ & $0.497$ & $0.502$ & $0.500$\\ 
UNI & $0.489$ & $0.500$ & $0.489$ & $0.500$ & $0.489$ & $0.500$\\ 
\cmidrule{1-7}
\multicolumn{7}{l}{Panel C. Sum of Negative Weights} \\\cmidrule{1-7}
CP2 & $-0.496$ & $-0.497$ & $-0.498$ & $-0.498$ & $-0.497$ & $-0.498$\\ 
PP2 & $-0.500$ & $-0.500$ & $-0.500$ & $-0.500$ & $-0.500$ & $-0.500$\\ 
MVO & $-0.500$ & $-0.500$ & $-0.501$ & $-0.501$ & $-0.500$ & $-0.500$\\ 
UNI & $-0.500$ & $-0.500$ & $-0.500$ & $-0.500$ & $-0.500$ & $-0.500$\\ 
\cmidrule{1-7}
\multicolumn{7}{l}{Panel D. Minimum Weight} \\\cmidrule{1-7}
CP2 & $-0.106$ & $-0.033$ & $-0.106$ & $-0.033$ & $-0.103$ & $-0.032$\\ 
PP2 & $-0.102$ & $-0.034$ & $-0.105$ & $-0.037$ & $-0.104$ & $-0.033$\\ 
MVO & $-0.097$ & $-0.030$ & $-0.094$ & $-0.030$ & $-0.094$ & $-0.029$\\ 
UNI & $-0.077$ & $-0.020$ & $-0.077$ & $-0.020$ & $-0.077$ & $-0.020$\\ 
\cmidrule{1-7}
\multicolumn{7}{l}{Panel E. Maximum Weight} \\\cmidrule{1-7}
CP2 & $0.104$ & $0.033$ & $0.103$ & $0.032$ & $0.103$ & $0.032$\\ 
PP2 & $0.099$ & $0.033$ & $0.100$ & $0.036$ & $0.103$ & $0.033$\\ 
MVO & $0.098$ & $0.030$ & $0.095$ & $0.029$ & $0.096$ & $0.030$\\ 
UNI & $0.077$ & $0.020$ & $0.077$ & $0.020$ & $0.077$ & $0.020$\\ 
\cmidrule{1-7}
\multicolumn{7}{l}{Panel F. Standard Deviation of Weights} \\\cmidrule{1-7}
CP2 & $0.051$ & $0.013$ & $0.051$ & $0.013$ & $0.051$ & $0.013$\\ 
PP2 & $0.051$ & $0.013$ & $0.052$ & $0.013$ & $0.052$ & $0.013$\\ 
MVO & $0.050$ & $0.012$ & $0.050$ & $0.012$ & $0.050$ & $0.012$\\ 
UNI & $0.047$ & $0.012$ & $0.047$ & $0.012$ & $0.047$ & $0.012$\\ 
\bottomrule \end{tabular}
\label{tab:weightstat}
\end{table}
\vspace{\fill}

\newpage

\vspace*{\fill}
\begin{table}[H]
    \caption{Portfolio Returns Decomposition}
    \captionsetup{font=footnotesize}
    \caption*{This table shows the average returns (in percentages) of various portfolio selection methods due to static/dynamic trades and long/short trades. Panels A--C display the results for the dynamic and static trades as well as the share of returns due to dynamic trades. Panels D-E display the results for the long and short trades. The signal is a 21-day momentum. The test assets are the 25 and 100 equity portfolios independently sorted on size and three other factors (book-to-market, operating profitability, and investment). The covariances are estimated on a rolling walk-forward basis with $T=120$ non-overlapping monthly returns. All measures are computed with 579 monthly portfolio weights from 09-09-1974 until 10-21-2022. }
    % \vspace*{5mm}
    \centering
    \footnotesize
\begin{tabular}{lcccccc} \toprule
 & FF25 & FF100 & \makecell{ME/OP \\ 25} & \makecell{ME/OP \\ 100} & \makecell{ME/INV \\ 25} & \makecell{ME/INV \\ 100}\\ \cmidrule{1-7}
\cmidrule{1-7}
\multicolumn{7}{l}{Panel A. Static Returns (\%)} \\\cmidrule{1-7}
CP2 & $0.184$ & $0.204$ & $0.182$ & $0.159$ & $0.178$ & $0.138$\\ 
PP2 & $0.043$ & $0.032$ & $0.209$ & $0.165$ & $0.081$ & $0.026$\\ 
MVO & $0.118$ & $0.159$ & $0.084$ & $0.103$ & $0.142$ & $0.153$\\ 
UNI & $0.168$ & $0.170$ & $0.190$ & $0.166$ & $0.157$ & $0.158$\\ 
\cmidrule{1-7}
\multicolumn{7}{l}{Panel B. Dynamic Returns (\%)} \\\cmidrule{1-7}
CP2 & $2.008$ & $1.127$ & $1.561$ & $0.791$ & $1.640$ & $0.657$\\ 
PP2 & $3.718$ & $3.429$ & $3.614$ & $3.349$ & $3.302$ & $3.426$\\ 
MVO & $0.263$ & $0.029$ & $0.268$ & $-0.191$ & $0.217$ & $-0.067$\\ 
UNI & $2.097$ & $1.614$ & $2.059$ & $1.303$ & $1.988$ & $1.316$\\ 
\cmidrule{1-7}
\multicolumn{7}{l}{Panel C. Share of Dynamic Returns (\%)} \\\cmidrule{1-7}
CP2 & $91.466$ & $84.510$ & $89.383$ & $83.179$ & $90.064$ & $82.603$\\ 
PP2 & $98.802$ & $99.069$ & $94.401$ & $95.229$ & $97.558$ & $99.237$\\ 
MVO & $68.917$ & $15.496$ & $75.939$ & $217.015$ & $60.439$ & $-77.498$\\ 
UNI & $92.467$ & $90.326$ & $91.399$ & $88.562$ & $92.652$ & $89.249$\\ 
\cmidrule{1-7}
\multicolumn{7}{l}{Panel D. Long Leg Returns (\%)} \\\cmidrule{1-7}
CP2 & $8.012$ & $7.621$ & $7.666$ & $7.505$ & $8.021$ & $7.725$\\ 
PP2 & $8.893$ & $8.532$ & $8.664$ & $8.506$ & $8.716$ & $8.924$\\ 
MVO & $7.506$ & $7.306$ & $7.248$ & $7.197$ & $7.570$ & $7.457$\\ 
UNI & $8.057$ & $7.916$ & $7.940$ & $7.822$ & $8.255$ & $8.015$\\ 
\cmidrule{1-7}
\multicolumn{7}{l}{Panel E. Short Leg Returns (\%)} \\\cmidrule{1-7}
CP2 & $5.714$ & $6.297$ & $5.785$ & $6.522$ & $6.011$ & $6.848$\\ 
PP2 & $4.764$ & $4.747$ & $4.543$ & $4.697$ & $5.081$ & $5.240$\\ 
MVO & $6.939$ & $6.979$ & $6.637$ & $7.152$ & $7.001$ & $7.224$\\ 
UNI & $5.476$ & $5.861$ & $5.442$ & $6.163$ & $5.866$ & $6.347$\\ 
\bottomrule \end{tabular}
\label{tab:decom}
\end{table}
\vspace{\fill}

\newpage
\vspace*{\fill}
\begin{table}[H]
    \caption{Subsample Analysis}
    \captionsetup{font=footnotesize}
    \caption*{This table shows the annualized out-of-sample Sharpe ratio of various portfolio selection methods over different subsamples. Panels A--D display the results for Sample 1 (1974  through 1986), Sample 2 (1986  through 1998), Sample 3 (1998 through 2010), and Sample 4 (2010 through 2022). The signal is a 21-day momentum. The test assets are the 25 and 100 equity portfolios independently sorted on size and three other factors (book-to-market, operating profitability, and investment). The numbers in parenthesis are the t-statistics for the portfolio's Sharpe ratio and information ratio using the standard error of \cite{lo2002statistics}. The covariances are estimated on a rolling walk-forward basis with $T=120$ non-overlapping monthly returns. }
    \centering
    \footnotesize
\begin{tabular}{lccccccc} \toprule
 & FF25 & FF100 & \makecell{ME/OP \\ 25} & \makecell{ME/OP \\ 100} & \makecell{ME/INV \\ 25} & \makecell{ME/INV \\ 100}\\ \cmidrule{1-7}
\multicolumn{7}{l}{Panel A. Sharpe Ratio, Sample Period 1974 to 1986} \\\cmidrule{1-7}
CP2 & $\underset{(4.208)}{ 1.254}$ & $\underset{(2.874)}{ 0.842}$ & $\underset{(3.685)}{ 1.094}$ & $\underset{(1.378)}{ 0.401}$ & $\underset{(4.411)}{ 1.324}$ & $\underset{(2.163)}{ 0.632}$\\ 
PP2 & $\underset{(3.741)}{ 1.107}$ & $\underset{(3.979)}{ 1.182}$ & $\underset{(4.224)}{ 1.264}$ & $\underset{(3.660)}{ 1.086}$ & $\underset{(3.973)}{ 1.184}$ & $\underset{(4.093)}{ 1.222}$\\ 
MVO & $\underset{(0.056)}{ 0.016}$ & $\underset{(-0.527)}{-0.152}$ & $\underset{(1.232)}{ 0.358}$ & $\underset{(-0.580)}{-0.168}$ & $\underset{(1.419)}{ 0.412}$ & $\underset{(0.707)}{ 0.205}$\\ 
UNI & $\underset{(3.181)}{ 0.935}$ & $\underset{(2.607)}{ 0.762}$ & $\underset{(3.802)}{ 1.130}$ & $\underset{(2.773)}{ 0.814}$ & $\underset{(3.446)}{ 1.020}$ & $\underset{(2.850)}{ 0.838}$\\ 
\cmidrule{1-7}
\multicolumn{7}{l}{Panel B. Sharpe Ratio, Sample Period 1986 to 1998} \\\cmidrule{1-7}
CP2 & $\underset{(5.117)}{ 1.549}$ & $\underset{(5.199)}{ 1.577}$ & $\underset{(4.124)}{ 1.232}$ & $\underset{(3.854)}{ 1.146}$ & $\underset{(4.100)}{ 1.224}$ & $\underset{(3.604)}{ 1.069}$\\ 
PP2 & $\underset{(3.560)}{ 1.051}$ & $\underset{(3.989)}{ 1.185}$ & $\underset{(4.239)}{ 1.269}$ & $\underset{(3.892)}{ 1.159}$ & $\underset{(4.094)}{ 1.222}$ & $\underset{(4.074)}{ 1.216}$\\ 
MVO & $\underset{(2.025)}{ 0.589}$ & $\underset{(0.977)}{ 0.283}$ & $\underset{(1.185)}{ 0.344}$ & $\underset{(-0.431)}{-0.125}$ & $\underset{(1.616)}{ 0.470}$ & $\underset{(0.134)}{ 0.039}$\\ 
UNI & $\underset{(2.755)}{ 0.806}$ & $\underset{(2.466)}{ 0.720}$ & $\underset{(2.230)}{ 0.652}$ & $\underset{(1.393)}{ 0.405}$ & $\underset{(2.134)}{ 0.623}$ & $\underset{(1.307)}{ 0.380}$\\ 
\cmidrule{1-7}
\multicolumn{7}{l}{Panel C. Sharpe Ratio, Sample Period 1998 to 2010} \\\cmidrule{1-7}
CP2 & $\underset{(3.010)}{ 0.883}$ & $\underset{(4.237)}{ 1.263}$ & $\underset{(3.086)}{ 0.909}$ & $\underset{(3.788)}{ 1.126}$ & $\underset{(3.390)}{ 1.002}$ & $\underset{(3.034)}{ 0.894}$\\ 
PP2 & $\underset{(1.407)}{ 0.408}$ & $\underset{(1.342)}{ 0.389}$ & $\underset{(2.747)}{ 0.807}$ & $\underset{(3.418)}{ 1.011}$ & $\underset{(2.369)}{ 0.693}$ & $\underset{(2.692)}{ 0.790}$\\ 
MVO & $\underset{(0.977)}{ 0.283}$ & $\underset{(1.214)}{ 0.351}$ & $\underset{(0.693)}{ 0.201}$ & $\underset{(-0.378)}{-0.109}$ & $\underset{(0.557)}{ 0.161}$ & $\underset{(0.416)}{ 0.121}$\\ 
UNI & $\underset{(1.140)}{ 0.330}$ & $\underset{(1.137)}{ 0.329}$ & $\underset{(1.606)}{ 0.467}$ & $\underset{(1.483)}{ 0.431}$ & $\underset{(1.373)}{ 0.399}$ & $\underset{(1.151)}{ 0.334}$\\ 
\cmidrule{1-7}
\multicolumn{7}{l}{Panel D. Sharpe Ratio, Sample Period 2010 to 2022} \\\cmidrule{1-7}
CP2 & $\underset{(1.539)}{ 0.448}$ & $\underset{(1.339)}{ 0.389}$ & $\underset{(0.819)}{ 0.238}$ & $\underset{(1.105)}{ 0.322}$ & $\underset{(0.401)}{ 0.117}$ & $\underset{(-0.599)}{-0.174}$\\ 
PP2 & $\underset{(1.205)}{ 0.350}$ & $\underset{(0.515)}{ 0.149}$ & $\underset{(0.034)}{ 0.010}$ & $\underset{(-0.731)}{-0.213}$ & $\underset{(-0.468)}{-0.136}$ & $\underset{(-0.185)}{-0.054}$\\ 
MVO & $\underset{(-0.212)}{-0.062}$ & $\underset{(-0.021)}{-0.006}$ & $\underset{(0.186)}{ 0.054}$ & $\underset{(0.132)}{ 0.038}$ & $\underset{(-0.140)}{-0.041}$ & $\underset{(-0.552)}{-0.160}$\\ 
UNI & $\underset{(-0.037)}{-0.011}$ & $\underset{(-0.108)}{-0.031}$ & $\underset{(0.505)}{ 0.147}$ & $\underset{(0.412)}{ 0.120}$ & $\underset{(0.358)}{ 0.104}$ & $\underset{(0.327)}{ 0.095}$\\ 
\bottomrule \end{tabular}
\label{tab:subsample}
\end{table}
\vspace{\fill}

\newpage

\vspace*{\fill}
\begin{table}[H]
    \caption{Forecast Horizon}
    \captionsetup{font=footnotesize}
    \caption*{This table shows the annualized out-of-sample Sharpe ratio of various portfolio selection methods for different rebalancing frequencies. Panels A--C display the results for 1 day, 5 days, and 10 days. The signal is a 21-day momentum. The test assets are the 25 and 100 equity portfolios independently sorted on size and three other factors (book-to-market, operating profitability, and investment). The numbers in parenthesis are the t-statistics for the portfolio's Sharpe ratio and information ratio using the standard error of \cite{lo2002statistics}. The covariances are estimated on a rolling walk-forward basis with $T=120$ non-overlapping monthly returns. Panel A is based on $14{,}550$ daily out-of-sample returns from 01-06-1965 until 10-21-2022, Panel B is based on 2814 weekly out-of-sample returns from 12-06-1966 until 10-21-2022, Panel C is based on 1347 fortnightly out-of-sample returns from 06-12-1969 until 10-21-2022.  }
    \centering
    \footnotesize
\begin{tabular}{lcccccc} \toprule
 & FF25 & FF100 & \makecell{ME/OP \\ 25} & \makecell{ME/OP \\ 100} & \makecell{ME/INV \\ 25} & \makecell{ME/INV \\ 100}\\ \cmidrule{1-7}
\multicolumn{7}{l}{Panel A. Sharpe Ratio, Horizon: 1 day} \\\cmidrule{1-7}
CP2 & $\underset{(15.920)}{ 2.104}$ & $\underset{(9.869)}{ 1.301}$ & $\underset{(14.031)}{ 1.853}$ & $\underset{(7.535)}{ 0.993}$ & $\underset{(13.717)}{ 1.811}$ & $\underset{(8.280)}{ 1.091}$\\ 
PP2 & $\underset{(9.299)}{ 1.226}$ & $\underset{(8.119)}{ 1.070}$ & $\underset{(8.767)}{ 1.155}$ & $\underset{(6.042)}{ 0.796}$ & $\underset{(10.133)}{ 1.336}$ & $\underset{(7.266)}{ 0.957}$\\ 
MVO & $\underset{(8.264)}{ 1.089}$ & $\underset{(-2.896)}{-0.381}$ & $\underset{(5.971)}{ 0.786}$ & $\underset{(-5.518)}{-0.727}$ & $\underset{(6.195)}{ 0.816}$ & $\underset{(-5.201)}{-0.685}$\\ 
UNI & $\underset{(8.702)}{ 1.147}$ & $\underset{(5.016)}{ 0.660}$ & $\underset{(7.444)}{ 0.981}$ & $\underset{(2.434)}{ 0.320}$ & $\underset{(7.679)}{ 1.012}$ & $\underset{(2.762)}{ 0.364}$\\ 
\cmidrule{1-7}
\multicolumn{7}{l}{Panel B. Sharpe Ratio, Horizon: 5 days} \\\cmidrule{1-7}
CP2 & $\underset{(11.279)}{ 1.527}$ & $\underset{(8.810)}{ 1.187}$ & $\underset{(10.156)}{ 1.372}$ & $\underset{(6.313)}{ 0.848}$ & $\underset{(10.233)}{ 1.383}$ & $\underset{(7.980)}{ 1.074}$\\ 
PP2 & $\underset{(7.448)}{ 1.002}$ & $\underset{(7.676)}{ 1.033}$ & $\underset{(7.728)}{ 1.040}$ & $\underset{(6.946)}{ 0.934}$ & $\underset{(8.953)}{ 1.207}$ & $\underset{(8.004)}{ 1.078}$\\ 
MVO & $\underset{(4.140)}{ 0.555}$ & $\underset{(-1.096)}{-0.147}$ & $\underset{(3.097)}{ 0.415}$ & $\underset{(-3.351)}{-0.449}$ & $\underset{(3.627)}{ 0.486}$ & $\underset{(-1.336)}{-0.179}$\\ 
UNI & $\underset{(7.222)}{ 0.971}$ & $\underset{(5.075)}{ 0.681}$ & $\underset{(6.410)}{ 0.861}$ & $\underset{(3.378)}{ 0.453}$ & $\underset{(6.350)}{ 0.853}$ & $\underset{(3.469)}{ 0.465}$\\ 
\cmidrule{1-7}
\multicolumn{7}{l}{Panel C. Sharpe Ratio, Horizon: 10 days} \\\cmidrule{1-7}
CP2 & $\underset{(9.356)}{ 1.302}$ & $\underset{(7.416)}{ 1.025}$ & $\underset{(9.917)}{ 1.382}$ & $\underset{(7.469)}{ 1.033}$ & $\underset{(9.255)}{ 1.287}$ & $\underset{(7.160)}{ 0.989}$\\ 
PP2 & $\underset{(6.702)}{ 0.925}$ & $\underset{(6.892)}{ 0.952}$ & $\underset{(6.882)}{ 0.950}$ & $\underset{(6.005)}{ 0.827}$ & $\underset{(6.266)}{ 0.864}$ & $\underset{(6.627)}{ 0.914}$\\ 
MVO & $\underset{(3.078)}{ 0.422}$ & $\underset{(1.123)}{ 0.154}$ & $\underset{(2.666)}{ 0.365}$ & $\underset{(-1.125)}{-0.154}$ & $\underset{(3.351)}{ 0.459}$ & $\underset{(0.202)}{ 0.028}$\\ 
UNI & $\underset{(5.727)}{ 0.789}$ & $\underset{(4.993)}{ 0.686}$ & $\underset{(6.245)}{ 0.861}$ & $\underset{(4.555)}{ 0.626}$ & $\underset{(5.650)}{ 0.778}$ & $\underset{(4.142)}{ 0.569}$\\ 
\bottomrule \end{tabular}
\label{tab:rebalfreq}
\end{table}
\vspace{\fill}

\newpage

\vspace*{\fill}
\begin{table}[H]
    \caption{Shrinkage Intensity}
    \captionsetup{font=footnotesize}
    \caption*{This table shows the annualized out-of-sample Sharpe Ratio of Canonical Portfolio with the leading two components for different shrinkage intensities in the estimation of the covariance of the signals. Panels A--E display the results for the full sample period, Sample 1 (1974  through 1986), Sample 2 (1986  through 1998), Sample 3 (1998 through 2010), and Sample 4 (2010 through 2022). The rows within each panel correspond to results with shrinkage parameter $\delta_x \in \{ 0, 0.2, 0.5, 1 \}$. The signal is a 21-day momentum. The test assets are the 25 and 100 equity portfolios independently sorted on size and three other factors (book-to-market, operating profitability, and investment). The numbers in parenthesis are the t-statistics for the portfolio's Sharpe ratio and information ratio using the standard error of \cite{lo2002statistics}. The covariances are estimated on a rolling walk-forward basis with $T=120$ non-overlapping monthly returns.  All measures are computed with 579 monthly portfolio weights from 09-09-1974 until 10-21-2022. }
    \centering
    \footnotesize
\begin{tabular}{lcccccc} \toprule
 & FF25 & FF100 & \makecell{ME/OP \\ 25} & \makecell{ME/OP \\ 100} & \makecell{ME/INV \\ 25} & \makecell{ME/INV \\ 100}\\ \cmidrule{1-7}
\multicolumn{7}{l}{Panel A. Sharpe Ratio, Full Sample Period} \\\cmidrule{1-7}
CP2 ($\delta _x = 1)$ & $\underset{(6.707)}{ 0.986}$ & $\underset{(6.400)}{ 0.939}$ & $\underset{(5.752)}{ 0.842}$ & $\underset{(4.359)}{ 0.634}$ & $\underset{(5.798)}{ 0.849}$ & $\underset{(4.047)}{ 0.588}$\\ 
CP2 ($\delta _x = 0.8)$ & $\underset{(6.764)}{ 0.994}$ & $\underset{(6.766)}{ 0.995}$ & $\underset{(5.791)}{ 0.848}$ & $\underset{(5.817)}{ 0.851}$ & $\underset{(5.970)}{ 0.875}$ & $\underset{(4.545)}{ 0.661}$\\ 
CP2 ($\delta _x = 0.5)$ & $\underset{(6.624)}{ 0.973}$ & $\underset{(5.805)}{ 0.849}$ & $\underset{(5.712)}{ 0.836}$ & $\underset{(6.232)}{ 0.914}$ & $\underset{(5.935)}{ 0.869}$ & $\underset{(4.994)}{ 0.728}$\\ 
CP2 ($\delta _x = 0.2)$ & $\underset{(6.679)}{ 0.981}$ & $\underset{(3.393)}{ 0.491}$ & $\underset{(4.949)}{ 0.721}$ & $\underset{(3.778)}{ 0.548}$ & $\underset{(6.125)}{ 0.898}$ & $\underset{(3.694)}{ 0.536}$\\ 
% CP2 ($\delta _x = 0.0)$ & $\underset{(6.252)}{ 0.916}$ & $\underset{(2.770)}{ 0.400}$ & $\underset{(4.474)}{ 0.651}$ & $\underset{(3.911)}{ 0.568}$ & $\underset{(5.271)}{ 0.769}$ & $\underset{(3.404)}{ 0.493}$\\ 
\cmidrule{1-7}
\multicolumn{7}{l}{Panel B. Sharpe Ratio, Subsample Period 1974 to 1986} \\\cmidrule{1-7}
CP2 ($\delta _x = 1)$ & $\underset{(3.969)}{ 1.178}$ & $\underset{(3.284)}{ 0.966}$ & $\underset{(3.558)}{ 1.054}$ & $\underset{(1.364)}{ 0.396}$ & $\underset{(4.092)}{ 1.222}$ & $\underset{(1.482)}{ 0.431}$\\ 
CP2 ($\delta _x = 0.8)$ & $\underset{(4.401)}{ 1.315}$ & $\underset{(2.593)}{ 0.757}$ & $\underset{(3.714)}{ 1.103}$ & $\underset{(1.701)}{ 0.495}$ & $\underset{(4.598)}{ 1.384}$ & $\underset{(2.566)}{ 0.752}$\\ 
CP2 ($\delta _x = 0.5)$ & $\underset{(4.278)}{ 1.276}$ & $\underset{(1.672)}{ 0.485}$ & $\underset{(4.040)}{ 1.205}$ & $\underset{(1.645)}{ 0.479}$ & $\underset{(4.990)}{ 1.513}$ & $\underset{(2.896)}{ 0.852}$\\ 
CP2 ($\delta _x = 0.2)$ & $\underset{(4.291)}{ 1.280}$ & $\underset{(0.203)}{ 0.059}$ & $\underset{(3.179)}{ 0.937}$ & $\underset{(-0.123)}{-0.036}$ & $\underset{(5.192)}{ 1.580}$ & $\underset{(1.670)}{ 0.486}$\\ 
% CP2 ($\delta _x = 0.0)$ & $\underset{(4.084)}{ 1.214}$ & $\underset{(-0.209)}{-0.060}$ & $\underset{(2.270)}{ 0.664}$ & $\underset{(-0.264)}{-0.076}$ & $\underset{(4.735)}{ 1.429}$ & $\underset{(1.632)}{ 0.475}$\\ 
\cmidrule{1-7}
\multicolumn{7}{l}{Panel C. Sharpe Ratio, Subsample Period 1986 to 1998} \\\cmidrule{1-7}
CP2 ($\delta _x = 1)$ & $\underset{(4.863)}{ 1.465}$ & $\underset{(4.680)}{ 1.406}$ & $\underset{(3.814)}{ 1.134}$ & $\underset{(2.485)}{ 0.728}$ & $\underset{(3.892)}{ 1.159}$ & $\underset{(3.521)}{ 1.043}$\\ 
CP2 ($\delta _x = 0.8)$ & $\underset{(5.267)}{ 1.599}$ & $\underset{(5.493)}{ 1.676}$ & $\underset{(4.348)}{ 1.303}$ & $\underset{(4.518)}{ 1.358}$ & $\underset{(4.158)}{ 1.243}$ & $\underset{(3.644)}{ 1.081}$\\ 
CP2 ($\delta _x = 0.5)$ & $\underset{(5.571)}{ 1.703}$ & $\underset{(5.113)}{ 1.548}$ & $\underset{(4.300)}{ 1.288}$ & $\underset{(4.917)}{ 1.489}$ & $\underset{(3.580)}{ 1.061}$ & $\underset{(3.854)}{ 1.147}$\\ 
CP2 ($\delta _x = 0.2)$ & $\underset{(6.245)}{ 1.939}$ & $\underset{(4.176)}{ 1.244}$ & $\underset{(4.023)}{ 1.200}$ & $\underset{(3.500)}{ 1.036}$ & $\underset{(3.336)}{ 0.986}$ & $\underset{(2.669)}{ 0.783}$\\ 
% CP2 ($\delta _x = 0.0)$ & $\underset{(5.455)}{ 1.663}$ & $\underset{(3.988)}{ 1.184}$ & $\underset{(3.773)}{ 1.121}$ & $\underset{(3.547)}{ 1.051}$ & $\underset{(3.017)}{ 0.888}$ & $\underset{(2.609)}{ 0.765}$\\ 
\cmidrule{1-7}
\multicolumn{7}{l}{Panel D. Sharpe Ratio, Subsample Period 1998 to 2010} \\\cmidrule{1-7}
CP2 ($\delta _x = 1)$ & $\underset{(3.009)}{ 0.883}$ & $\underset{(3.896)}{ 1.156}$ & $\underset{(3.244)}{ 0.958}$ & $\underset{(4.052)}{ 1.209}$ & $\underset{(3.594)}{ 1.065}$ & $\underset{(3.599)}{ 1.067}$\\ 
CP2 ($\delta _x = 0.8)$ & $\underset{(2.907)}{ 0.852}$ & $\underset{(4.338)}{ 1.295}$ & $\underset{(2.954)}{ 0.869}$ & $\underset{(3.463)}{ 1.025}$ & $\underset{(3.096)}{ 0.912}$ & $\underset{(2.720)}{ 0.798}$\\ 
CP2 ($\delta _x = 0.5)$ & $\underset{(2.322)}{ 0.677}$ & $\underset{(3.819)}{ 1.132}$ & $\underset{(2.441)}{ 0.715}$ & $\underset{(2.556)}{ 0.749}$ & $\underset{(2.473)}{ 0.724}$ & $\underset{(2.193)}{ 0.641}$\\ 
CP2 ($\delta _x = 0.2)$ & $\underset{(1.687)}{ 0.489}$ & $\underset{(1.547)}{ 0.448}$ & $\underset{(2.621)}{ 0.769}$ & $\underset{(0.980)}{ 0.285}$ & $\underset{(2.765)}{ 0.812}$ & $\underset{(2.025)}{ 0.591}$\\ 
% CP2 ($\delta _x = 0.0)$ & $\underset{(1.948)}{ 0.566}$ & $\underset{(0.865)}{ 0.250}$ & $\underset{(2.777)}{ 0.816}$ & $\underset{(1.162)}{ 0.338}$ & $\underset{(2.746)}{ 0.806}$ & $\underset{(1.658)}{ 0.483}$\\ 
\cmidrule{1-7}
\multicolumn{7}{l}{Panel E. Sharpe Ratio, Subsample Period 2010 to 2022} \\\cmidrule{1-7}
CP2 ($\delta _x = 1)$ & $\underset{(1.671)}{ 0.487}$ & $\underset{(1.295)}{ 0.376}$ & $\underset{(0.782)}{ 0.228}$ & $\underset{(0.416)}{ 0.121}$ & $\underset{(0.298)}{ 0.087}$ & $\underset{(-1.059)}{-0.308}$\\ 
CP2 ($\delta _x = 0.8)$ & $\underset{(1.240)}{ 0.360}$ & $\underset{(1.348)}{ 0.392}$ & $\underset{(0.802)}{ 0.233}$ & $\underset{(1.764)}{ 0.516}$ & $\underset{(0.390)}{ 0.113}$ & $\underset{(-0.078)}{-0.023}$\\ 
CP2 ($\delta _x = 0.5)$ & $\underset{(1.501)}{ 0.437}$ & $\underset{(0.937)}{ 0.272}$ & $\underset{(1.288)}{ 0.376}$ & $\underset{(3.318)}{ 0.984}$ & $\underset{(1.133)}{ 0.330}$ & $\underset{(1.160)}{ 0.338}$\\ 
CP2 ($\delta _x = 0.2)$ & $\underset{(1.799)}{ 0.524}$ & $\underset{(0.762)}{ 0.221}$ & $\underset{(0.490)}{ 0.143}$ & $\underset{(3.083)}{ 0.912}$ & $\underset{(1.442)}{ 0.421}$ & $\underset{(0.898)}{ 0.261}$\\ 
% CP2 ($\delta _x = 0.0)$ & $\underset{(1.779)}{ 0.518}$ & $\underset{(0.915)}{ 0.265}$ & $\underset{(0.304)}{ 0.088}$ & $\underset{(3.215)}{ 0.952}$ & $\underset{(0.595)}{ 0.173}$ & $\underset{(0.840)}{ 0.244}$\\ 
\bottomrule \end{tabular}
\label{tab:shrinkage}
\end{table}
\vspace{\fill}

\newpage

\vspace*{\fill}
\begin{table}[H]
    \caption{Momentum Lookback Window}
    \captionsetup{font=footnotesize}
    \caption*{This table shows the annualized out-of-sample Sharpe Ratio of various portfolio selection methods based on momentum signals of different lookback window sizes. Panels A--E display the results for 42 days, 63 days, 84 days, 126 days, and 252 days. The test assets are the 25 and 100 equity portfolios independently sorted on size and three other factors (book-to-market, operating profitability, and investment). The numbers in parenthesis are the t-statistics for the portfolio's Sharpe ratio and information ratio using the standard error of \cite{lo2002statistics}. The covariances are estimated on a rolling walk-forward basis with $T=120$ non-overlapping monthly returns.  All measures are computed with 579 monthly portfolio weights from 09-09-1974 until 10-21-2022. }
    \centering
    \footnotesize
\begin{tabular}{lcccccc} \toprule
 & FF25 & FF100 & \makecell{ME/OP \\ 25} & \makecell{ME/OP \\ 100} & \makecell{ME/INV \\ 25} & \makecell{ME/INV \\ 100}\\ \cmidrule{1-7}
\multicolumn{7}{l}{Panel A. Sharpe Ratio, Lookback Window Size: 42} \\\cmidrule{1-7}
CP2 & $\underset{(6.054)}{ 0.887}$ & $\underset{(6.837)}{ 1.006}$ & $\underset{(5.895)}{ 0.863}$ & $\underset{(5.074)}{ 0.740}$ & $\underset{(6.198)}{ 0.909}$ & $\underset{(4.853)}{ 0.707}$\\ 
PP2 & $\underset{(4.465)}{ 0.649}$ & $\underset{(4.298)}{ 0.624}$ & $\underset{(4.587)}{ 0.668}$ & $\underset{(3.759)}{ 0.545}$ & $\underset{(5.342)}{ 0.780}$ & $\underset{(4.728)}{ 0.688}$\\ 
MVO & $\underset{(1.647)}{ 0.238}$ & $\underset{(1.380)}{ 0.199}$ & $\underset{(1.444)}{ 0.208}$ & $\underset{(-0.355)}{-0.051}$ & $\underset{(1.468)}{ 0.212}$ & $\underset{(0.761)}{ 0.110}$\\ 
UNI & $\underset{(3.031)}{ 0.438}$ & $\underset{(2.532)}{ 0.366}$ & $\underset{(3.462)}{ 0.502}$ & $\underset{(2.748)}{ 0.398}$ & $\underset{(2.982)}{ 0.432}$ & $\underset{(2.265)}{ 0.327}$\\ 
\cmidrule{1-7}
\multicolumn{7}{l}{Panel B. Sharpe Ratio, Lookback Window Size: 63} \\\cmidrule{1-7}
CP2 & $\underset{(4.575)}{ 0.665}$ & $\underset{(6.681)}{ 0.982}$ & $\underset{(5.948)}{ 0.871}$ & $\underset{(5.389)}{ 0.787}$ & $\underset{(6.143)}{ 0.901}$ & $\underset{(6.199)}{ 0.909}$\\ 
PP2 & $\underset{(4.216)}{ 0.612}$ & $\underset{(3.962)}{ 0.575}$ & $\underset{(4.913)}{ 0.716}$ & $\underset{(4.555)}{ 0.663}$ & $\underset{(4.627)}{ 0.674}$ & $\underset{(4.479)}{ 0.652}$\\ 
MVO & $\underset{(2.273)}{ 0.328}$ & $\underset{(1.995)}{ 0.288}$ & $\underset{(1.284)}{ 0.185}$ & $\underset{(-1.354)}{-0.195}$ & $\underset{(1.753)}{ 0.253}$ & $\underset{(1.229)}{ 0.177}$\\ 
UNI & $\underset{(2.901)}{ 0.420}$ & $\underset{(2.536)}{ 0.366}$ & $\underset{(3.308)}{ 0.479}$ & $\underset{(2.325)}{ 0.336}$ & $\underset{(2.716)}{ 0.393}$ & $\underset{(1.979)}{ 0.286}$\\ 
\cmidrule{1-7}
\multicolumn{7}{l}{Panel C. Sharpe Ratio, Lookback Window Size: 84} \\\cmidrule{1-7}
CP2 & $\underset{(4.306)}{ 0.625}$ & $\underset{(6.196)}{ 0.908}$ & $\underset{(4.628)}{ 0.674}$ & $\underset{(3.743)}{ 0.543}$ & $\underset{(4.237)}{ 0.616}$ & $\underset{(5.834)}{ 0.854}$\\ 
PP2 & $\underset{(3.504)}{ 0.508}$ & $\underset{(3.129)}{ 0.453}$ & $\underset{(2.826)}{ 0.409}$ & $\underset{(3.329)}{ 0.482}$ & $\underset{(3.211)}{ 0.465}$ & $\underset{(3.460)}{ 0.502}$\\ 
MVO & $\underset{(1.829)}{ 0.264}$ & $\underset{(2.110)}{ 0.305}$ & $\underset{(1.117)}{ 0.161}$ & $\underset{(-0.675)}{-0.097}$ & $\underset{(1.480)}{ 0.214}$ & $\underset{(2.169)}{ 0.313}$\\ 
UNI & $\underset{(1.917)}{ 0.277}$ & $\underset{(1.891)}{ 0.273}$ & $\underset{(2.352)}{ 0.340}$ & $\underset{(1.796)}{ 0.259}$ & $\underset{(1.715)}{ 0.248}$ & $\underset{(1.678)}{ 0.242}$\\ 
\cmidrule{1-7}
\multicolumn{7}{l}{Panel D. Sharpe Ratio, Lookback Window Size: 126} \\\cmidrule{1-7}
CP2 & $\underset{(4.605)}{ 0.670}$ & $\underset{(6.042)}{ 0.885}$ & $\underset{(3.755)}{ 0.545}$ & $\underset{(3.352)}{ 0.486}$ & $\underset{(4.179)}{ 0.607}$ & $\underset{(5.270)}{ 0.769}$\\ 
PP2 & $\underset{(3.169)}{ 0.459}$ & $\underset{(2.957)}{ 0.428}$ & $\underset{(3.055)}{ 0.442}$ & $\underset{(3.865)}{ 0.561}$ & $\underset{(2.486)}{ 0.359}$ & $\underset{(1.907)}{ 0.275}$\\ 
MVO & $\underset{(2.172)}{ 0.314}$ & $\underset{(2.620)}{ 0.379}$ & $\underset{(0.808)}{ 0.117}$ & $\underset{(-0.743)}{-0.107}$ & $\underset{(2.331)}{ 0.337}$ & $\underset{(1.780)}{ 0.257}$\\ 
UNI & $\underset{(2.455)}{ 0.355}$ & $\underset{(2.154)}{ 0.311}$ & $\underset{(2.390)}{ 0.346}$ & $\underset{(1.910)}{ 0.276}$ & $\underset{(2.040)}{ 0.295}$ & $\underset{(1.760)}{ 0.254}$\\ 
\cmidrule{1-7}
\multicolumn{7}{l}{Panel E. Sharpe Ratio, Lookback Window Size: 252} \\\cmidrule{1-7}
CP2 & $\underset{(6.851)}{ 1.008}$ & $\underset{(8.089)}{ 1.200}$ & $\underset{(5.511)}{ 0.805}$ & $\underset{(4.953)}{ 0.722}$ & $\underset{(6.460)}{ 0.949}$ & $\underset{(7.629)}{ 1.129}$\\ 
PP2 & $\underset{(4.035)}{ 0.585}$ & $\underset{(3.526)}{ 0.511}$ & $\underset{(3.683)}{ 0.534}$ & $\underset{(3.654)}{ 0.530}$ & $\underset{(3.521)}{ 0.511}$ & $\underset{(2.282)}{ 0.330}$\\ 
MVO & $\underset{(3.458)}{ 0.501}$ & $\underset{(3.986)}{ 0.578}$ & $\underset{(2.502)}{ 0.362}$ & $\underset{(0.300)}{ 0.043}$ & $\underset{(3.240)}{ 0.469}$ & $\underset{(2.701)}{ 0.391}$\\ 
UNI & $\underset{(3.023)}{ 0.437}$ & $\underset{(2.672)}{ 0.386}$ & $\underset{(3.612)}{ 0.524}$ & $\underset{(3.103)}{ 0.449}$ & $\underset{(3.042)}{ 0.441}$ & $\underset{(2.650)}{ 0.383}$\\ 
\bottomrule \end{tabular}
\label{tab:momspeed}
\end{table}
\vspace{\fill}

\newpage

\vspace*{\fill}
\begin{table}[H]
    \caption{Approximate versus Actual Solution}
    \captionsetup{font=footnotesize}
    \caption*{This table shows the out-of-sample performances (in percentages) of the `approximated' versus `actual' portfolio policy using a 21-day momentum signal. Panels A--G display the results for the average returns, the Sharpe ratio of returns, the alpha with respect to the Fama-French 5-factor regression augmented with the univariate factor (UNI), the beta of the simple factor, the idiosyncratic volatility, and the information ratio. The test assets are the 25 and 100 equity portfolios independently sorted on size and three other factors (book-to-market, operating profitability, and investment).  The numbers in parenthesis are the t-statistics for the portfolio's Sharpe ratio and information ratio using the standard error of \cite{lo2002statistics}. The covariances are estimated on a rolling walk-forward basis with $T=120$ non-overlapping monthly returns. All measures are computed with 579 monthly portfolio weights from 09-09-1974 until 10-21-2022.}
    \centering
    \footnotesize
\begin{tabular}{lcccccc} \toprule
 & FF25 & FF100 & \makecell{ME/OP \\ 25} & \makecell{ME/OP \\ 100} & \makecell{ME/INV \\ 25} & \makecell{ME/INV \\ 100}\\ \cmidrule{1-7}
\multicolumn{7}{l}{Panel A. Average Returns (\%)} \\\cmidrule{1-7}
CP2 (Approx) & $2.199$ & $1.341$ & $1.731$ & $0.957$ & $1.797$ & $0.777$\\ 
CP2 (Full) & $2.253$ & $1.175$ & $1.578$ & $0.822$ & $1.675$ & $0.596$\\ 
\cmidrule{1-7}
\multicolumn{7}{l}{Panel B. Standard Deviation of Returns (\%)} \\\cmidrule{1-7}
CP2 (Approx) & $2.177$ & $1.366$ & $2.049$ & $1.260$ & $2.057$ & $1.248$\\ 
CP2 (Full) & $2.210$ & $1.356$ & $2.119$ & $1.229$ & $2.061$ & $1.309$\\ 
\cmidrule{1-7}
\multicolumn{7}{l}{Panel C. Sharpe Ratio} \\\cmidrule{1-7}
CP2 (Approx) & $\underset{(6.872)}{ 1.011}$ & $\underset{(6.685)}{ 0.982}$ & $\underset{(5.797)}{ 0.848}$ & $\underset{(5.226)}{ 0.763}$ & $\underset{(5.985)}{ 0.877}$ & $\underset{(4.308)}{ 0.626}$\\ 
CP2 (Full) & $\underset{(6.933)}{ 1.020}$ & $\underset{(5.934)}{ 0.868}$ & $\underset{(5.140)}{ 0.750}$ & $\underset{(4.620)}{ 0.672}$ & $\underset{(5.588)}{ 0.817}$ & $\underset{(3.177)}{ 0.460}$\\ 
\cmidrule{1-7}
\multicolumn{7}{l}{Panel D. Alpha (\%)} \\\cmidrule{1-7}
CP2 (Approx) & $1.849$ & $1.197$ & $1.109$ & $0.826$ & $1.334$ & $0.679$\\ 
CP2 (Full) & $1.908$ & $1.043$ & $0.942$ & $0.739$ & $1.224$ & $0.536$\\ 
\cmidrule{1-7}
\multicolumn{7}{l}{Panel E. Beta to Simple Factor} \\\cmidrule{1-7}
CP2 (Approx) & $0.146$ & $0.071$ & $0.259$ & $0.086$ & $0.172$ & $0.053$\\ 
CP2 (Full) & $0.144$ & $0.067$ & $0.267$ & $0.057$ & $0.165$ & $0.046$\\ 
\cmidrule{1-7}
\multicolumn{7}{l}{Panel F. Idiosyncratic Volatility (\%)} \\\cmidrule{1-7}
CP2 (Approx) & $2.043$ & $1.321$ & $1.749$ & $1.220$ & $1.911$ & $1.229$\\ 
CP2 (Full) & $2.083$ & $1.315$ & $1.807$ & $1.208$ & $1.921$ & $1.297$\\ 
\cmidrule{1-7}
\multicolumn{7}{l}{Panel G. Information Ratio} \\\cmidrule{1-7}
CP2 (Approx) & $\underset{(6.176)}{ 0.905}$ & $\underset{(6.183)}{ 0.906}$ & $\underset{(4.360)}{ 0.634}$ & $\underset{(4.651)}{ 0.677}$ & $\underset{(4.792)}{ 0.698}$ & $\underset{(3.809)}{ 0.553}$\\ 
CP2 (Full) & $\underset{(6.247)}{ 0.916}$ & $\underset{(5.434)}{ 0.793}$ & $\underset{(3.596)}{ 0.522}$ & $\underset{(4.210)}{ 0.612}$ & $\underset{(4.379)}{ 0.637}$ & $\underset{(2.858)}{ 0.414}$\\ 
\bottomrule \end{tabular}
\label{tab:approx}
\end{table}
\vspace{\fill}

\newpage
\vspace*{\fill}
\begin{table}[H]
    \caption{Performance Summary for Two Signals Case}
    \captionsetup{font=footnotesize}
    \caption*{This table shows the out-of-sample performances (in percentages) of various portfolio selection methods using two momentum signals with a lookback window of 21 days and 252 days. Panels A--G display the results for the average returns, the Sharpe ratio of returns, the alpha with respect to the Fama-French 5-factor regression augmented with the univariate factor (UNI), the beta of the simple factor, the idiosyncratic volatility, and the information ratio. The test assets are 25 and 100 equity portfolios independently sorted on size and three other factors (book-to-market, operating profitability, and investment). The numbers in parenthesis are the t-statistics for the portfolio's Sharpe ratio and information ratio using the standard error of \cite{lo2002statistics}. The covariances are estimated on a rolling walk-forward basis with $T=120$ non-overlapping monthly returns. All measures are computed with 579 monthly portfolio weights from 09-09-1974 until 10-21-2022.}
    \centering
    \footnotesize
\begin{tabular}{lcccccc} \toprule
 & FF25 & FF100 & \makecell{ME/OP \\ 25} & \makecell{ME/OP \\ 100} & \makecell{ME/INV \\ 25} & \makecell{ME/INV \\ 100}\\ \cmidrule{1-7}
\multicolumn{7}{l}{Panel A. Average Returns (\%)} \\\cmidrule{1-7}
CP2 & $2.925$ & $1.807$ & $2.308$ & $1.027$ & $2.274$ & $1.551$\\ 
PP2 & $4.160$ & $4.143$ & $3.779$ & $4.082$ & $2.893$ & $2.522$\\ 
MVO & $0.761$ & $0.581$ & $0.595$ & $-0.032$ & $0.746$ & $0.386$\\ 
UNI & $2.882$ & $2.316$ & $2.743$ & $2.047$ & $2.431$ & $1.838$\\ 
\cmidrule{1-7}
\multicolumn{7}{l}{Panel B. Standard Deviation of Returns (\%)} \\\cmidrule{1-7}
CP2 & $2.343$ & $1.384$ & $2.153$ & $1.312$ & $2.072$ & $1.293$\\ 
PP2 & $6.042$ & $6.053$ & $5.288$ & $5.488$ & $4.707$ & $4.760$\\ 
MVO & $1.893$ & $1.431$ & $1.807$ & $1.332$ & $1.642$ & $1.347$\\ 
UNI & $5.489$ & $5.097$ & $4.543$ & $4.095$ & $4.619$ & $4.280$\\ 
\cmidrule{1-7}
\multicolumn{7}{l}{Panel C. Sharpe Ratio} \\\cmidrule{1-7}
CP2 & $\underset{(8.365)}{ 1.244}$ & $\underset{(8.735)}{ 1.302}$ & $\underset{(7.266)}{ 1.072}$ & $\underset{(5.387)}{ 0.786}$ & $\underset{(7.426)}{ 1.096}$ & $\underset{(8.073)}{ 1.197}$\\ 
PP2 & $\underset{(4.849)}{ 0.706}$ & $\underset{(4.823)}{ 0.702}$ & $\underset{(5.005)}{ 0.729}$ & $\underset{(5.197)}{ 0.758}$ & $\underset{(4.338)}{ 0.630}$ & $\underset{(3.777)}{ 0.548}$\\ 
MVO & $\underset{(2.836)}{ 0.410}$ & $\underset{(2.851)}{ 0.412}$ & $\underset{(2.336)}{ 0.337}$ & $\underset{(-0.120)}{-0.017}$ & $\underset{(3.184)}{ 0.461}$ & $\underset{(2.028)}{ 0.293}$\\ 
UNI & $\underset{(3.766)}{ 0.546}$ & $\underset{(3.285)}{ 0.476}$ & $\underset{(4.263)}{ 0.619}$ & $\underset{(3.561)}{ 0.516}$ & $\underset{(3.751)}{ 0.544}$ & $\underset{(3.093)}{ 0.448}$\\ 
\cmidrule{1-7}
\multicolumn{7}{l}{Panel D. Alpha (\%)} \\\cmidrule{1-7}
CP2 & $2.287$ & $1.576$ & $1.508$ & $0.781$ & $1.641$ & $1.300$\\ 
PP2 & $2.048$ & $2.097$ & $1.691$ & $2.263$ & $1.318$ & $1.205$\\ 
MVO & $0.025$ & $0.185$ & $-0.195$ & $-0.391$ & $0.117$ & $-0.003$\\ 
\cmidrule{1-7}
\multicolumn{7}{l}{Panel E. Beta to Simple Factor} \\\cmidrule{1-7}
CP2 & $0.206$ & $0.100$ & $0.277$ & $0.106$ & $0.202$ & $0.115$\\ 
PP2 & $0.735$ & $0.866$ & $0.783$ & $0.922$ & $0.600$ & $0.716$\\ 
MVO & $0.254$ & $0.174$ & $0.291$ & $0.171$ & $0.257$ & $0.198$\\ 
\cmidrule{1-7}
\multicolumn{7}{l}{Panel F. Idiosyncratic Volatility (\%)} \\\cmidrule{1-7}
CP2 & $2.040$ & $1.285$ & $1.745$ & $1.240$ & $1.856$ & $1.198$\\ 
PP2 & $4.494$ & $4.155$ & $3.909$ & $3.986$ & $3.858$ & $3.673$\\ 
MVO & $1.288$ & $1.124$ & $1.235$ & $1.142$ & $1.136$ & $1.058$\\ 
\cmidrule{1-7}
\multicolumn{7}{l}{Panel G. Information Ratio} \\\cmidrule{1-7}
CP2 & $\underset{(7.584)}{ 1.121}$ & $\underset{(8.254)}{ 1.226}$ & $\underset{(5.906)}{ 0.864}$ & $\underset{(4.336)}{ 0.630}$ & $\underset{(6.037)}{ 0.884}$ & $\underset{(7.349)}{ 1.084}$\\ 
PP2 & $\underset{(3.149)}{ 0.456}$ & $\underset{(3.485)}{ 0.505}$ & $\underset{(2.990)}{ 0.433}$ & $\underset{(3.914)}{ 0.568}$ & $\underset{(2.364)}{ 0.342}$ & $\underset{(2.272)}{ 0.328}$\\ 
MVO & $\underset{(0.133)}{ 0.019}$ & $\underset{(1.139)}{ 0.164}$ & $\underset{(-1.095)}{-0.158}$ & $\underset{(-2.373)}{-0.343}$ & $\underset{(0.715)}{ 0.103}$ & $\underset{(-0.019)}{-0.003}$\\ 
\bottomrule \end{tabular}
\label{tab:twosig}
\end{table}
\vspace{\fill}

\newpage

\vspace*{\fill}
\begin{table}[H]
    \caption{Weight Statistics for Two Signals Case}
    \captionsetup{font=footnotesize}
    \caption*{This table shows the average monthly weight statistics of various portfolio selection methods estimated using a 21-day momentum signal. Panels A--F display the results for the turnover, the proportional leverage, the sum of negative weights in the portfolio, and the minimum and maximum portfolio weights. The test assets are the 25 and 100 equity portfolios independently sorted on size and three other factors (book-to-market, operating profitability, and investment). The covariances are estimated on a rolling walk-forward basis with $T=120$ non-overlapping monthly returns. All measures are computed with 579 monthly portfolio weights from 09-09-1974 until 10-21-2022. }
    % \vspace*{5mm}
    \centering
    \footnotesize
\begin{tabular}{lcccccc} \toprule
 & FF25 & FF100 & \makecell{ME/OP \\ 25} & \makecell{ME/OP \\ 100} & \makecell{ME/INV \\ 25} & \makecell{ME/INV \\ 100}\\ \cmidrule{1-7}
\multicolumn{7}{l}{Panel A. Turnover} \\\cmidrule{1-7}
CP2 & $0.749$ & $0.437$ & $0.806$ & $0.542$ & $0.762$ & $0.466$\\ 
PP2 & $0.806$ & $0.763$ & $0.878$ & $0.794$ & $0.891$ & $0.856$\\ 
MVO & $1.064$ & $1.075$ & $1.065$ & $1.082$ & $1.060$ & $1.079$\\ 
UNI & $0.959$ & $0.978$ & $0.970$ & $0.990$ & $0.969$ & $0.989$\\ 
\cmidrule{1-7}
\multicolumn{7}{l}{Panel B. Proportional Leverage} \\\cmidrule{1-7}
CP2 & $0.477$ & $0.473$ & $0.480$ & $0.488$ & $0.470$ & $0.486$\\ 
PP2 & $0.478$ & $0.483$ & $0.472$ & $0.473$ & $0.477$ & $0.484$\\ 
MVO & $0.506$ & $0.499$ & $0.509$ & $0.503$ & $0.497$ & $0.497$\\ 
UNI & $0.485$ & $0.498$ & $0.488$ & $0.497$ & $0.485$ & $0.498$\\ 
\cmidrule{1-7}
\multicolumn{7}{l}{Panel C. Sum of Negative Weights} \\\cmidrule{1-7}
CP2 & $-0.483$ & $-0.480$ & $-0.486$ & $-0.482$ & $-0.481$ & $-0.479$\\ 
PP2 & $-0.500$ & $-0.500$ & $-0.500$ & $-0.500$ & $-0.500$ & $-0.500$\\ 
MVO & $-0.499$ & $-0.499$ & $-0.500$ & $-0.501$ & $-0.498$ & $-0.498$\\ 
UNI & $-0.500$ & $-0.500$ & $-0.500$ & $-0.500$ & $-0.500$ & $-0.500$\\ 
\cmidrule{1-7}
\multicolumn{7}{l}{Panel D. Minimum Weight} \\\cmidrule{1-7}
CP2 & $-0.121$ & $-0.035$ & $-0.113$ & $-0.036$ & $-0.114$ & $-0.035$\\ 
PP2 & $-0.104$ & $-0.034$ & $-0.108$ & $-0.037$ & $-0.104$ & $-0.033$\\ 
MVO & $-0.099$ & $-0.031$ & $-0.096$ & $-0.030$ & $-0.097$ & $-0.031$\\ 
UNI & $-0.090$ & $-0.025$ & $-0.089$ & $-0.025$ & $-0.089$ & $-0.024$\\ 
\cmidrule{1-7}
\multicolumn{7}{l}{Panel E. Maximum Weight} \\\cmidrule{1-7}
CP2 & $0.109$ & $0.032$ & $0.102$ & $0.032$ & $0.102$ & $0.032$\\ 
PP2 & $0.092$ & $0.031$ & $0.096$ & $0.033$ & $0.096$ & $0.030$\\ 
MVO & $0.102$ & $0.031$ & $0.099$ & $0.031$ & $0.097$ & $0.030$\\ 
UNI & $0.088$ & $0.025$ & $0.087$ & $0.025$ & $0.089$ & $0.025$\\ 
\cmidrule{1-7}
\multicolumn{7}{l}{Panel F. Standard Deviation of Weights} \\\cmidrule{1-7}
CP2 & $0.053$ & $0.013$ & $0.052$ & $0.013$ & $0.052$ & $0.013$\\ 
PP2 & $0.051$ & $0.013$ & $0.052$ & $0.013$ & $0.051$ & $0.013$\\ 
MVO & $0.051$ & $0.013$ & $0.050$ & $0.012$ & $0.050$ & $0.012$\\ 
UNI & $0.049$ & $0.012$ & $0.049$ & $0.012$ & $0.049$ & $0.012$\\ 
\bottomrule \end{tabular}
\label{tab:twosigswgtstat}
\end{table}
\vspace{\fill}

\newpage

\newpage
\section{ Proofs for all the Propositions}\label{appendix:proofofprop}
\begin{proof}[Proof of Proposition \ref{prop:fullobj}] 
We provide the proof of this proposition in two parts. 

\noindent \textbf{Part 1}. We drop the time subscript $t$ for brevity. From the cyclic property of the trace operator and linearity of the expectation, the expected value of the portfolio returns is 
\begin{eqnarray} \label{eqn:mu}
\mathbb{E}[x' A r] 
= \mathbb{E}[ \mathsf{Tr}(A rx')]
= \mathsf{Tr}(A\Sigma_{rx}).
\end{eqnarray}
In order to analyze the variance of the portfolio returns, we appeal to the following one-dimensional from \cite{isserlis1918formula} or \cite{wick1950evaluation}, which expresses the higher moments of centered multivariate Gaussian vectors in terms of its second-moments. 
\begin{theorem}\label{thm:wick}
Let $z_1,z_2,z_3$, and $z_4$ be jointly Gaussian random variables with mean zero. Then we have the following results:
\begin{equation}
\begin{aligned}\label{eqn:wicks}
\mathbb{E}[z_1] &=0 \hspace{70pt}\\
\mathbb{E}[z_1z_2] &= \mathsf{Cov}(z_1, z_2) &\hspace{70pt} \\
\mathbb{E}[z_1z_2z_3] &=0 \hspace{70pt}\\
\mathbb{E}[z_1z_2 z_3 z_4] &= \mathbb{E}[z_1z_2] \mathbb{E}[z_3z_4] + \mathbb{E}[z_1z_3] \mathbb{E}[z_2z_4] + \mathbb{E}[z_1z_4]\mathbb{E}[z_2z_3] .
\end{aligned}
\end{equation}
\end{theorem}
The last equation of \eqref{eqn:wicks} takes all partitions of size two of the four variables, which gives us three separate terms. By recasting fourth-order terms using the covariance terms, we can express the variance as
% https://math.stackexchange.com/questions/4182166/expectation-of-bilinear-form-squared-with-probabilistic-matrix
%
\begin{equation}
\begin{aligned}
\mathsf{Var}[x' A r] &= \mathbb{E}[(x' A r)^2] - \mathbb{E}[x' A r]^2 \\
&= \sum_{i,j,k,l}\mathbb{E}[A_{i, j} A_{k,l} x_{i} r_{j}x_{k}  r_{l} ] -
\sum_{i,j,k,l} \mathbb{E}[ A_{i, j} x_{i} r_{j}]\mathbb{E}[A_{k,l} x_{k}  r_{l}]\\
&= \sum_{i,j,k,l}A_{i, j}A_{k,l} \mathbb{E}[x_{i} r_{j}] \mathbb{E}[x_{k}  r_{l} ] 
 + \sum_{i,j,k,l}A_{i, j}A_{k,l} \mathbb{E}[x_{i} x_{k} ] \mathbb{E}[ r_{j} r_{l} ] \\
&\qquad + \sum_{i,j,k,l}A_{i, j}A_{k,l} \mathbb{E}[x_{i} r_{l} ] \mathbb{E}[ x_{k} r_{j} ]  -
\sum_{i,j,k,l}  A_{i, j} A_{k,l} \mathbb{E}[x_{i} r_{j}] \mathbb{E}[x_{k}  r_{l}]\\
&= 
 \sum_{i,j,k,l}A_{i, j}A_{k,l} \mathbb{E}[x_{i} x_{k} ] \mathbb{E}[ r_{j} r_{l} ] +
\sum_{i,j,k,l}A_{i, j}A_{k,l} \mathbb{E}[x_{i} r_{l} ] \mathbb{E}[ x_{k} r_{j} ] \\
&= 
 \sum_{i,j,k,l} \mathbb{E}[x_{k} x_{i} ]A_{i, j} \mathbb{E}[ r_{j} r_{l} ]A_{k,l} +
\sum_{i,j,k,l} \mathbb{E}[r_{l} x_{i}  ]A_{i, j} \mathbb{E}[ r_{j}x_{k}   ] A_{k,l} ,
\end{aligned}
\end{equation}
where $A_{i, j}$ is the $(i, j)$-th entry of $A$. Reverting back to matrix notation, we have
\begin{align} \label{eqn:sigma}
\mathsf{Var}[x' A r]
&= \mathsf{Tr}(\Sigma_{x} A \Sigma_{r} A') + \mathsf{Tr}(\Sigma_{rx} A \Sigma_{rx} A).
\end{align}
\noindent \textbf{Part 2}. In order to solve the full mean-variance objective function \eqref{eq:obj2}, we use the tools from CCA that we have laid the ground in Section \ref{sec:cca}. To assist in our derivation, we expand our notation and let $Q_r$ and $Q_x$ be matrices whose columns contain the so-called $i$th canonical directions $ q_{r,i}\coloneqq\Sigma_r ^{-1/2} u_i$ and $q_{x,i}\coloneqq\Sigma_x ^{-1/2} v_i$, respectively, and $D \coloneqq \mathsf{Diag} ( s_1, s_2,\ldots, s_N)$ be a diagonal matrix containing the canonical correlations.  We now state the main theorem:
\begin{theorem}\label{thm:main}
Suppose $r_{t+1}$ and $x_t$ are $N$-dimensional jointly Gaussian random variables. Then the objective function \eqref{eq:cond} is maximized at 
\begin{align}\label{eq:optAfull}
    A = Q_x  \mathsf{Diag}( L_1 , \ldots , L_N ) Q_r',
\end{align}
where the diagonal matrix consists of the following elements
\begin{align}\label{eq:sol}
    L_{i}= \frac{1}{\gamma} \frac{s_i}{1+s_i^2}, 
\end{align}
for $i=1,\ldots , N$. 
\end{theorem}
Without loss of generality, let us reparameterize the matrix of coefficients in terms of the canonical directions as $ A = Q_x L Q_r '$ where $L$ is a variable matrix that we now have to optimize on. The expected value of the portfolio returns at time $t$ can be written as
\begin{align}
\mathbb{E}[x'_t A r_{t+1}]  &= \mathsf{Tr}( A \Sigma_{rx} ) \\
&= \mathsf{Tr}(L D),
\end{align}
and the variance as
\begin{align}
\mathsf{Var}[x_t ' A r_{t+1}] &=  \mathsf{Tr}(\Sigma_x A \Sigma_r  A') + \mathsf{Tr}(\Sigma_{rx}  A \Sigma_{rx}  A) \\
&= \mathsf{Tr}(L L ') + \mathsf{Tr}(DL D L) .
\end{align}
Putting all together, the objective function \eqref{eq:obj2} can be written as:
\begin{align}
    \max _{L}   \mathsf{Tr}(L D ) - \frac{\gamma}{2} ( \mathsf{Tr}(L L' )  - \mathsf{Tr}(D  L D  L )   ).
\end{align}

 Using the rules for matrix derivatives from \cite{lutkepohl1997handbook}, we take first-order conditions with respect to the matrix $L$ to get
\begin{align}\label{eq:foc}
   D - \gamma( L +  D L' D )= 0.
\end{align}
Since the matrix $D$ is diagonal and thus symmetric, we necessarily have the following relationship
\begin{align}
    \frac{1}{\gamma} D =  L + DL'D = L' + D L D .
\end{align}
Let $L^a \coloneqq L-L'$, which is an anti-symmetric matrix (that is, $L^{a'}=- L^a$) and has diagonal elements zero. Then we have 
\begin{align}\label{eq:as}
    L^a = D L^a D.
\end{align} 
If we restrict our attention to the off-diagonal elements of \eqref{eq:as}, we see that $L^A_{i, j} = s_i s_j L^A_{i, j}$ for all $i \neq j$, so $(1-s_i s_j) L^A_{i, j}=0$. Therefore, either $s_i s_j = 1$ for all $i \neq j$ or $L^A_{i, j} =0$. Consequently, $L^a$ is identically zero and hence, $L$ must be a symmetric matrix.

The benefit of variable matrix $L$ being symmetric is that the condition \eqref{eq:foc} can now be written as
\begin{align}
    L+ D L D = \frac{1}{\gamma} D.
\end{align} 
If we focus on off-diagonal elements, we see that $(1 + s_i s_j)L_{i,j} =0$ for all $i\neq j$. But we also know that the canonical correlations satisfy the following ordering $0 \leq s_1 \leq  \ldots \leq s_N \leq 1$, and so it is impossible that $s_i s_j = -1$. Thus, $L_{i,j} =0$ for all $i \neq j$, and so $L$ must be a diagonal matrix. Thus, optimizing over the elements $L_{ii} \equiv L_i$ for $i = 1,\ldots ,N$, boils down to maximizing the following univariate problems
\begin{align}
    \max _{L_1 , \ldots , L_N} \sum_{i=1}^N L_i s_i - \frac{\gamma}{2} ( L_i^2 + s_i^2 L_i^2 )  .
\end{align}
Our problem can now be easily solved to give us
\begin{align}
    L_i = \frac{1}{\gamma} \frac{s_i}{1+s_i^2} .
\end{align}
Hence, the diagonal elements are a nonlinear function of the canonical correlations. The optimal matrix $A$ in \eqref{eq:optAfull} is effectively the optimal scaling of the canonical portfolios of the asset returns and signals. 
\end{proof}

\begin{proof}[Proof of Proposition \ref{prop:partialobj}]
Applying a change-of-variables $\tilde{A} \coloneqq \Sigma^{1/2}_x A \Sigma^{1/2}_r$, we have
\begin{align}\label{eq:obj200}
    \max _{\tilde{A}}  \mathsf{Tr}(\tilde{A}  \Sigma_{\tilde{r} \tilde{x}} '  ) -\frac{\gamma}{2} \mathsf{Tr}( \tilde{A} \tilde{A}' )   .
\end{align}
The problem \eqref{eq:obj200} can be solved to yield 
\begin{align}
    \tilde{A} =  \frac{1}{\gamma} \Sigma^{-1/2}_{x} \Sigma_{rx} ' \Sigma^{-1/2}_{r} .
\end{align} 
Since the solution is expressed in a different basis, we rescale it back to the original assets through the following operation
\begin{align} \label{eq:optAhat}
    A = \Sigma^{-1/2}_{x}  \tilde{A} \Sigma^{-1/2}_{r} = \frac{1}{\gamma} \Sigma^{-1}_{x}  \Sigma_{rx} ' \Sigma^{-1}_{r} .
\end{align}
As a result of this approximation, we now have 
\begin{align} \label{eq:approx}
    L_i \approx \frac{1}{\gamma}  s_i.
\end{align}
Figure \ref{fig:approx} contrasts the optimal singular value adjustment in \eqref{eq:sol} to the linear approximation \eqref{eq:approx} for $\gamma = 1$. We can see that for small values of $s_i$,  the optimal adjustment is approximately linear while for large values of $s_i$, the optimal adjustment downweighs the value of the canonical correlation.  
\end{proof}

\begin{proof}[Proof of Proposition \ref{prop:fullyinvested}]
% https://quant.stackexchange.com/questions/59202/derivation-of-mean-variance-portfolio-weights-as-closed-form-analytical-solution
We start by writing down the expression for the Lagrangian 
\begin{align}
    \mathsf{Tr}( A \Sigma_{rx} ) - \frac{\gamma}{2} \mathsf{Tr}(\Sigma_x A \Sigma_r  A')  + \lambda (1 - \mathbf{1}'A' x_t) ,
\end{align}
$\lambda >0$ is the Lagrange multiplier. Using the change-of-variables $\tilde{A} \coloneqq \Sigma^{1/2}_x A \Sigma^{1/2}_r$, we have
\begin{align}
    \mathsf{Tr}( \tilde{A} \Sigma_{\tilde{r}\tilde{x}} ) - \frac{\gamma}{2} \mathsf{Tr}(\tilde{A} \tilde{A}')  + \lambda (1 - \mathbf{1}' \Sigma_r^{-1/2}\tilde{A}' \Sigma_x^{-1/2} x_t) .
\end{align}
Performing the first-order conditions 
\begin{align}
    [\text{FOC } \tilde{A}]: & \enskip \Sigma_{\tilde{r}\tilde{x}}' 
 - \gamma \tilde{A} - \lambda \Sigma_x ^{-1/2} x_t \mathbf{1}' \Sigma_r ^{-1/2} = 0, \\
    [\text{FOC } \lambda]: & \enskip 1 - \mathbf{1}' w_t = 0. \label{eq:lambda2}
\end{align}
By reverting the change-of-variables we have made and using the fact that the portfolio policies are linear the signals, the portfolio weights is then
\begin{align}\label{eq:tmp}
    w_t^{\mathrm{FI}} &= \frac{1}{\gamma} \Sigma_r ^{-1} \Sigma_{rx} \Sigma_x ^{-1} x_t + \lambda (x_t ' \Sigma ^{-1}_x x_t) \mathbf{1}' \Sigma_r ^{-1} .
\end{align}
Solving for $\lambda$ gives
\begin{align}
    \lambda  = \frac{1 - \kappa}{(x_t ' \Sigma ^{-1}_x x_t) (\mathbf{1}' \Sigma_r ^{-1} \mathbf{1})} ,
\end{align}
where $\kappa \coloneqq \gamma ^{-1}\mathbf{1}' \Sigma_r ^{-1} \Sigma_{rx} \Sigma_x ^{-1} x_t$. Inserting the $\lambda$ into \eqref{eq:tmp} and rearranging the expression gives us the optimal portfolio \eqref{eq:finvested}.
\end{proof}
\begin{proof}[Proof of Proposition \ref{prop:return2singvals}]
If $r_{t+1}$ and $x_t$ are $N$-dimensional zero-mean Gaussian random variables, then their projections onto the  $i$th canonical portfolios given by $u_i ' \tilde{r}_{t+1}$ and  $v_i ' \tilde{x}_t $ are also Gaussian random variables with mean zero and unit variances and a joint correlation $s_i$. For the expected return of the $i$th canonical portfolio, we have 
\begin{align}
    \mathbb{E}[\pi _{t,i}] = s_i v_i ' \Sigma_{\tilde{r}\tilde{x}}' u_i = s_i ^2.
\end{align} 
Furthermore, the variance of the $i$th canonical portfolio is
\begin{align}
    \mathsf{Var}[\pi _{t,i}] &= \mathbb{E}[s_i ^2 (v_i ' \tilde{x}_t )^2 (u_i' \tilde{r}_{t+1} )^2] - s_i^4 \\
    &= s_i ^2 (1 + 2s_i ^2) - s_i^4 \\
    &= s_i^2(1+s_i ^2).
\end{align}
The second equality follows from the second-moment-based result of the product of two correlated Gaussian variables from \citet[Section 6]{haldane1942moments}. Substituting the expression of the optimal policy $A$ into the definition of the expected value and variance of the portfolio returns, we have:
\begin{align}
   \mathbb{E}[x_t'Ar_{t+1}] = \frac{1}{\gamma}\mathsf{Tr}(\Sigma_x ^{-1}\Sigma_{rx}'\Sigma_r ^{-1}\Sigma_{rx}) &= \frac{1}{\gamma}\sum ^N_{i=1} s_i^2 \\
   \mathbb{E}[x_t ' A \Sigma _r A' x_{t}] = \frac{1}{\gamma^2}\mathsf{Tr}(\Sigma_x ^{-1}\Sigma_{rx}'\Sigma_r ^{-1}\Sigma_{rx}) &= \frac{1}{\gamma^2}\sum ^N_{i=1} s_i^2. 
\end{align}
The proof concludes.
\end{proof}

% \begin{proof}[Proof of Proposition \ref{prop:bias}]
% % https://math.stackexchange.com/questions/3322556/jensens-inequality-over-definite-positive-matrices
    
% \end{proof}

\begin{proof}[Proof of Corollary \ref{corollary:sharpe}]
    This follows immediately Proposition \ref{prop:return2singvals} and the definition of the Sharpe ratio.
\end{proof}

\begin{proof}[Proof of Proposition \ref{prop:staticdyn}]
Substituting the second-moment matrix of managed portfolio returns \eqref{eq:sdbets} into the policy matrix $A$ from \eqref{eq:optA}, the expected returns based on this matrix of coefficients is then
    \begin{align}
         \mathbb{E}[x_t ' A r_{t+1}] = \frac{1}{\gamma} \mathsf{Tr}(\Sigma_x ^{-1}\Sigma_{rx}'\Sigma_r ^{-1}\Sigma_{rx}) = \frac{1}{\gamma}(\mu _r ' \Sigma _r^{-1} \mu _r)(\mu _x ' \Sigma _x^{-1} \mu _x) + \frac{1}{\gamma}\sum ^N_{i=1} s_i ^2,
    \end{align}
    and the claim follows.
\end{proof}

\end{document}